\documentclass[11pt,letterpaper]{article}
\usepackage{verbatim}
\usepackage{titling}
\usepackage{cancel}
\usepackage{enumitem} 
\setlist{noitemsep,topsep=0pt,parsep=0pt} 
\usepackage{subfig}
\usepackage{mathtools} 
\usepackage{multirow}
\usepackage[all]{xy}
\usepackage{tikz}
\usetikzlibrary{arrows,backgrounds,calc,fit,decorations.pathreplacing,decorations.markings,shapes.geometric}

\tikzset{every fit/.append style=text badly centered}

\usepackage{notation}
\numberwithin{equation}{section}
\usepackage[bookmarks=true,hypertexnames=false]{hyperref}
\hypersetup{colorlinks=true, citecolor=blue, linkcolor=red, urlcolor=blue}
\makeatletter

\newcommand{\Rmnum}[1]{\expandafter\@slowromancap\romannumeral #1@}
\makeatother

\newcommand{\ii}{\mathfrak{i}}

\newcommand{\Holant}{\operatorname{Holant}}

\newcommand{\holant}[2]{\ensuremath{\Holant(#1\mid #2)}}

\newcommand{\CSP}{\operatorname{\#CSP}}

\newenvironment{remark}{\medskip{\bfseries \noindent Remark:}}{\par\medskip}{\par\medskip}

\tikzstyle{internal} = [draw, fill, shape=circle]
\tikzstyle{external} = [shape=circle]
\tikzstyle{square}   = [draw, fill, rectangle]
\tikzstyle{triangle} = [draw, fill, regular polygon, regular polygon sides=3, inner sep=3pt]
\tikzstyle{pentagon} = [draw, fill, regular polygon, regular polygon sides=5, inner sep=2pt, minimum size=14pt]

\begin{document}

\title{{\bf From Holant 
to Quantum Entanglement 
and Back}}

\vspace{0.3in}

\author{Jin-Yi Cai\thanks{Department of Computer Sciences, University of Wisconsin-Madison. Supported by NSF  CCF-1714275.
 } \\ {\tt jyc@cs.wisc.edu}
\and Zhiguo Fu\thanks{School of Information Science and Technology and KLAS, Northeast Normal University, China.
Supported by NSFC-61872076.
}\\ \tt
fuzg432@nenu.edu.cn
\and
Shuai Shao$^\ast$ \\ \tt sh@cs.wisc.edu}

\date{}
\maketitle
\thispagestyle{empty}
\bibliographystyle{plain}

\begin{abstract}
Holant problems are intimately connected with quantum  theory as tensor networks. We first use  techniques from Holant theory to derive 
new and improved results for quantum entanglement theory. We discover two particular entangled states   $|{\Psi_6}\rangle$ of 6 qubits  and $|{\Psi_8}\rangle$ of 8 qubits respectively, that have 
extraordinary  and unique closure properties in terms of the Bell property.
Then we use entanglement properties of constraint functions  to derive a new complexity dichotomy for all real-valued Holant problems containing an odd-arity signature.
The signatures need not be symmetric, and no auxiliary signatures are assumed.
 \end{abstract}

\newpage
\setcounter{page}{1}
\section{Introduction}
\subsection{Holant problems}
Holant problems are a broad class of Sum-of-Products. 
It generalizes other frameworks such as counting constraint satisfaction problems (\#CSP) and counting graph homomorphisms (\#GH). 
Both have been well studied 
and full complexity dichotomies have been established 
\cite{bulatov-ccsp, dyer-richerby, bulatov2012csp, cai-chen-lu-nonnegative-csp, cai-chen-csp, dyer-green-gh, bulatov2005gh, goldberg2010gh, cai-chen-lu-gh}.
On the other hand, the understanding of Holant problems, even  restricted to
the Boolean domain, is still limited. In this paper, we focus on Holant problems defined on the Boolean domain.

Holant problems are parameterized by a set of constraint functions, also called signatures. A signature (over the Boolean domain) of arity $n>0$ is a map $\mathbb{Z}_2^{n} \rightarrow \mathbb{C}$. Let $\mathcal{F}$ be any fixed set of signatures. 
A signature grid
$\Omega=(G, \pi)$ over $\mathcal{F}$
 is a tuple, where $G = (V,E)$
is a graph without isolated vertices,
 $\pi$ labels each $v\in V$ with a signature
$f_v\in\mathcal{F}$ of arity ${\operatorname{deg}(v)}$,
and labels the incident edges
$E(v)$ at $v$ with input variables of $f_v$.
We consider all 0-1 edge assignments $\sigma$, and
each gives an evaluation
$\prod_{v\in V}f_v(\sigma|_{E(v)})$, where $\sigma|_{E(v)}$
denotes the restriction of $\sigma$ to $E(v)$.

\begin{definition}[Holant problems]
The input to the problem {\rm Holant}$(\mathcal{F})$
is a signature grid $\Omega=(G, \pi)$ over $\mathcal{F}$.
The output is  the partition function
\begin{equation*}
\Holant_{\Omega} =\sum\limits_{\sigma:E(G)\rightarrow\{0, 1\}}\prod_{v\in V(G)}f_v(\sigma|_{E_{(v)}}).    
\end{equation*}
Bipartite Holant problems $\holant{\mathcal{F}}{\mathcal{G}}$
are {Holant} problems
 over bipartite graphs $H = (U,V,E)$,
where each vertex in $U$ or $V$ is labeled by a signature in $\mathcal{F}$ or $\mathcal{G}$ respectively.
\end{definition}

Weighted \#CSP  is a special class of  Holant problems.
So are all weighted \#GH.
Other problems  expressible as Holant problems include counting matchings and perfect matchings~\cite{cai-artem}, counting Eulerian orientations~\cite{cai-fu-shao-eo}, computing the partition functions of six-vertex models~\cite{cai-fu-xia-six} and eight-vertex models~\cite{cai-fu-eight}. 
It is proved that counting perfect matchings cannot be expressed by \#GH \cite{Freedman-et-al,cai-artem}. 
Thus, Holant problems are provably more expressive. 

Progress has been made in  the complexity classification of Holant problems.
When all signatures are restricted to be symmetric, a full dichotomy is proved \cite{cai-guo-tyson-vanishing}. 
When asymmetric signatures are allowed, some dichotomies are proved for special families of Holant problems 
by assuming that certain auxiliary  signatures are available, e.g.,  $\Holant^\ast$, $\Holant^+$ and $\Holant^c$~\cite{cai-lu-xia-holant*, Backens-holant-plus, CLX-HOLANTC, Backens-Holant-c}.
Without assuming auxiliary signatures a Holant dichotomy is established only  for non-negative real valued signatures~\cite{wang-lin}. 
\subsection{Quantum entanglement theory}
Holant problems are in fact synonymous with tensor networks in quantum theory.  The partition function $\Holant_{\Omega}$ can be used in
a (strong) simulation of quantum circuits \cite{valiant-holo-02}.
A signature grid is just a tensor network, where each signature is a tensor with its inputs associated with its incident edges. In this sense, a signature of arity $n$ represents a state of \emph{n qubits}.
In quantum theory, the basic component of a system is a qubit. The (pure) state
of an $n$ qubits $|\Psi\rangle$ is described by a vector in $\mathbb{C}^{2^n}$. 
(The standard notion requires quantum states to have norm 1, but in this paper, normalization by a nonzero scalar makes no difference for complexity, so we work with states having arbitrary nonzero norms.)
A nonzero $n$-ary signature $f$ is synonymous with an  $n$-qubit state $|f\rangle=\sum_{x\in\{0,1\}^n}f(x)|x\rangle$. 
In this paper, we use them interchangeably.
When $f$ is a zero signature (i.e., $f\equiv 0$), we agree that $|f\rangle$ is a null state, denoted by $\mathfrak{N}$.

 A core concept in quantum theory is \emph{entanglement}. It is perhaps the most distinguishing characteristic feature separating  quantum  and classical physics.
 
 \begin{definition}[Quantum entanglement]
 A state of $n$ qubits  ($n >1$, representing a multiple system) is \emph{entangled} if it cannot be decomposed as a tensor product of single-qubit states (individual systems). It is \emph{genuinely entangled}
if it cannot be decomposed as a tensor product of states of proper subsystems. It exhibits \emph{multipartite entanglement} if it involves a genuinely entangled state of subsystem of more than two qubits (i.e., it cannot be decomposed as a tensor product of single-qubit states and 2-qubit states).
\end{definition}

Today, entanglement is recognized as an important resource in quantum computing and quantum information theory. It has been shown that quantum computing speedups essentially depend on
unbounded
entanglement \cite{jl}. While in quantum information theory, an entangled state is shared by several parties, one can perform operations on a subsyetem locally without access to the other subsystems. This set-up is commonly used in quantum teleportation and quantum key distribution \cite{ ekert, bennett-tele}.
 For different information-theoretic tasks, different types of entanglement can be used \cite{michael-book}. The classification of them under \emph{stochastic local operation with classical communication}  (SLOCC) equivalence was  proposed in 2000 by D\"{u}r et al. \cite{slocc2}, and
 is an area of active 
 research \cite{4-qubit-v, 2-2-n,  3-qubit-l, 4-qubit-l, Backens-class, 4-qubit-g}.
 Yet so far, even the classification of  entangled 4-qubit states is 
 not completely settled.
 For more about quantum entanglement theory, we refer to the survey \cite{surver-entangle}.
 
 \subsection{Existing dichotomies inspired by entanglement theory}

There are many natural connections between Holant problems and quantum theory. 
The introduction of Holant problems is inspired by holographic transformations~\cite{ valiant-holo-08}. Such a holographic transformation applied separately on each qubit $i$ with a matrix $A_i$  is just a SLOCC in quantum theory.
Also, many known P-time computable
signature sets for Holant problems can be clearly described in the quantum literature \cite{ clfford-gate-1, Backens-holant-plus}
and they correspond directly to sets of states that are of independent interest in quantum theory  \cite{clfford-gate-2, Daniel-Gottesman}.


Going beyond that,  Backens recently applied knowledge from the theory of quantum entanglement, directly to the study of Holant problems and derived new dichotomy results \cite{Backens-holant-plus, Backens-Holant-c}. 
We give a short description for these results in this subsection.
We use $\langle\Phi|$ to denote the Hermitian adjoint  (complex conjugate) of $|\Phi\rangle$, and $\langle\Phi|\Psi\rangle$ to denote the (complex)  inner product of two $n$-qubit states.


\begin{definition}[Projection]
The projection of the $i$-th qubit of an $n$-qubit $(n\geqslant 2)$ state $|{\Psi}\rangle$ onto a single-qubit state $|\theta\rangle=a|0\rangle+b|1\rangle$ is defined as 
$\langle\theta|_i|\Psi\rangle=\bar{a}|{\Psi}_i^0\rangle+\bar{b}|{\Psi_i^1}\rangle$ where 
$\bar{a}$ and $\bar b$ are complex conjugates of $a$ and $b$, and  
$|{\Psi}_i^0\rangle$ and $|{\Psi_i^1}\rangle$ are states of the remaining
$n-1$ qubits 
when  the i-th qubit of   $|{\Psi}\rangle$ is set to  $0$ and $1$ respectively.
\end{definition}

\begin{theorem}
[\cite{2-system-1},\cite{2-system-2}]\label{thm:popescu-rohrlich}
 Let $|{\Psi}\rangle$ be a genuinely entangled $n$-qubit $(n\geqslant 3)$ state. For any two qubits of $|{\Psi}\rangle$, there exist  projections  of the other $n-2$ qubits onto $n-2$ many single-qubit states that result in an entangled 2-qubit state.
\end{theorem}

This result was presented to show that any pure entangled multipartite quantum state violates some Bell's inequality \cite{2-system-1}.
The original proof \cite{2-system-1}
was flawed  and was corrected recently \cite{2-system-2}. 
Theorem \ref{thm:popescu-rohrlich} shows that two particle entanglement can be realized via 
performing 
local projections on a multiparticle state.
It is observed in \cite{Backens-holant-plus} that the theorem  holds
even when restricted to only local projections onto computational or Hadamard basis states, i.e., $|{0}\rangle$, $|{1}\rangle$, $|{+}\rangle=|{0}\rangle+|{1}\rangle$ and $|{-}\rangle=|{0}\rangle-|{1}\rangle$.



Based on Theorem \ref{thm:popescu-rohrlich} and the inductive entanglement classification under SLOCC equivalence \cite{3-qubit-l, 4-qubit-l, Backens-class}, Backens 
showed that 
beyond
entangled 2-qubit  states, genuinely entangled 3-qubit states can be realized via local projections onto computational or Hadamard basis states (Theorem 12 in \cite{Backens-holant-plus}). 
This theorem is equivalent to the following inductive statement.
\begin{theorem}[\cite{Backens-holant-plus}]\label{three-qubit-entanglement}
Let $|{\Psi}\rangle$ be an $n$-qubit $(n\geqslant 4)$ state exhibiting multipartite entanglement. Then, there exists some $i$ and some $|\theta\rangle\in \{|{0}\rangle, |{1}\rangle, |{+}\rangle,  |{-}\rangle\}$ such that 
$\langle\theta|_i|\Psi\rangle$ exhibits multipartite entanglement.
\end{theorem}
\begin{remark}
This result shows that  multipartite entanglement of an $n$-qubit $(n\geqslant 4)$ state can be \emph{preserved} under projections onto states $|{0}\rangle$, $|{1}\rangle$, $|{+}\rangle$ and $|{-}\rangle$.
\end{remark}

The Holant$^{+}(\mathcal{F})$ problem is defined as $\Holant(\mathcal{F}\cup \{|{0}\rangle, |{1}\rangle, |{+}\rangle, |{-}\rangle\})$.
According to Theorem \ref{three-qubit-entanglement},
we know that  in the framework of Holant$^+$ problems, a genuinely entangled 3-qubit state can always be realized from an $n$-qubit ($n\geqslant4$) state exhibiting multipartite entanglement. 
Then, using a genuinely entangled 3-qubit state,
a full dichotomy was proved for Holant$^{+}$ problems  \cite{Backens-holant-plus}.  
Later, it was generalized to Holant$^c$ problems  \cite{CLX-HOLANTC,Backens-Holant-c} where $\Holant^c(\mathcal{F})$ 
 is defined as $\Holant(\mathcal{F}\cup \{|{0}\rangle, |{1}\rangle\})$.

\subsection{Our results}
In this paper, we consider whether multipartite entanglement can be preserved under projections onto only computational basis states, i.e., $|{0}\rangle$ or $|{1}\rangle$. 
We have the following result.
 

\begin{theorem}\label{preservation}
Let $|{\Psi}\rangle$ be an $n$-qubit $(n\geqslant 4)$  state exhibiting multipartite entanglement and  $\langle0^n|\Psi\rangle\neq 0$.
If $|{\Psi}\rangle$ is not of the form $a|0^n \rangle +b|1^n \rangle$ when $n\geqslant 5$, or when $n=4$ and $|{\Psi}\rangle$ is not of the form $a|0000 \rangle +b|1111 \rangle+c|0011 \rangle+d|1100 \rangle$ (up to a permutation of the  four qubits) where $a, b, c$ and $d$ can  possibly be zero,  then  there exists some $i$  such that $|{\Psi}_i^0\rangle$ or $|{\Psi}_i^1\rangle$ exhibits multipartite entanglement.
\end{theorem}



That $\langle0^n|\Psi\rangle\neq 0$
is a normalization condition, and the other conditions are all necessary
to ensure the preservation of multipartite entanglement under projections to $|0\rangle$
and $|1\rangle$. Thus Theorem~\ref{preservation} is a strengthening 
of Theorem \ref{three-qubit-entanglement}.
More importantly, our approach is in the opposite direction to Backens’.
While Backens proved
results in quantum entanglement theory  to apply it to the complexity classification of Holant problems, we prove new results in quantum entanglement theory 
by employing the machinery from Holant problems. 
We prove Theorem~\ref{preservation} using
a technique developed for Holant problems
called the interplay between the \emph{unique prime factorization} of signatures and \emph{gadget constructions}.
This technique is at the heart of a standard approach (arity reduction) to build inductive arguments for  Holant problems \cite{cai-fu-shao-eo}. 
The new result in quantum entanglement theory sheds light on the classification of entanglement under SLOCC equivalence.


Going one step further, we
ask  whether we can restrict projections onto only one state $|0\rangle$,
while multipartite entanglement is still preserved.
The answer is no. 
Then, one way to salvage the situation is  to consider the self-loop gadget using   one of the 
Bell states, $|\phi^+\rangle=|00\rangle+|11\rangle$ together with projections onto $|0\rangle$.

\begin{definition}[Self-loop]
The self-loop on the $i$-th and $j$-th qubits of a state $|\Psi\rangle$ by the Bell state $|\phi^+\rangle  =|00\rangle+|11\rangle$ is defined as $\langle\phi^+|_{ij}|\Psi\rangle=|\Psi^{00}_{ij}\rangle+|\Psi^{11}_{ij}\rangle$, where $|\Psi^{00}_{ij}\rangle$ and $|\Psi^{11}_{ij}\rangle$ are states of  $n-2$ qubits  when setting the i-th and j-th qubits of $|\Psi\rangle$ to  $00$ and $11$ respectively. 
\end{definition}

\begin{lemma}\label{01-induction-intro}
 Let $|{\Psi}\rangle$ be an $n$-qubit $(n\geqslant 4)$  state exhibiting multipartite entanglement. 
 There exists some choice of three or four of the $n$ qubits such that by performing self-loops 
by $|{\phi^+}\rangle$ and projections    onto  $|{0}\rangle$ of the other qubits, we get
 \begin{itemize}
     \item  a $3$-qubit state exhibiting multipartite entanglement, or
     \item  a {\rm GHZ} type $4$-qubit state, i.e., $|{\rm GHZ}_4\rangle=|0000\rangle+|1111\rangle$, or
     \item  the state $|1\rangle$.
 \end{itemize}
\end{lemma}

Why do we consider $|\phi^+\rangle$ and $|0\rangle$?
The state $|\phi^+\rangle$ is synonymous with the binary {\sc equality} signature $=_2$. It is always available in the Holant framework as it means \emph{merging} two dangling edges in a graph. 
Moreover, we can show that $|0\rangle$  is realizable from any state of odd number of qubits under some mild assumptions. 
Then, we can apply Lemma \ref{01-induction-intro} to get a new dichotomy for Holant problems 
where at least one signature of odd arity is present.


\begin{theorem}\label{odd-intro-sec1}
Let $\mathcal{F}$ be a set of real-valued signatures containing at least one signature of odd arity. Then $\Holant(\mathcal{F})$ is either P-time computable or \#P-hard.
\end{theorem}

\begin{remark}
 Theorem \ref{preservation}  and Lemma \ref{01-induction-intro} hold for  complex-valued $n$-qubit states.
However, Theorem \ref{odd-intro-sec1} is restricted to real-valued signatures, in which the Hermitian conjugate and the complex inner product can be represented by a \emph{mating} gadget in the Holant framework. 
\end{remark}

\subsection{Surprising  discovery of two extraordinary quantum states}
What about signature sets containing only even arity  signatures, in which  $|0\rangle$ cannot be realized.
Since $|\phi^+\rangle$ is always available, we consider whether multipartite entanglement is preserved under self-loops by $|\phi^+\rangle$ alone.
Given an $n$-qubit ($n\geqslant 6$ is even) state $|\Psi\rangle$  exhibiting multipartite entanglement, are there some $i$ and $j$ such that performing a self-loop by $|\phi^+\rangle$ on the $i$-th and $j$-th qubits of $|\Psi\rangle$ results in an $(n-2)$-qubit state exhibiting multipartite entanglement?
(By the definition of multipartite entanglement,
for even $n$ it must be $n \ge 6$.)

The answer is no.
Here we made our most surprising discovery:
There \emph{exist} genuinely entangled  6-qubit and 8-qubit states such that multipartite entanglement is \emph{not} preserved under self-loops. Furthermore, it is not  preserved
under self-loops not only by $|\phi^+\rangle$, but also by all four Bell states, $|\phi^+\rangle, |\psi^+\rangle=|01\rangle+|10\rangle, |\phi^-\rangle=|00\rangle-|11\rangle,$ and $|\psi^{-}\rangle=|01\rangle-|10\rangle.$
The self loop of the $i$-th and $j$-th qubits of $|\Psi\rangle$ by $|\psi^+\rangle$ is defined as $\langle\psi^+|_{ij}|\Psi\rangle=|\Psi^{01}_{ij}\rangle+|\Psi^{10}_{ij}\rangle$. Similarly, we can define $\langle\phi^-|_{ij}|\Psi\rangle$ and $\langle\psi^-|_{ij}|\Psi\rangle$.



\begin{definition}[Bell property]
Let $|{\Psi}\rangle$ be a genuinely entangled state. We say that it satisfies the Bell property if for any two qubits $i$ and $j$ of $|{\Psi}\rangle$
and any Bell state $|{\phi}\rangle$, $\langle\phi|_{ij}|{\Psi}\rangle$ is a tensor product of Bell states. It satisfies the strong Bell property if for any two  $i$ and $j$ 
and any Bell state $|{\phi}\rangle$, $\langle\phi|_{ij}|{\Psi}\rangle$ is a tensor product of the Bell state $|{\phi}\rangle$, i.e., $\langle\phi|_{ij}|{\Psi}\rangle=|\phi\rangle\otimes\cdots \otimes|\phi\rangle$. 
\end{definition}

\begin{theorem}
There exist genuinely entangled  $6$-qubit states that satisfy the Bell property, and  genuinely entangled  $8$-qubit states that satisfy the strong Bell property.
\end{theorem} 

We first give an $8$-qubit state $|\Psi_8\rangle$  that satisfies the strong Bell property. 
\begin{equation*}
\begin{aligned}
    |\Psi_8 \rangle{\scriptstyle =}
    & {\scriptstyle|00000000\rangle+|00001111\rangle+|00110011\rangle+|00111100\rangle
    +|01010101\rangle+|01011010\rangle+|10011001\rangle+|10010110\rangle}\\ {\scriptstyle+}&{\scriptstyle|01101001\rangle+|01100110\rangle+|10100101\rangle+|10101010\rangle+|11000011\rangle+|11001100\rangle+|11110000\rangle+|11111111\rangle.}\\
    \end{aligned}
\end{equation*}
    $ |\Psi_8\rangle$ can be represented by an 8-ary signature    $\Psi_8$.
    The support of $\Psi_8$ has the following structure: the sums of the first four variables, and the last four variables are both even; the assignment of the first four variables are either identical to, or complement of the assignment of the last four variables.  While it is not obvious
    from this description that the support set
    is an affine subspace of  $\mathbb{Z}_2^{8}$,
    but \emph{it is}.  Another interesting description of   $\Psi_8$ is
    as follows: Take  4 bits $x_1, x_2, x_3, x_5$, (these are not the first 4 bits in the description above), then on the support the remaining 4 bits are mod 2 sums of ${4 \choose 3}$ subsets of
    $\{x_1, x_2, x_3, x_5\}$.

        The $6$-qubit state $ |\Psi_6\rangle$ satisfying the Bell property has $32$ nonzero coefficients. We give it in the signature form.
        \vspace{-0.5ex}
        \[\Psi_6(x_1, \ldots, x_6)=(-1)^{x_1x_4+x_2x_5+x_3x_6+x_4x_5+x_5x_6+x_4x_6},\] 
        where the support of $\Psi_6$ is $\sum_{i=1}^{6}x_i=0 \bmod 2$ (even parity). We can write $\Psi_6$ as the following 8-by-8 matrix where the assignment of the first three variables in lexicographic order (from $000$ to $111$) is the row index and the assignment of the last three variables in lexicographic order is the column index. 
        \[M_{123,456}(\Psi_6)=\left[\begin{smallmatrix}
1 & 0 & 0 & 1 & 0 & 1 & 1 & 0\\
0 & -1 & 1 & 0 & 1 & 0 & 0 & -1\\
0 & 1 & -1 & 0 & 1 & 0 & 0 & -1\\
-1 & 0 & 0 & -1 & 0 & 1 & 1 & 0\\
0 & 1 & 1 & 0 & -1 & 0 & 0 & -1\\
-1 & 0 & 0 & 1 & 0 & -1 & 1 & 0\\
-1 & 0 & 0 & 1 & 0 & 1 & -1 & 0\\
0 & 1 & 1 & 0 & 1 & 0 & 0 & 1\\
\end{smallmatrix}\right].\]

We can use Pauli operations to generate more states satisfying the Bell property. 
Consider the following four Pauli operators
\[I=\begin{bmatrix}
1 & 0\\
0 & 1\\
\end{bmatrix}, 
~~~~X=\begin{bmatrix}
0 & 1\\
1 & 0\\
\end{bmatrix},
~~~~Y=\begin{bmatrix}
0 & -\ii\\
\ii & 0\\
\end{bmatrix} ~~~~\text{ and }
~~~~Z=\begin{bmatrix}
1 & 0\\
0 & -1\\
\end{bmatrix}.
\]
A Pauli operation on an $n$-qubit state $|{\Psi}\rangle$ is defined as $P_1\otimes P_2\otimes \ldots \otimes P_n|{\Psi}\rangle$ 
(which produces another $n$-qubit) where each $P_i$ is a Pauli operator. Let $|{\Psi_6}\rangle$ and $|{\Psi_8}\rangle$ be states  described above. 
Let $\mathfrak{P}_6$ and $\mathfrak{P}_8$ denote the sets of states realized by performing Pauli operations on $|{\Psi_6}\rangle$ and $|{\Psi_8}\rangle$ respectively. 
All states in $\mathfrak{P}_6$ and $\mathfrak{P}_8$ satisfy the Bell property.

Due to the existence of these  $6$-qubit and $8$-qubit
states with such extraordinary properties, 
it remains as a difficult task to achieve a full dichotomy for real-valued Holant problems. 
On the other hand, we hope such states can be further investigated and perhaps applied to quantum computing or quantum information theory. 

\vspace{1ex}
 The paper is organized  as follows. In Section~\ref{sec-preserve},  we give a proof of our main quantum entanglement result (Theorem \ref{preservation})
 by using the theory of signatures. 
We give some preliminaries for Holant problems in Section \ref{sec-prelim}. Then, we give
 a proof outline for our dichotomy result (Theorem \ref{odd-intro-sec1}) in Section \ref{sec-outline} and the full proof in Section \ref{full-proof}.

\section{Preservation of Multipartite Entanglement under Projections}\label{sec-preserve}
\subsection{Unique prime factorization and pinning gadgets}
We use the theory of signatures to prove Theorem \ref{preservation}.
Recall that by our definition,  a signature always has arity at least one.
A nonzero signature $g$ \emph{divides} $f$ denoted by $g \mid f$, if there is a signature $h$ such that  $f=g \otimes h$
(with possibly a permutation of variables) or there is a constant $\lambda$ such that $f= \lambda \cdot g$.
In the latter case, if $\lambda \neq 0$, then we also have $f \mid g$ since $g= \frac{1}{\lambda} \cdot f$.
For nonzero signatures, if both $g\mid f$ and $f \mid g$, then they are nonzero constant multiples of
each other, and 
we say $g$ is an \emph{associate} of $f$, 
denoted by $g \sim f$.
In terms of  this division relation, \emph{irreducible} signatures and \emph{prime} signatures are defined. It is proved that they are equivalent, which gives us the \emph{unique prime factorization} (UPF) of signatures \cite{cai-fu-shao-eo}.

\begin{definition}[Irreducible signature]
A nonzero signature $f$ is irreducible if $g \mid f$ implies  $g \sim f$.
\end{definition}

\begin{definition}[Prime signature]
A nonzero signature $f$ is a prime signature, if
for any nonzero signatures $g$ and  $h$,
$f \mid g \otimes h$ 
implies that $f\mid g$ or $f \mid h$. 
\end{definition}

\begin{lemma}\label{prime=irreducible}
The notions of  irreducible signatures and prime signatures are equivalent.
\end{lemma}

A prime factorization of a signature $f$ is $f=g_1\otimes \ldots \otimes g_k$ up to a permutation of
variables, where each $g_i$ is a prime (irreducible) signature. 

\begin{lemma}[Unique prime factorization]\label{unique}
Every nonzero signature $f$ has a prime factorization.
If  $f$ has  prime factorizations
 $f=g_1\otimes \ldots \otimes g_k$ and $f=h_1\otimes \ldots \otimes h_\ell$,
both up to a permutation of variables,
then $k=\ell$ and after reordering the factors we have $g_i \sim h_i$  for all $i$. 
\end{lemma}

A nonzero signature $f$ is irreducible if $f$ cannot be written as $g\otimes h$ for some signatures $g$ and $h$. 
This is equivalent to saying that
 $|f\rangle$ is a genuinely entangled state of multiple qubits or $|f\rangle$ is a single-qubit state.
 Let  $\mathscr{T}_1$ denote the set of tensor products of unary signatures and $\mathscr{T}$ denote the set of tensor products of unary and binary signatures. Then  a state $|f\rangle$ of multiple qubits
is entangled iff  $f\notin\mathscr{T}_1$, and $|f\rangle$ exhibits multipartite entanglement iff $f\notin\mathscr{T}$.
The following result is a direct corollary of Lemma \ref{unique}.

\begin{corollary}\label{cor-upf}
Let $f$ be a nonzero $n$-ary signature. Suppose that there are two irreducible signatures $g$ on a variable set $A$ and $h$ on a variable set $B$ such that $g\mid f$ and $h\mid f$.
Then, either $A$ is disjoint with $B$, or $A=B$ and $g \sim h$. 
\end{corollary}

We use $f_i^{0}$ and $f_i^{1}$ to denote the signature forms of quantum states $|f_i^{0}\rangle$ and $|f_i^{1}\rangle$ realized by projections onto $|0\rangle$ and $|1\rangle$. 
In the Holant framework, these signatures are
realized  by a \emph{pinning} gadget, i.e., connecting  the variable $x_i$ of $f$ with unary signatures $\Delta_0=(1, 0)$ and $\Delta_1=(0, 1)$ respectively. $\Delta_0$ and $\Delta_1$ are signature forms of $|0\rangle$ and $|1\rangle$.
We may further pick a variable $x_j$ of $f_i^{c}$ and pin it to the value $d$ $(c, d\in \{0, 1\})$. 
Obviously, the pinning gadgets on different variables $x_i$ and $x_j$ commute. 
Thus, we have 
$(f_i^c)_{j}^d=(f_j^d)_{i}^c$. 
We denote it by $f^{cd}_{ij}$.

The following lemma is easy to check.

\begin{lemma}\label{pinning-zero}

Let $f$ has arity $n$. If $f_i^0\equiv 0$ for all $i\in[n]$ and $f_j^1\equiv 0$ for some $j\in[n]$, then  $f\equiv 0$.

\end{lemma}

Suppose  that $g\mid f$ where $g$ is on a variable set $A$. Then for any variable $x_i$ of $f$  that is not in $A$, we have $g\mid f_i^0$ and $g\mid f_i^1$ 
(By definition, the division relation holds even if $f_i^0$ or  $f_i^1$ is a zero signature).
Thus, the division relation is unchanged under  pinning gadgets on variables out of $A$. 
The following lemma shows that 
a stronger converse is also true. 

\begin{lemma}\label{lem-division}
Let $f$ be an $n$-ary $(n\geqslant 2)$ signature. If there exists a signature $g$ on a variable set $A$ such that $g\mid f_i^0$ for all $x_i\notin A$, and furthermore $g\mid f_j^1$ for some $x_j\notin A$, then $g\mid f$.
\end{lemma}
\begin{proof}
We may assume $f$ is nonzero, for otherwise the conclusion trivially holds.
We now prove this  for a unary signature $g=(a, b)$. We assume $g$ is on the variable $x_u$.
Consider the signature $f'= bf_u^0-af_u^1$.
Clearly, $f'$ has arity at least $1$. 
For every $i\neq u$, we have $f_i^0=(a, b)\otimes h$ for some $h$. Then, $(f_i^0)_u^0=a\cdot h$, $(f_i^0)_u^1=b\cdot h$, and hence
$$(f')_i^0= (bf_u^0-af_u^1)_i^0=bf_{ui}^{00}-af_{ui}^{10}=b(f_{i}^{0})_u^0-a(f_{i}^{0})_u^1=ba\cdot h-ab\cdot h \equiv 0.$$
Moreover, there is an index $j\neq u$ such that 
 $g\mid f_j^1$, i.e., $f_j^1=(a, b)\otimes h'$ for some $h'$.
 Then, $(f_j^1)_u^0=a\cdot h'$, $(f_j^1)_u^1=b\cdot h'$, and hence
 $(f')_j^1= (bf_u^0-af_u^1)_j^1=b(f_{j}^{1})_u^0-a(f_{j}^{1})_u^1=ba\cdot h'-ab\cdot h' \equiv 0.$
By Lemma \ref{pinning-zero}, we have $f'\equiv 0$. Thus, we have $f_u^0 : f_u^1 = a : b$, and hence $g\mid f$.

For a signature $g$ of arity $\geqslant n-2$, the proof is essentially the same, which we omit here.
\end{proof}

\subsection{The proof of Theorem~\ref{preservation}}
We restate  Theorem \ref{preservation} in terms of signatures.
We use $\mathscr{S}(f)$ to denote the support of $f$, i.e., $\mathscr{S}(f)=\{\alpha\mid f(\alpha)\neq 0\}$. We use $0^n$ and $1^n$ to denote the $n$-bit all-$0$ and all-$1$ strings.

\begin{theorem}\label{tho-entangle}
Let $f$ be an $n$-ary $(n\geqslant 4)$ signature, $f\notin \mathscr{T}$ and $f(0^n)\neq 0$.
If $\mathscr{S}(f)\not\subseteq \{0^n, 1^n\}$ when $n\geqslant 5$, or $\mathscr{S}(f)\not\subseteq \{0000, 1111, 0011, 1100\}$ up to any permutation of four variables when $n=4$, then there exists some $i$ such that $f_i^0$ or $f_i^1$ is not in $\mathscr{T}$.

\end{theorem}
\begin{proof}
Since $f(0^n)\neq 0$, we have $f^0_i\not\equiv 0$ and $f^{00}_{ij}\not\equiv 0$ 
(not identically 0) for all indices $i$ and $j$.
Also, 
since the support $\mathscr{S}(f)\not\subseteq \{0^n, 1^n\}$,
there exist some $s$ and $t$ such that $f^{01}_{st}\not\equiv 0$.
For a contradiction, we assume  $f_i^0, f_i^1 \in\mathscr{T}$ for all $i$. 
We consider the following two possible cases.

{\bf Case 1}. For all indices $i$, $f_i^0 \in \mathscr{T}_1$ (i.e., tensor product of unary signatures). 

We will show that in this case, there is a unary signature  $a(x_u)$ on some variable $x_u$,
such that $a(x_u)\mid f$. This will lead us to a contradiction.

 Recall that there exist variables $x_s$ and $x_t$ such that $f^{01}_{st}\not\equiv 0$. 
 We consider $f^1_{t}$.
 Clearly, $f^1_{t}\not\equiv 0$. Since $f^1_{t}\in \mathscr{T}$, in the UPF of $f^1_{t}$, the variable $x_s$ may appear in a unary signature or an irreducible binary signature. In both cases,
 since $f$ has arity at least $4$,  we can pick a variable $x_u$ such that $x_u$ and $x_s$ appear in two distinct irreducible signatures in the UPF of $f^1_{t}$ (i.e., $x_u$ and $x_s$ are not entangled in $f^1_{t}$).
 Then, we show that $x_u$ must appear in a unary signature  in the UPF of $f^1_{t}$.
 Otherwise, there is an irreducible binary signature  $b(x_u, x_{v})$ such that $b(x_u, x_v)\mid f_t^1$.
 Since $x_u$ is not entangled with $x_s$ in $f_t^1$, we have $v\neq s$. Then, $b(x_u, x_{v})\mid (f_{t}^{1})_s^0$.
 On the other hand, we consider $f_s^0$. 
 By our assumption, $f_s^0\in \mathscr{T}_1$ and hence there exists some unary signature $a'(x_u)$ such that $a'(x_u)\mid f_s^0$. 
 Then, we have $a'(x_u)\mid (f_{s}^{0})_t^1$.
 Recall that the pinning gadgets on different variables commute. Thus, $(f_t^1)_s^0=(f_s^0)_t^1=f_{st}^{01}$, and we know that it is a nonzero signature. 
 By Corollary \ref{cor-upf}, we have $b(x_u, x_v)\sim a'(x_u)$. This is a contradiction. 
 Thus, there exists some unary signature $a(x_u)$ such that $a(x_u)\mid f_{t}^1$. 
 
 Now we show that $a(x_u)\mid f_i^0$ for all indices $i\neq u$. 
 First, we show that $a(x_u)\mid f_s^0$. 
Since $f_s^0\in \mathscr{T}_1$, there exists some unary signature $a'(x_u)$ such that $a'(x_u)\mid f_s^0$, and then $a'(x_u)\mid f_{st}^{01}$.
Also, we have $a(x_u)\mid f_{st}^{01}$ since $a(x_u)\mid f_{t}^1$. 
Since $f^{01}_{st}\not\equiv 0$, by Corollary \ref{cor-upf}, we have $a(x_u)\sim a'(x_u)$. 
Thus, $a(x_u)\mid f_{s}^0$.
Then, we consider $f_i^0$ for all indices $i\notin  \{u, s\}$.
 Since $f_i^0\in \mathscr{T}_1$, 
 there exists a unary signature $a''(x_u)$ such that $a''(x_u)\mid f_i^0$, and then  $a''(x_u)\mid f_{is}^{00}.$
Also, we have $a(x_u)\mid f_{is}^{00}$ 
since $a(x_u)\mid f_{s}^{0}$. 
Recall that $f_{is}^{00}\not\equiv 0$ for all $i$. 
Then, by Corollary \ref{cor-upf}, we have $a(x_u)\sim a''(x_u)$. Thus, $a(x_u)\mid f_i^0$.

Since $a(x_u)\mid f_{i}^0$ for all $i\neq u$ and $a(x_u)\mid f_t^1$ for some $t\neq u$, by Lemma \ref{lem-division},  we have $a(x_u)\mid f$. 
In other words, $f=a(x_u)\otimes g$ where $g$ is a nonzero signature of arity $n-1$ on variables other than $x_u$.
Since $f\notin \mathscr{T}$, we have $g \notin \mathscr{T}$. 
Consider $f_u^0$. We know that it is a nonzero signature and hence $f_u^0\sim g$. 
Thus, $f_u^0\notin \mathscr{T}$. We have reached a contradiction.

{\bf Case 2.} There exists some index $k$ and an irreducible binary signature $b(x_v, x_w)$ such that $b(x_v, x_w)\mid f_k^0$.

We will show that in this case,  $b(x_v, x_w)\mid f$. 
First, we show that $b(x_v, x_w)\mid f_{i}^0$ for all $i \notin \{v, w\}$. 
We already have $b(x_v, x_w)\mid f_{k}^0$.
Consider $f_i^0$ for all indices $i\notin \{v, w, k\}$. 
Since $f_i^0\in \mathscr{T}$ and $f_i^0 \not\equiv 0$, there is either a unary signature $a(x_v)$ or an irreducible binary signature $b'(x_v, x_{w'})$ that appears in the UPF of $f_i^0$, i.e., $a(x_v)\mid f_i^0$ or $b'(x_v, x_{w'})\mid f_i^0$. 
In the former case, we have $a(x_v)\mid f_{ik}^{00}$. 
In the latter case and if $w'\neq k$, we have $b'(x_v, x_{w'})\mid f_{ik}^{00}$.
In the latter case and if $w'=k$, then
let $a'(x_v)$ be the unary signature realized from $b'(x_v, x_{w'})$ by pinning $x_{w'}=x_k$ to $0$,
we get 
$a'(x_v)\mid f_{ik}^{00}$.
On the other hand, 
since $b(x_v, x_w)\mid f_k^0$, we have $b(x_v, x_w)\mid f_{ik}^{00}$.
Since $f_{ik}^{00}\not\equiv 0$, by Corollary \ref{cor-upf},  we know that the two cases that $a(x_v)\mid f_{ik}^{00}$ and $a'(x_v)\mid f_{ik}^{00}$ cannot occur. 
Thus,  $w'\neq k$ and $b'(x_v, x_{w'})\mid f_{ik}^{00}$. By  Corollary~\ref{cor-upf}, $b(x_v, x_w)\sim b'(x_v, x_{w'})$.
Thus, we have that $b(x_v, x_w)\mid f_i^0$ for all $i\notin \{v, w\}$.

Then we want to show that there exists some  $j\notin \{v, w\}$ such that $b(x_v, x_w)\mid f_j^1$.
\begin{itemize}
    \item 
 We first consider the case that there exist some indices $i$ and $j$ where $\{i, j\}$ is disjoint with $\{v, w\}$ such that $f_{ij}^{01}\not\equiv 0$. 
 We show that $b(x_v, x_w)\mid f_{j}^1$. 
Since $b(x_v, x_w)\mid f_i^0$, we have $b(x_v, x_w)\mid f_{ij}^{01}$. 
By assumption $f_{ij}^{01}\not\equiv 0$, and then clearly $f_j^1\not\equiv 0$.
Recall that $f_{j}^1 \in \mathscr{T}$. 
Again, there is either a unary signature $a(x_v)$ or an irreducible binary signature $b'(x_v, x_{w'})$ that appears in the UPF of $f_j^1$, i.e., $a(x_v)\mid f_j^1$ or $b'(x_v, x_{w'})\mid f_j^1$. 
In the first case since $i \not = v$, we can pin $x_i$ of $f_j^1$ to $0$, and 
we get $a(x_v)\mid f_{ij}^{01}$. In the second case and if $w' = i$,
again we can get $a'(x_v)\mid f_{ij}^{01}$, where $a'(x_v) = b'(x_v, 0)$,
obtained from pinning $x_i$ to 0.  But $f_{ij}^{01} \not\equiv 0$ and $b(x_v, x_w)\mid f_{ij}^{01}$. Then, in
the UPF of $f_{ij}^{01}$, it does not have a unary signature on $x_v$ as a factor.
Thus, it must be the case that $b'(x_v, x_{w'})\mid f_j^1$ where $w' \not = i$. Then, 
we have $b'(x_v, x_{w'})\mid f_{ij}^{01}$.
Since $b(x_v, x_{w})\mid f_{ij}^{01}$ and $f_{ij}^{01}\not\equiv 0$,
by Corollary \ref{cor-upf}, $b'(x_v, x_{w'}) \sim b(x_v, x_{w})$, and thus
$b(x_v, x_w)\mid f_j^1$. 
%
Then, by Lemma \ref{lem-division}, we have $b(x_v, x_w)\mid f$. 
In other words, $f=b(x_v, x_w)\otimes h$ where $h$ is a nonzero signature of arity $n-2$ on variables other than $x_v$ and $x_w$.
Since $f\notin \mathscr{T}$, we have $h \notin \mathscr{T}$. 
Then consider $f_v^0$. We know that it is a nonzero signature and  $h\mid f_v^0$. 
Thus, $f_v^0\notin \mathscr{T}$. Contradiction.

\item Then we consider the case that $f_{ij}^{01}\equiv 0$ for all indices $\{i, j\}$ that are disjoint with $\{v, w\}$. 
Consider an $n$-bit input $\alpha$ of $f$. 
We write $\alpha$ as $\alpha_v\alpha_w\beta$ where $\alpha_v$ is the input on variable $x_v$, $\alpha_w$ is the input on variable $x_w$, and $\beta$ is the input on the other $n-2$ variables.
Then, $f(\alpha)=0$ if $\beta$ is not the all-0 or all-1 bit string in $\{0, 1\}^{n-2}$.
It follows that $f$ has at most eight nonzero entries. 
We list all its entries by the following $4$-by-$2^{n-2}$ matrix $M_{vw}(f)$ with variables $(x_v, x_w)\in \{0, 1\}^{2}$ as the row index (in the order 00, 01, 10, 11) and the assignment of the other variables in lexicographic order as the column index. 
\[M_{vw}(f)=\left[\begin{matrix}
c_1 & 0 & \ldots & \ldots & 0 & c_2\\
c_3 & 0 & \ldots & \ldots & 0 & c_4\\
c_5 & 0 & \ldots & \ldots & 0 & c_6\\
c_7 & 0 & \ldots & \ldots & 0 & c_8\\
\end{matrix}\right].\]
Here, $c_1=f(0^n)\neq 0$.
Consider signatures $f_v^0$ and $f_v^1$. They have the following matrix forms with the variable $x_w\in \{0, 1\}$ as the row index.
\[M_{w}(f_v^0)=\left[\begin{matrix}
c_1 & 0 & \ldots & \ldots & 0 & c_2\\
c_3 & 0 & \ldots & \ldots & 0 & c_4\\
\end{matrix}\right] \text{ ~~~and~~~ } M_{w}(f_v^1)=\left[\begin{matrix}
c_5 & 0 & \ldots & \ldots & 0 & c_6\\
c_7 & 0 & \ldots & \ldots & 0 & c_8\\
\end{matrix}\right].\]
Also consider signatures $f_w^0$ and $f_w^1$. They have the following matrix forms with the variable $x_v\in \{0, 1\}$ as the row index.
\[M_{v}(f_w^0)=\left[\begin{matrix}
c_1 & 0 & \ldots & \ldots & 0 & c_2\\
c_5 & 0 & \ldots & \ldots & 0 & c_6\\
\end{matrix}\right] \text{ ~~~and~~~ } M_{v}(f_w^1)=\left[\begin{matrix}
c_3 & 0 & \ldots & \ldots & 0 & c_4\\
c_7 & 0 & \ldots & \ldots & 0 & c_8\\
\end{matrix}\right].\]

Consider $f_v^0$. 
Since $f$ has arity at least $4$, $f_v^0$ has arity at least $3$.
Since $f_v^0\in \mathscr{T}$, the variable $x_w$ either appears in a unary factor $a(x_w)$ of $f_v^0$ or an irreducible binary factor $b(x_w, x_{w'})$ of $f_v^0$. 
In the latter case, we can pick another variable $x_r$ of $f_v^0$ where $r\neq w$ or $w'$, and we consider 
$f_{vr}^{00}$. 
We know that $f_{vr}^{00}\not\equiv 0$ since $c_1\neq 0$ and $b(x_w, x_{w'})\mid f_{vr}^{00}$ since  $b(x_w, x_{w'})\mid f_v^0$.
Notice that the column with $c_2$ and $c_4$ does not appear in $M_w(f_{vr}^{00})$.
Thus, the signature $f_{vr}^{00}$ is of the form $(c_1, c_3)\otimes (1, 0)^{\otimes{(n-2)}}$ which is a tensor product of unary signatures. Contradiction.
Thus, there is a unary signature $a(x_w)$ such that $a(x_w)\mid f_v^0$. 
Then, we have $c_1c_4=c_2c_3$.
Similarly by considering $f_w^0$, we have $c_1c_6=c_2c_5$.
Now, we consider $f_v^1$, and prove $c_5c_8=c_6c_7$. If $c_5=c_7=0$, then clearly we have $c_5c_8=c_6c_7=0$. 
Otherwise, for any $r\neq w$ or $v$, we have  $f_{vr}^{10}=(c_5, c_7)\otimes (1, 0)^{\otimes{(n-2)}}\not\equiv 0$ which is a tensor product of unary signatures.
If there is a binary signature $b(x_w, x_{w'})$ such that $b(x_w, x_{w'})\mid f_v^1$, then we can find some $r\neq w, w'$ such that $b(x_w, x_{w'})\mid f_{vr}^{10}$. Contradiction.
Thus, there is a unary signature $a(x_w)$ such that $a(x_w)\mid f_v^1$. 
Then, we have $c_5c_8=c_6c_7$.
Similarly by considering $f_w^1$, we have $c_3c_8=c_4c_7$.
\begin{itemize}
    \item Suppose  $n\geqslant 5$. Then $f_v^0$ has arity at least $4$. We first show that $c_2=0$.
    We consider $f_{vw}^{00}=[c_1, 0, \ldots \ldots, 0, c_2].$
    Since $f_v^0\in \mathscr{T}$, we have $f_{vw}^{00}\in \mathscr{T}.$
    Note that $f_{vw}^{00}$ has arity at least $3$. Since $c_1\neq 0$, the only possible value of $c_2$ to make $f_{vw}^{00}\in \mathscr{T}$ is $0$. Thus, $c_2=0$.
    Since $c_1c_4=c_2c_3=0$ and $c_1\neq 0$, we have $c_4=0$. 
    Also, since $c_1c_6=c_2c_5=0$ and $c_1\neq 0$, we have $c_6=0$.
If $c_8=0$, then $f=b(x_v, x_w)\otimes(1, 0)^{\otimes(n-2)}\in \mathscr{T}$.  A contradiction with $f\notin \mathscr{T}$. Thus, we have $c_8\neq 0$.
Since $c_5c_8=c_6c_7=0$ and $c_8\neq 0$, we have $c_5=0$.
Also since $c_3c_8=c_4c_7=0$ and $c_8\neq 0$, we have $c_3=0$.
Consider $f_{vw}^{11}=[c_7, 0, \ldots \ldots, 0, c_8]$. 
Since $f_{vw}^{11}\in \mathscr{T}$ and it has arity at least $3$, and $c_8\neq 0$, we have $c_7=0$.
Thus, $f$ has only two nonzero entries that are on the all-0 input and the all-1 input. 
A contradiction with our assumption that $\mathscr{S}(f)\not\subseteq \{0^n, 1^n\}.$
\item Suppose  $n=4$. 
If $c_2=0$, then with the same 
proof as in the case that $n\geqslant 5$, we have $c_4=c_6=0$, $c_8\neq 0$ and then $c_3=c_5=0$.
Thus, $\mathscr{S}(f)\subseteq\{0000, 1111, 1100\}$. Contradiction. 
Otherwise, $c_2\neq 0$. Suppose that $c_2=kc_1$. Then $c_4=kc_3$ since $c_1c_4=c_2c_3$ and $c_6=kc_5$ since $c_1c_6=c_2c_5$.
If $c_3$ and $c_4$ are not zero, then $c_8=kc_7$ since $c_3c_8=c_4c_7$.
Then, $f=b(x_v, x_w)\otimes (1, 0, 0, k)\in \mathscr{T}.$ Contradiction. 
Thus, $c_3=c_4=0$.
Similarly, if $c_5$ and $c_6$ are not zero, then we still have $c_8=kc_7$ since $c_5c_8=c_6c_7$. Then, we have $f\in \mathscr{T}$. Contradiction.
Thus, $c_5=c_6=0$. 
Then, $\mathscr{S}(f)\subseteq\{0000, 1111, 0011, 1100\}$. Contradiction.
\end{itemize}
\end{itemize}
Therefore, there exists some $i$ such that $f_i^0$ or $f_i^1$ is not in $\mathscr{T}$.\end{proof}
\vspace{-1ex}

Our result can be used in the classification of entanglement under SLOCC equivalence.
An $n$-qubit state  $|{\Psi}\rangle$ is equivalent to another $n$-qubit state  $|{\Phi}\rangle$ under SLOCC if there exist some invertible 2-by-2 matrices $M_1$, $M_2$,  $\ldots$, $M_n$ such that 
$|{\Psi}\rangle=M_1\otimes M_2 \otimes \ldots \otimes M_n|{\Phi}\rangle$.
Physicists are interested in the classification of SLOCC equivalence classes. 
For 2-qubit states there are two SLOCC classes, and for 3-qubit states there are six SLOCC classes \cite{slocc2}.
However, for states of 4 or more qubits there are infinitely many SLOCC classes \cite{slocc2}.
Then, the goal is to categorize these  classes into some finitely many families with common physical or mathematical properties. 
Depending on which properties are used, there are different approaches. 
One powerful approach that can possibly handle states of a high number of qubits is by induction \cite{3-qubit-l, 4-qubit-l, Backens-class, 4-qubit-g}. In this approach, the classification of $n$-qubit states relies on the classification of $(n-1)$-qubit  states. 

Consider an $n$-qubit state $|{\Psi}\rangle$. We can pick some index $i$ and write $|{\Psi}\rangle$ as $|{\Psi}\rangle=|{0}\rangle|{\Psi}_i^0\rangle+|{1}\rangle|{\Psi}_i^1\rangle$.
Families of entanglement classes of $|{\Psi}\rangle$ can be defined according to the types of entanglements found in the linear span$\{|{\Psi}_i^0\rangle, |{\Psi}_i^1\rangle\}$ which is related to the entanglement types of $|{\Psi}_i^0\rangle$ and $|{\Psi}_i^1\rangle$ themselves. 
Theorem~\ref{preservation} gives a direct relation between the entanglement types of $|{\Psi}\rangle$ and $\{|{\Psi}_i^0\rangle, |{\Psi}_i^1\rangle\}$. 
For example,  consider a $5$-qubit state exhibiting multipartite entanglement. 
First, by performing SLOCC using the matrix $N=\begin{bsmallmatrix}
0 & 1\\
1 & 0\\
\end{bsmallmatrix}$ on this state, 
we can always get a state  $|{\Psi}\rangle$ where the coefficient of $|0^5\rangle$ is nonzero. 
If $|{\Psi}\rangle$ has the form $a|0^5\rangle+b|1^5\rangle$, then 
it is equivalent to $|{\rm GHZ}_5\rangle=|0^5\rangle+|1^5\rangle$. 
Otherwise, we can apply Theorem \ref{preservation}. There exists some $i$ such that $|{\Psi}_i^0\rangle$ or $|{\Psi}_i^1\rangle$ exhibits multipartite entanglement. 
Then, in order to classify the state $|{\Psi}\rangle$, 
we only need to consider possible entanglement types of $\{|{\Psi}_i^0\rangle, |{\Psi}_i^1\rangle\}$ where at least one  state exhibits multipartite entanglement. 
This eliminates many cases compared to considering all  entanglement types of $\{|{\Psi}_i^0\rangle, |{\Psi}_i^1\rangle\}$.

\section{Preliminaries for Holant Problems}\label{sec-prelim}
\subsection{Definitions and notations}
A constraint function $f$, or a signature, of arity $n>0$
is a map $\mathbb{Z}_2^{n} \rightarrow \mathbb{C}$.
We use $f^\alpha$ to denote $f(\alpha)$.
If $\overline{f^{\alpha}}=f^{\overline{\alpha}}$ for all $\alpha$ where $\overline{f^{\alpha}}$ denotes the complex conjugation of $f^\alpha$ and $\overline{\alpha}$ denotes the bit-wise complement of $\alpha$, we say $f$ satisfies the {\sc ars}.
We use ${\rm wt}({\alpha})$ to denote the Hamming weight of $\alpha$.
The support $\mathscr{S}(f)$ of a signature is 
$\{\alpha \in \mathbb{Z}_2^{n} \mid f^\alpha \neq 0\}$.
A signature $f$ of arity $n$ has even (odd) parity if $\mathscr{S}(f) \subseteq \{\alpha \in \mathbb{Z}^n_2 \mid {\rm wt}(\alpha)$ is even (odd)$\}$.
Let $\mathcal{F}$ be any fixed set of signatures. 
A signature grid
$\Omega=(G, \pi)$ over $\mathcal{F}$
 is a tuple, where $G = (V,E)$
is a graph without isolated vertices,
 $\pi$ labels each $v\in V$ with a signature
$f_v\in\mathcal{F}$ of arity ${\operatorname{deg}(v)}$,
and labels the incident edges
$E(v)$ at $v$ with input variables of $f_v$.
 We consider all 0-1 edge assignments $\sigma$,
each gives an evaluation
$\prod_{v\in V}f_v(\sigma|_{E(v)})$, where $\sigma|_{E(v)}$
denotes the restriction of $\sigma$ to $E(v)$.
\begin{definition}[Holant problems]
The input to the problem {\rm Holant}$(\mathcal{F})$
is a signature grid $\Omega=(G, \pi)$ over $\mathcal{F}$.
The output is  the partition function
$$\Holant_{\Omega} =\sum\limits_{\sigma:E(G)\rightarrow\{0, 1\}}\prod_{v\in V(G)}f_v(\sigma|_{E_{(v)}}).$$
The bipartite Holant problems  {\rm Holant}$(\mathcal{F} \mid \mathcal{G})$
are {\rm Holant} problems
 over bipartite graphs $H = (U,V,E)$,
where each vertex in $U$ or $V$ is labeled by a signature in $\mathcal{F}$ or $\mathcal{G}$ respectively.
\end{definition}

Counting  constraint satisfaction problems (\#CSP) can be expressed as Holant problems (Lemma 1.2 in \cite{jcbook}).
We use $=_n$ to denote the {\sc Equality} signature of arity $n$,
which takes value $1$ on the all-0 or all-1 input  
and $0$ elsewhere. 
(We may also denote the n-bits all-0 and all-1 strings by $\Vec{0}_n$ and $\Vec{1}_n$ respectively.  We may omit $n$ when it is clear from the context.)
Let $\mathcal{EQ}=\{=_1, =_2, \ldots, =_n, \ldots\}$ denote the set of all {\sc Equality} signatures. 
We use $(\neq_2)$ to denote the binary {\sc Disequality} signature with truth table $(0, 1, 1, 0)$. 

\begin{lemma}\cite{jcbook}
$\CSP(\mathcal{F})\equiv_T \Holant(\mathcal{EQ} \mid \mathcal{F})$.
\end{lemma}

A signature $f$ of arity $n \geqslant 2$ can be expressed as 
a $2^k \times 2^{n-k}$ matrix $M_{[k],[n-k]}(f)$, which lists the $2^n$ many values of $f$ with the  $[k]$ many variables as row index and the other $[n-k]$ many variables as column index.
In particular, $f$ can be expressed as a $2\times 2^{n-1}$ matrix $M_{i}(f)$
which lists the $2^n$ values of $f$ with variable $x_i
\in \{0, 1\}$ as row index and the assignments of the other
 $n-1$ variables in lexicographic order as column index.
 That is, $$M_{i}(f)=\begin{bmatrix}
f^{0,00\ldots0} & f^{0, 00\ldots1} & \ldots & f^{0, 11\ldots1}\\
f^{1,00\ldots0} & f^{1, 00\ldots1} & \ldots & f^{1, 11\ldots1}\\
\end{bmatrix}=\begin{bmatrix}
{\bf {f}}^{0}_{i}\\
{\bf {f}}^{1}_{i}\\
\end{bmatrix},$$
where $ {\bf {f}}^{a}_{i}$ 
 denotes the row vector indexed by $x_i=a$ in $M_{i}(f)$.
We may omit the subscript 
when the meaning is clear from context.

%

\subsection{Holographic transformation}
To introduce the idea of holographic transformation,
it is convenient to consider bipartite graphs.
For a general graph,
we can always transform it into a bipartite graph while preserving the Holant value,
as follows.
For each edge in the graph,
we replace it by a path of length two.
(This operation is called the \emph{2-stretch} of the graph and yields the edge-vertex incidence graph.)
Each new vertex is assigned the binary \textsc{Equality} signature $(=_2)$. Thus, we have $\holant{=_2}{\mathcal{F}}\equiv_T \Holant(\mathcal{F})$.

For an invertible $2$-by-$2$ matrix $T \in {\rm GL}_2({\mathbb{C}})$
 and a signature $f$ of arity $n$, written as
a column vector $f \in \mathbb{C}^{2^n}$, we denote by
$Tf = T^{\otimes n} f$ the transformed signature.
  For a signature set $\mathcal{F}$,
define $T\mathcal{F} = \{T f \mid  f \in \mathcal{F}\}$ the set of
transformed signatures.
For signatures written as
 row vectors  we define
$f T^{-1}$ and  $\mathcal{F} T^{-1}$ similarly.
Whenever we write $T f$ or $T \mathcal{F}$,
we view the signatures as column vectors;
similarly for $f T^{-1}$ or $\mathcal{F} T^{-1}$ as row vectors.
We can also represent $Tf$ as the matrix $M_{[k], [n-k]}(Tf)$ with $[k]$ variables as row index and the other $[n-k]$ variables as column index. Then,  we  have $M_{[k], [n-k]}(Tf)=T^{\otimes k}M_{[k], [n-k]}(f)(T^{\tt T})^{\otimes n-k}$. Similarly, $M_{[k], [n-k]}(fT^{-1})=({T^{-1}}^{\tt T})^{\otimes k}M_{[k], [n-k]}(f)(T^{-1})^{\otimes n-k}$.

Let $T \in {\rm GL}_2({\mathbb{C}})$.
The holographic transformation defined by $T$ is the following operation:
given a signature grid $\Omega = (H, \pi)$ of $\holant{\mathcal{F}}{\mathcal{G}}$,
for the same bipartite graph $H$,
we get a new signature grid $\Omega' = (H, \pi')$ of $\holant{\mathcal{F} T^{-1}}{T \mathcal{G}}$ by replacing each signature in
$\mathcal{F}$ or $\mathcal{G}$ with the corresponding signature in $\mathcal{F} T^{-1}$ or $T \mathcal{G}$.

\begin{theorem}[Valiant's Holant Theorem~\cite{valiant-holo-08}]
 For any $T \in {\rm GL}_2({\mathbb{C}})$,
  \[\Holant_\Omega(\mathcal{F} \mid \mathcal{G}) = \Holant_{\Omega'}(\mathcal{F} T^{-1} \mid T \mathcal{G}).\]
\end{theorem}

Therefore,
a holographic transformation does not change the complexity of the Holant problem in the bipartite setting. 
Let $Q\in \mathbb{R}^{2 \times 2}$ be a 2-by-2 orthogonal matrix. 
Note that $(=_2)Q^{-1}=(=_2)$.
 We have $$\holant{=_2}{\mathcal{F}}\equiv_T \holant{=_2}{Q\mathcal{F}}.$$

A particular holographic transformation that will be commonly  used in this paper is the transformation defined by $Z^{-1}=\frac{1}{\sqrt{2}}\left[\begin{smallmatrix} 
1 & - \ii \\
1 & \ii \\
\end{smallmatrix}
\right].$ 
Recall that $\neq_2$ denotes the binary  {\sc Disequality} signature
with truth table $(0, 1, 1, 0)$. It can also be expressed by the $2$-by-$2$ matrix $N_2=\left[\begin{smallmatrix}
0 & 1 \\
 1 & 0
\end{smallmatrix}\right]$ with one variable indexing rows  and the other
indexing
 columns respectively. 
 Note that $(=_2)Z=(\neq_2)$. 
Therefore, we have $$\holant{=_2}{\mathcal{F}}\equiv_T \holant{\neq_2}{Z^{-1}{\mathcal{F}}}.$$
We denote $Z^{-1}{\mathcal{F}}$ by $\widehat{{\mathcal{F}}}$ and $Z^{-1}f$ by $\widehat{f}$.
We know $f$ and $\widehat{f}$ have the following relation \cite{cai-fu-shao-eo}.
\begin{lemma}\label{real-ars}
$f$ is a real valued signature iff $\widehat{f}$ satisfies {\sc ars}.
\end{lemma} 

The following fact is easy to check.
\begin{lemma}\label{q-parity}
Let $Q\in \mathbb{R}^{2 \times 2}$. $Q$ is orthogonal up to a scalar iff $\widehat{Q}=Z^{-1}Q(Z^{-1})^{\tt T}$ is diagonal or anti-diagonal.
\end{lemma}
\begin{proof}
For $Q=\left[\begin{smallmatrix} 
a & b \\
-b & a \\
\end{smallmatrix}
\right]$, we have  $\widehat{Q}=\left[\begin{smallmatrix} 
0 & a+b\ii \\
a-b\ii & 0 \\
\end{smallmatrix}
\right];$
 for $Q=\left[\begin{smallmatrix} 
a & b \\
b & -a \\
\end{smallmatrix}
\right]$ we have  $\widehat{Q}=\left[\begin{smallmatrix} 
a-b\ii & 0 \\
0 & a+b\ii \\
\end{smallmatrix}
\right].$
\end{proof}

\subsection{Gadget construction}
One basic tool used throughout the paper is gadget construction.
An $\mathcal{F}$-gate is similar to a signature grid $(G, \pi)$ for $\Holant(\mathcal{F})$ except that $G = (V,E,D)$ is a graph with internal edges $E$ and dangling edges $D$.
The dangling edges $D$ define input variables for the $\mathcal{F}$-gate.
We denote the regular edges in $E$ by $1, 2, \dotsc, m$ and the dangling edges in $D$ by $m+1, \dotsc, m+n$.
Then the  $\mathcal{F}$-gate  defines a function $f$
\[
f(y_1, \dotsc, y_n) = \sum_{\sigma: E \rightarrow\{0, 1\}} \prod_{v\in V}f_v(\hat{\sigma}\mid_{E(v)})
\]
where $(y_1, \dotsc, y_n) \in \{0, 1\}^n$ is an assignment on the dangling edges, $\hat{\sigma}$ is the extension of  $\sigma$ on $E$ by the assignment $(y_1, \ldots, y_m)$, and $f_v$ is the signature assigned at each vertex $v \in V$.
This function $f$ is called the signature of the $\mathcal{F}$-gate.
There may be no internal edges in an $\mathcal{F}$-gate at all. In this case, $f$ is simply a tensor product of these signatures $f_v$, i.e., $f={\bigotimes}_{v\in V}f_v$ (with possibly a permutation of variables).
We say a signature $f$ is \emph{realizable} from a signature set $\mathcal{F}$ by gadget construction
if $f$ is the signature of an 
 $\mathcal{F}$-gate. 
If $f$ is realizable from a set $\mathcal{F}$,
then we can freely add $f$ into $\mathcal{F}$ while preserving the complexity (Lemma 1.3 in \cite{jcbook}). 
\begin{lemma}
\cite{jcbook}
If $f$ is realizable from a set $\mathcal{F}$, then $\Holant(f, \mathcal{F})\equiv_T\Holant(\mathcal{F})$.
\end{lemma}
Note that, in the setting of $\holant{=_2}{\mathcal{F}}$, every edge is labeled by $=_2$, while in the setting of  $\holant{\neq_2}{\widehat{\mathcal{F}}}$, every edge in a gadget is labeled by $\neq_2$.

A basic gadget construction is \emph{merging}. In the setting of  $\holant{=_2}{\mathcal{F}}$,  given a signature $f\in \mathcal{F}$ of arity $n$, we can connect two variables $x_i$ and $x_j$ of $f$ using $=_2$, and this operation gives a signature of arity $n-2$. We use $\partial_{ij}f$ to denote this signature and $\partial_{ij}f=f^{00}_{ij}+f^{11}_{ij}$, where $f^{ab}_{ij}$ denotes the signature obtained by setting $(x_i, x_j)=(a, b)\in \{0, 1\}^2$. (We use
$f^{ab}_{ij}$ to denote a function, and ${\bf f}^{ab}_{ij}$ to denote a vector that lists the truth table of $f^{ab}_{ij}$ in a given order.)
While in the setting of $\holant{\neq_2}{\widehat{\mathcal{F}}}$, the above merging operation is equivalent to connecting two variables  $x_i$ and $x_j$ of $\widehat{f}$ using $\neq_2$. We denote the resulting signature by $\widehat{\partial}_{ij}\widehat{f}$, and we have $\widehat{\partial_{ij}f}=\widehat{\partial}_{ij}\widehat{f}={\widehat{f}}^{01}_{ij}+{\widehat f}^{10}_{ij}.$
Note that, the merging operation using $=_2$ on a signature is synonymous with the self-loop of an $n$-qubit by bell state $|\phi^{+}\rangle$, and the merging operation using $\neq_2$ is synonymous with the self-loop by bell state $|\psi^{+}\rangle$.

If by merging any two variables of $\widehat{f}$ in the setting of $\holant{\neq_2}{\widehat{\mathcal{F}}}$, we can only realize the zero signature, then the following result shows that $\widehat{f}$ itself is ``almost'' a zero signature \cite{cai-fu-shao-eo}.

\begin{lemma}\cite{cai-fu-shao-eo}\label{zero}
Let $\widehat{f}$ be a signature of arity $n\geqslant 3$. If for any indices $\{i, j\}$, by merging variables $x_i$ and $x_j$ of $\widehat{f}$ using $\neq_2$, we have $\widehat \partial_{ij}\widehat f\equiv0$, then ${\widehat f}^\alpha=0$ for any $\alpha$ with $0 < {\rm wt}(\alpha) < n$. 
\end{lemma}

\begin{proof}
Suppose there exists some $\alpha$, where $0 < {\rm wt}(\alpha) < n$, such that $f^{\alpha}\neq 0$.
Since $\alpha$ is neither all-0 or all-1, and $\alpha$ has length at least $3$, we can find three bits in some order such that on these three bits, $\alpha$ takes value $001$ or $110$.
Without loss of generality, we assume they are the first three bits of $\alpha$ and we denote $\alpha$ by $001\delta$ or $110\delta$ ($\delta$ maybe empty). 
We first consider the case that $\alpha=001\delta$.
Consider another two strings $\beta=010\delta$ and $\gamma=100\delta$. 
Note that if we merge variables $x_1$ and $x_2$ of $f$, we get $\partial_{12}f$, 
its entry $(\partial_{12}f)^{0\delta}$ on the input $0\delta$ 
(for bit positions 3 to $n$) is the sum of $f^{010\delta}$ and $f^{100\delta}$. 
Since $\partial_{12}f\equiv0$, we have $$f^{010\delta}+f^{100\delta}=0.$$
Similarly, by merging variables $x_1$ and $x_3$, we have $$f^{001\delta}+f^{100\delta}=0,$$ and by merging variables $x_2$ and $x_3$, we have $$f^{001\delta}+f^{010\delta}=0.$$
These three equations have only a trivial solution, $f^{001\delta}=f^{010\delta}=f^{100\delta}=0$. A contradiction.
If $\alpha=110\delta$, the proof is symmetric.
\end{proof}
\begin{remark}
The above proof actually gives a stronger result.
\end{remark}
\begin{lemma}\label{zero_2}
Let $\widehat{f}$ be a signature of arity $n\geqslant 3$.
If $\widehat{f}^\alpha \neq 0$ for some $\alpha$
with $0< {\rm wt}(\alpha) < n$, then there is a pair of indices $\{i, j\}$ such that $(\widehat{\partial}_{ij}\widehat{f})^{\beta}\neq 0$ for some $\beta$, with ${\rm wt}(\beta)={\rm wt}(\alpha)-1$.

\end{lemma}




Recall that we can construct the pinning gadget when  $\Delta_0=(1, 0)$ is available. 
\begin{lemma}\label{all-zero}
Let $f$ be a signature of arity $n\geqslant 2$. If for any index $i$, by pinning the variable $x_i$ of $f$ to $0$, we have $f_i^0\equiv 0$, then $f^\alpha=0$ for any ${\rm wt}(\alpha)\neq n$. If furthermore,  there is a pair of indices $\{j, k\}$ such that $\partial_{jk}f\equiv 0$, then $f\equiv 0$.
\end{lemma}
\begin{proof}
For any ${\rm wt}(\alpha)\neq n$, there is an index $i$ such that $\alpha_i= 0$. 
By pinning $x_i$ to $0$, we get the signature $f_i^0$. We know $f^\alpha$ is an entry in $f_i^0$, and then $f^\alpha=0$ since $f_i^0 \equiv 0$.

Suppose there is a pair of indices $\{j, k\}$ such that $\partial_{jk}f\equiv 0$. 
Let $\beta$ denote the string of $n$ bits where $\beta_j=\beta_k=0$ and $\beta_\ell=1$ elsewhere,
and $\gamma$ denote the string of $n$ bits $1$s.
Consider the signature $\partial_{jk}f$. 
We know $f^\beta+f^\gamma$ is an entry in $\partial_{jk}f$ (when $\partial_{jk}f$ is a constant, we have $f^\beta+f^\gamma=\partial_{jk}f$). We know $f^\beta+f^\gamma=0$ since $\partial_{jk}f\equiv 0$.
Clearly, ${\rm wt}(\beta)\neq n$ and we have $f^\beta =0$. Thus, we have $f^\gamma=0$. Thus, we have $f\equiv 0$.
\end{proof}

A gadget construction often used in this paper is \emph{mating}. 
Given a real valued signature $f$ of arity $n\geqslant 2$, we connect two copies of $f$ in the following manner:
For any $m < n$,  fix a set $S$ of $n-m$ variables among all $n$ variables of $f$. For each $x_k\in S$, connect $x_k$ of one copy of $f$ with $x_k$ of the other copy using $=_2$. 
The variables 
that are not in $S$ are called dangling variables.
In this paper, we only consider the case that $m=1$ or $2$.
For $m=1$, there is one dangling variable $x_i$. Then, 
the mating construction  realizes a signature of arity $2$, denoted by $\frak m_{i}f$.
It can be represented by matrix multiplication. 
We have 
\begin{equation}\label{m-form}
M(\frak m_{i}f)=M_{i}(f)I_2^{\otimes (n-1)}M^{\tt T}_{i}(f)
=\begin{bmatrix}
{\bf {f}}^{0}_i\\
{\bf {f}}^{1}_i\\
\end{bmatrix}
\left[\begin{matrix}
{{\bf {f}}^{0}_i}^{\tt T} &{{\bf {f}}^{1}_i}^{\tt T}\\
\end{matrix}\right]
=\left[\begin{matrix}
|{\bf f}_i^0|^2 &  \langle {\bf f}_i^0, {\bf f}_i^1 \rangle\\
\langle {\bf f}_i^0, {\bf f}_i^1 \rangle & |{\bf f}_i^1|^2\\
\end{matrix}\right],
\end{equation}
where $\langle \cdot, \cdot\rangle$ denotes the (complex) inner product and $|\cdot|$ denotes the  norm defined by this inner product.
Note that $|\langle{\bf f}^{0}_i, {\bf f}^{1}_i\rangle|^2\leqslant|{\bf f}^{0}_i|^2|{\bf f}^{1}_i|^2$ by Cauchy-Schwarz inequality. 
Similarly, in the setting of $\holant{\neq_2}{\widehat{\mathcal{F}}}$, the above mating operation is equivalent to connecting variables in $S$ using $\neq_2$. We denote the resulting signature by  $\widehat{\frak{m}}_i\widehat{f}=\widehat{{\frak{m}}_i{f}}$ and we have
\begin{equation}\label{hat-form-intro}
M(\widehat{\frak{m}}_i\widehat{f})=
M_{i}(\widehat f)N_2^{\otimes n-1}M^{\tt T}_{i}(\widehat f)=
\left[\begin{matrix}
\widehat{{\bf f}}_i^0\\
\widehat{{\bf f}}_i^1
\end{matrix}\right]
\left[\begin{matrix}
0 & 1 \\
 1 & 0
\end{matrix}\right]^{\otimes (n-1)}
\left[\begin{matrix}
{\widehat{{\bf f}_i^{0}}}^{\tt T} &{\widehat{{\bf f}_i^{1}}}^{\tt T}
\end{matrix}\right].\\
\end{equation}
Note that $\widehat{f}$ satisfies the {\sc ars} since $f$ is real, we have \begin{equation*}
    N_2^{\otimes (n-1)}{\widehat{{\bf f}_i^{0}}}^{\tt T}=(\widehat f^{0,11\ldots1}, \widehat f^{0, 11\ldots0}, \ldots, \widehat f^{0, 00\ldots0})^{\tt T}\\=(\overline{\widehat f^{1,00\ldots0}}, \overline{\widehat f^{1,00\ldots1}}, \ldots, \overline{\widehat f^{1, 11\ldots1}})={\overline{\widehat{{\bf f}}_i^{1}}}^{\tt T}.
\end{equation*}
Thus, we have
\begin{equation}
M(\widehat{\frak{m}}_i\widehat{f})=
\left[\begin{matrix}
\widehat{{\bf f}}_i^0\\
\widehat{{\bf f}}_i^1
\end{matrix}\right]
\left[\begin{matrix}
0 & 1 \\
 1 & 0
\end{matrix}\right]^{\otimes (n-1)}
\left[\begin{matrix}
{\widehat{{\bf f}_i^{0}}}^{\tt T} &{\widehat{{\bf f}_i^{1}}}^{\tt T}
\end{matrix}\right]\\
=\left[\begin{matrix}
\widehat{{\bf f}}_i^0\\
\widehat{{\bf f}}_i^1
\end{matrix}\right]
\left[\begin{matrix}
{\overline{\widehat{{\bf f}}_i^{1}}}^{\tt T} &\overline{{\widehat{{\bf f}}_i^{0}}}^{\tt T}
\end{matrix}\right]
=\left[\begin{matrix}
\langle \widehat{{\bf f}}_i^0, \widehat{{\bf f}}_i^1 \rangle & |\widehat{{\bf f}}_i^0|^2\\
|\widehat{{\bf f}}_i^1|^2 & \langle \widehat{{\bf f}}_i^1, \widehat{{\bf f}}_i^0 \rangle
\end{matrix}\right].
\end{equation}

\subsection{Signature factorization}
In this subsection, we give more results for the factorization of real-valued signatures and signatures satisfying {\sc ars}. 
We say a signature $f$ is reducible  if $f = g \otimes h$,
up to a permutation of variables, for some signatures $g$ and $h$. All zero signatures 
(of arity greater than 1) are reducible.

The following lemma is an equivalent restatement of Lemma 2.15 in \cite{cai-fu-shao-eo}.
\begin{lemma}
Let $\widehat f$ be a reducible signature satisfying {\sc ars}, then there exists a factorization $\widehat f=\widehat g\otimes \widehat h$ such that $\widehat g$ and $\widehat h$ both satisfies {\sc ars}.

Let $f$ be a real valued reducible signature, then there exists a factorization $f=g\otimes h$ such that $g$ and $h$ are both real-valued signatures.
\end{lemma}

If a vertex $v$ in a signature grid 
is labeled by a reducible signature $f=g\otimes h$, we can replace the vertex $v$ 
by two vertices $v_1$ and $v_2$ and label $v_1$ with $g$ and $v_2$ with $h$, respectively.
The incident edges of $v$ become incident edges of $v_1$ and $v_2$ respectively
according to the partition of variables of $f$ in the tensor product of $g$ and $h$.  This does not
change the Holant value.
Clearly, $f=g\otimes h$ is realizable from $\{g, h\}$. On the other hand,  Lin and Wang proved in~\cite{wang-lin} 
(Lemma 3.3) that, from a reducible signature $f=g\otimes h\not\equiv 0$ we can freely replace $f$ by $g$ and $h$ while preserving the complexity of a  Holant problem.

\begin{lemma}\label{lin-wang}
If a nonzero real-valued signature $f$ has a real factorization $g\otimes h$, then $$\Holant(g, h, \mathcal{F})\equiv_T\Holant(f, \mathcal{F}) \text{ and } \holant{\neq_2}{\widehat{g}, \widehat{h}, \widehat{F}}\equiv_T \holant{\neq_2}{\widehat{f}, \widehat{F}}$$ for any signature set $\mathcal{F}$ $(\widehat{\mathcal{F}})$. In this case, we say $g$ $(\widehat{g})$ and $h$ $(\widehat{h})$ are realizable from $f$ $(\widehat{f})$ by factorization.
\end{lemma}

\subsection{Polynomial interpolation}\label{sec2.4}
Polynomial interpolation is a powerful technique to prove \#P-hardness for counting problems. We give the following lemmas. Please refer to \cite{jcbook} for proofs and more details. 

\begin{lemma}\label{2by2-interpolation}
Let $g_0$ and $g$ be two nonzero binary signatures with 
$M(g_0)=P^{-1}\left[\begin{smallmatrix}
0 & 0\\
0 & 1\\
\end{smallmatrix}\right]P$ and $M(g)=P^{-1}\left[\begin{smallmatrix}
\lambda_1 & 0\\
0 & \lambda_2\\
\end{smallmatrix}\right]P$ for some invertible matrix $P$.
If $\lambda_1\lambda_2\neq 0$ and $|\frac{\lambda_1}{\lambda_2}|\neq 1$,
then $$\Holant(g_0, \mathcal{F})\leqslant_T\Holant(g, \mathcal{F})$$ for any signature set $\mathcal{F}$.
\end{lemma}

\begin{lemma}\label{unary-inter}
Let $g$ be a nonzero binary signature with $M(g)=P^{-1}\left[\begin{smallmatrix}
\lambda_1 & 0\\
0 & \lambda_2\\
\end{smallmatrix}\right]P$ for some invertible matrix $P$, and $h$ be a nonzero unary signature. If $\lambda_1\lambda_2\neq 0$, $|\frac{\lambda_1}{\lambda_2}|\neq 1$, and $h$  (as a column vector) is not an eigenvector of $M(g)$, then $$\Holant(h', g, \mathcal{F})\leqslant_T\Holant(h, g, \mathcal{F})$$ for any unary signature $h'$ and any signature set $\mathcal{F}$.
\end{lemma}

\subsection{Tractable signatures}\label{sec-tractable-signatures}
We give some known signature sets 
that
 define polynomial time computable (tractable) counting problems.

\begin{definition}
Let $\mathscr{T}$ denote the set of tensor products of unary and binary signatures. 
\end{definition}

\begin{definition}
\label{definition-product-2}
 A signature on a set of variables $X$
 is of \emph{product type} if it can be expressed as a
product of unary functions,
 binary equality functions $(=_2)$,
and binary disequality functions $(\neq_2)$, each on one or two
variables of $X$.
 We use $\mathscr{P}$ to denote the set of product-type functions.
\end{definition}

\begin{definition}\label{definition-affine}
 A signature $f(x_1, \ldots, x_n)$ of arity $n$
is \emph{affine} if it has the form
 \[
  \lambda \cdot \chi_{A X = 0} \cdot {\frak i} ^{Q(X)},
 \]
 where $\lambda \in \mathbb{C}$,
 $X = (x_1, x_2, \dotsc, x_n, 1)$,
 $A$ is a matrix over $\mathbb{Z}_2$,
 $Q(x_1, x_2, \ldots, x_n)\in \mathbb{Z}_4[x_1, x_2, \ldots, x_n]$
is a quadratic (total degree at most 2) multilinear polynomial
 with the additional requirement that the coefficients of all
 cross terms are even, i.e., $Q$ has the form
 \[Q(x_1, x_2, \ldots, x_n)=a_0+\displaystyle\sum_{k=1}^na_kx_k+\displaystyle\sum_{1\leq i<j\leq n}2b_{ij}x_ix_j,\]
 and $\chi$ is a 0-1 indicator function
 such that $\chi_{AX = 0}$ is~$1$ iff $A X = 0$.
 We use $\mathscr{A}$ to denote the set of all affine signatures.
\end{definition}
If the support set $\mathscr{S}(f)$ is an affine linear subspace, then we say $f$ has affine support. Clearly, any affine signature has affine support.
Moreover,  we have that 
any signature of product type has affine support.

Let $T_{\alpha^s}=\left[\begin{smallmatrix}
1 & 0\\
0 & \alpha^s\\
\end{smallmatrix}\right]$ where $\alpha=\frac{1+\ii}{\sqrt{2}}$.

\begin{definition}
 A signature $f$ is local-affine if for each $\sigma=s_1s_2\ldots s_n \in \{0, 1\}^{n}$ in the support of $f$, $(T_{\alpha^{s_1}}\otimes T_{\alpha^{s_2}}\otimes\cdots \otimes T_{\alpha^{s_n}})f\in \mathscr{A}$. We use $\mathscr{L}$ to denote the set of local-affine signatures. 
\end{definition}

\begin{definition} \label{def:prelim:trans}
 We say a signature set $\mathcal{F}$ is $\mathscr{C}$-transformable
 if there exists a $T \in \rm{GL}_2(\mathbb{C})$ such that
 $(=_2)(T^{-1})^{\otimes 2} \in \mathscr{C}$ and $T\mathcal{F} \subseteq \mathscr{C}$. 
\end{definition}

This definition is important because if $\Holant(\mathscr{C})$ is tractable,
then $\Holant(\mathcal{F})$ is tractable for any $\mathscr{C}$-transformable set $\mathcal{F}$. Then, the following tractable result is known.
\begin{theorem}\label{main-thr}
For any set of complex valued signatures  $\mathcal{F}$,
 $\Holant(\mathcal{F})$ is P-time computable
if
\begin{itemize}
    \item 
$\mathcal{F}\subseteq \mathscr{T}$, or
\item $\mathcal{F}$ is $\mathscr{P}$-transformable, or
\item $\mathcal{F}$ is $\mathscr{A}$-transformable, or
\item $\mathcal{F}$ is $\mathscr{L}$-transformable.
\end{itemize}
We denote the above tractable conditions by conditions (\ref{main-thr}).
\end{theorem}




Let $H=\frac{1}{\sqrt{2}}\left[\begin{smallmatrix}
1 & 1\\
1 & -1\\
\end{smallmatrix}\right]$ be the 2-by-2 Hadamard matrix. We have the following result.

\begin{lemma}\label{hard-sign}
Let $\mathcal{F}$ be a set of real valued signatures and $\mathcal{F}$ does not satisfy conditions (\ref{main-thr}). Then for any orthogonal matrix $Q\in \mathbb{R}^{2\times2}$,  $Q\mathcal{F}$ also does not satisfy conditions (\ref{main-thr}). Moreover, we have $H\mathcal{F}
\not\subseteq \mathscr{P}, \widehat{\mathcal{F}}\not\subseteq \mathscr{P}$, $\widehat{\mathcal{F}}\not\subseteq\mathscr{A}$ and $T_{\alpha}\mathcal{F}\not\subseteq \mathscr{A}$.
\end{lemma}

\subsection{Known hardness results}\label{sec-hard-result}
Let $\mathcal F$ by a set of real valued signatures. The following three results are known \cite{clx-csp, CLX-HOLANTC, Backens-Holant-c}.
\begin{theorem}\label{csp-dic}
 $\CSP(\mathcal{F})$ is \#P-hard unless $\mathcal{F} \subseteq \mathscr{A}$ or $\mathscr{P}$. 
\end{theorem}

\begin{theorem}\label{even-csp-dic}
$\CSP_2(\mathcal{F})$ is \#P-hard unless $\mathcal{F} \subseteq \mathscr{A}$, $\mathscr{P}$, $\mathscr{L}$, or $T_{\alpha}\mathcal{F} \subseteq \mathscr{A}$
\end{theorem}

\begin{theorem}\label{dic-holantc}
$\Holant^c(\mathcal F)$ is \#P-hard unless  $\mathcal{F} \subseteq \mathscr{T}$, $ \mathscr{A}$, $\mathscr{P}$, $\mathscr{L}$, $H\mathcal{F} \subseteq \mathscr{P}$,   $\widehat{\mathcal{F}} \subseteq \mathscr{P}$ or $T_{\alpha}\mathcal{F} \subseteq \mathscr{A}$.
\end{theorem}

    Based on  the above three theorems and Lemma \ref{hard-sign}, we have the following hardness result.

\begin{theorem}\label{hard-result}
Let $\mathcal{F}$ be a set of real valued signatures and $\mathcal{F}$ does not satisfy conditions (\ref{main-thr}). Then for any orthogonal matrix $Q\in \mathbb{R}^{2\times 2}$, we have $\CSP(Q\mathcal{F})$, $\CSP_2(Q\mathcal{F})$ and $\Holant^c(Q\mathcal F)$ are \#P-hard.
\end{theorem}

The following two reductions are known \cite{jcbook}. One states that we can realize all $=_k$ once we have $=_3$. The other states that we can realize all $=_{2k}$ once we have $=_4$.
\begin{lemma}\label{all-equ}
$\CSP(\mathcal{F})\leqslant_T\Holant(=_3, \mathcal F)$.
\end{lemma}

\begin{lemma}\label{all-even-equ}
$\CSP_2(\mathcal{F})\leqslant_T\Holant(=_4, \mathcal F)$.
\end{lemma}

In the following, without other specifications, we use $f$ to denote a real-valued signature and $\mathcal{F}$ to denote a set of real-valued signatures. While $\widehat{f}$ denotes a signature satisfying {\sc ars} and $\widehat{\mathcal{F}}$ denotes a set of such signatures. We use $Q\in \mathbb{R}^{2\times 2}$ to denote an orthogonal matrix and $\widehat{Q}$ denotes $Z^{-1}Q(Z^{-1})^{\tt T}$.
\section{Proof Outline for Theorem \ref{odd-intro-sec1}}\label{sec-outline}

We first give a more concrete restatement of
 Theorem \ref{odd-intro-sec1}.
 
 \begin{theorem}\label{odd-dic}
Let $\mathcal{F}$ be a set of real-valued signatures containing a (nonzero) signature of odd arity.
If $\mathcal F$ satisfies the 
tractability conditions stated in Theorem~\ref{main-thr} then $\Holant(\mathcal F)$  is P-time
computable.
Otherwise, 
$\Holant(\mathcal F)$ is {\rm \#}P-hard.
\end{theorem}

By Theorem \ref{main-thr}, the tractability part is known. 
We prove  \#P-hardness when  $\mathcal{F}$ does not satisfy these conditions.
First, we show that under some holographic transformations, either one can use a signature of odd arity in $\mathcal{F}$ to realize the unary signature $\Delta_0=(1, 0)$, or one can realize some equality signature $(=_k)$ $(k\geqslant 3)$. 
Then, we prove \#P-hardness results for $\Holant(\neq_2 \mid =_k, \widehat{\mathcal{F}})$ (Corollary \ref{cspk-dic})  and  $\Holant(\Delta_0, \mathcal{F})$ separately (Theorem \ref{01-dic}).

It is easy to verify that for any $k\geqslant 3$, $\CSP_k(\neq_2, \widehat{\mathcal{F}})\leqslant_T \Holant(\neq_2 \mid =_k, \widehat{\mathcal{F}})$. 
Thus, we prove the \#P-hardness of $\Holant(\neq_2 \mid =_k, \widehat{\mathcal{F}})$
by giving a dichotomy of  $\CSP_k(\neq_2, {\mathcal{G}})$ for any set  $\mathcal{G}$ of complex-valued signatures  (Theorem \ref{cspk-dic-thr}).  
This result should be of independent interest.

 Next, we focus on the \#P-hardness of  $\Holant(\Delta_0, \mathcal{F})$. 
 This is where our entanglement result can be applied. 
If $\Delta_1$ is realizable from $\Holant(\Delta_0, \mathcal{F})$, then we reduce $\Holant(\Delta_0, \mathcal F)$ from $\Holant^c( \mathcal F)$ and we are done by the 
dichotomy of $\Holant^c( \mathcal F)$.
By using $\Delta_0$, we first give two conditions that $\Delta_1$ can be easily realized  by pinning (Lemma \ref{01-by-pin}) or interpolation (Lemma \ref{01-by-inter}). 
As a corollary (Corollary \ref{first-ortho}) of Lemma \ref{01-by-inter}, 
we show that  either $\Holant^c(\mathcal F)\leqslant_T\Holant(\Delta_0, \mathcal F)$, or every irreducible $f\in \mathcal{F}$ satisfies 
the following important \emph{first order orthogonality} condition.

\begin{definition}[First order orthogonality, c.f. Definition \ref{def-first-order-orth}]
       Let $f$ be a complex-valued signature of arity $n \geqslant 2$, we say that it satisfies the \emph{first order orthogonality} condition if there exists some $\mu\neq 0$ such that for all indices $i\in [n]$, the entries of $f$ satisfy the following equations
\begin{equation*}
    |{\bf f}_{i}^{0}|^2=|{\bf f}_{i}^{1}|^2=\mu, \text{ and } \langle{\bf f}_{i}^{0}, {\bf f}_{i}^{1}\rangle =0.
\end{equation*}
To  restate it in the quantum terminology,  let $|\Psi\rangle$ be a normalized $n$-qubit $(n\geqslant 2)$ state, i.e., $\langle\Psi|\Psi\rangle=1$. Then it satisfies the first order orthogonality if for every $i$-th qubit of $|\Psi\rangle$, $\langle\Psi_i^0|\Psi_i^0\rangle=\langle\Psi_i^1|\Psi_i^1\rangle=1/2$ and $\langle\Psi_i^0|\Psi_i^1\rangle=0$.
    \end{definition}
    
     \begin{remark}
     A real-valued signature $f$ satisfies the first order orthogonality precisely when there is some $\mu \not =0$ such that for all indices $i$, $M(\frak{m}_if)= \mu I_2$.
     In this case, the complex inner product is represented by mating real-valued signatures using $(=_2)$.
     Also, by  (\ref{hat-form-intro}), when $\widehat{f}$ is a signature with {\sc ars},  the complex inner product can also be represented by mating using $(\neq_2)$.
We can show that $f$  satisfies first order orthogonality if and only if $\widehat{f}$ satisfies it. 
Although  first order orthogonality is well-defined for any complex-valued signature, when $f$ is not real-valued ($\widehat{f}$ does not satisfy {\sc ars}), we can not say anything about $\frak{m}_if$ and $\widehat{\frak{m}}_i\widehat{f}$ since now they can not represent the complex inner product. The properties of $\frak{m}_if$ and $\widehat{\frak{m}}_i\widehat{f}$ crucially depend on $f$ being real (equivalently $\widehat{f}$ satisfying {\sc ars}).
\end{remark}

Back to the proof of the \#P-hardness of $\Holant(\Delta_0, \mathcal F)$.
 Since  $\mathcal F$ does not satisfy conditions (\ref{main-thr}), 
$\mathcal F\not\subseteq\mathscr{T}$. Hence, there is a signature $f \in \mathcal{F}$ of arity $n\geqslant 3$ such that $f\notin \mathscr{T}$. 
In other words, $\mathcal{F}$ contains an $n$-qubit state exhibiting multipartite entanglement. 
We will prove the \#P-hardness by induction on $n$.
We first consider the base case that $n=3$. 
We show that an irreducible ternary signature (a genuinely entangled 3-qubit state) satisfying first order orthogonality has some special forms, from which one can realize $(=_3)$ or $(=_4)$ after some holographic transformations.
Then, we can reduce the problem from $\CSP(\mathcal F)$, or $\CSP_2(\mathcal F)$, or 
$\Holant(\neq_2 \mid =_3, \widehat{\mathcal{F}})$,
to  $\Holant(\Delta_0, \mathcal F)$ (Lemma \ref{01-base-3}).
This allows us to finish the proof by citing existing dichotomy results for  $\CSP(\mathcal F)$, or $\CSP_2(\mathcal F)$, or the result we showed above for $\Holant(\neq_2 \mid =_k, \widehat{\mathcal{F}})$.

Then, we consider the inductive step.
The general strategy is that we start with a signature $f\in \mathcal{F}$ of arity $n\geqslant 4$ that is not in $\mathscr{T}$, and realize a signature $g$ of arity $n-1$ or $n-2$ 
also not in $\mathscr{T}$, by pinning or merging.
(By the definition of  $\mathscr{T}$, when $n=4$, this $g$ must have arity 3.)
By a sequence of reductions (that is constant in length independent of the problem instance size), we can realize a signature $h$ of arity $3$ that is not in $\mathscr{T}$. Then we are done. In other words, given an $n$-qubit state with multipartite entanglement,
we want to show that multipartite entanglement is preserved under projections onto $|0\rangle$ and self-loops by $|\phi^{+}\rangle$.
Lemma \ref{01-induction-intro} says that the preservation holds, or $|1\rangle$ or $|{\rm GHZ_4}\rangle$ is realizable (Lemmas \ref{01-induction} and \ref{unary-n=4}). 
We give an inductive restatement of Lemma \ref{01-induction-intro} in the  Holant framework.
 \begin{lemma}
Let $f\in \mathcal{F}$ be a signature of arity $n\geqslant 4$ and $f\notin \mathscr{T}$. Then one of the following alternatives must hold:
\begin{itemize}
    \item $\Delta_1$ is realizable: 
$\Holant(\Delta_0, \Delta_1, \mathcal F)\leqslant_T \Holant(\Delta_0, \mathcal F)$, or
\item $(=_4)$ is realizable:  
$\Holant(=_4, \mathcal F)\leqslant_T 
\Holant(\Delta_0, \mathcal F)$, or
\item a signature $g\notin \mathscr T$ of arity $n-1$ or $n-2$ can be realized: $\Holant(\Delta_0, g, \mathcal{F})\leqslant_T \Holant(\Delta_0, \mathcal{F})$.
\end{itemize}
\end{lemma}
%


\noindent\emph{Proof Sketch.} 
For all indices $i$ and all pairs of indices $\{j, k\}$, consider $f_i^0$ and $\partial_{jk}f$.  If there exists $i$ or $\{j, k\}$ such that $f_i^0$ or $\partial_{jk}f \notin \mathscr{T}$, 
then we can realize
$g=f_i^0$ or $\partial_{jk}f$ which has arity $n-1$ or $n-2$, and we are done.
Otherwise, $f_i^0$ and $\partial_{jk}f \in \mathscr T$ for all $i$ and all $\{j, k\}$. 
 Under this assumption,
 our goal is to show that we can realize $\Delta_1$, 
 or there is a unary signature $a(x_u)$ or a binary signature $b(x_v, x_w)$ such that $a(x_u) \mid f$
or $b(x_v, x_w) \mid f$.
Then, we have $f=a(x_u)\otimes g$ or $f=b(x_v, x_w) \otimes g$ for some $g$ of arity $n-1$ or $n-2$. 
We know $g$ can be realized from $f$ by factorization. 
By the definition of $\mathscr T$, we have $g \notin \mathscr T$ since $f \notin \mathscr T$, and we are done.
When $n\geqslant 5$, the above induction proof can be achieved by the interplay of the unique prime factorization, and the commutivity of $f_i^0$ (pinning) and $\partial_{jk}f$ (merging) gadgets 
on disjoint indices.
(Lemmas \ref{unary-factor} and \ref{binary-factor}).
For $n=4$, there is the additional case that $(=_4)$ can be realized. Thus for $n=4$,
it requires more work  (Lemma \ref{unary-n=4}); we need to combine the induction proof and  first order orthogonality to handle it. \qed

We know that $\Holant^c(\mathcal{F})$ (which is
just $\Holant(\Delta_0, \Delta_1, \mathcal F))$
and $\CSP_2(\mathcal{F})$ are both \#P-hard when $\mathcal F$ does not satisfy the tractability conditions (\ref{main-thr}). Thus, we have shown that when  $\mathcal F$ does not satisfy the tractability conditions (\ref{main-thr}),
$\Holant(\Delta_0, \mathcal F)$ is {\rm \#}P-hard.
This finishes the proof of Theorem \ref{odd-intro-sec1}.

\section{Full Proof}\label{full-proof}
We give the full proof of Theorem \ref{odd-dic} in this section.

\begin{lemma}\label{odd-red}
Let $\mathcal{F}$ be a set of real-valued signatures containing a signature of odd arity.
Then, there exists some orthogonal matrix $Q$ such that 
\begin{itemize}
\item $\Holant(\Delta_0, Q\mathcal{F})\leqslant_T\Holant(\mathcal{F})$ or
\item $\holant{\neq_2}{=_{2k+1}, \widehat{Q}\widehat{\mathcal F}}\leqslant_T\Holant(\mathcal{F})$, for some $k \geqslant 1$.
\end{itemize}
\end{lemma}
\begin{proof}
Suppose $f$ has arity $n$. We prove our lemma by induction on $n$.

If $n=1$, let $f=(a, b)$ where $a, b$ are not both zero. Let $Q_1=\frac{1}{\sqrt{a^2+b^2}}\left[
\begin{smallmatrix}
a & b\\
-b & a\\
\end{smallmatrix}\right]$, an   orthogonal matrix.
Note that $\Holant(\mathcal{F})$ is just $ \holant{=_2}{\mathcal{F}}$, and $=_2$ is invariant under 
an orthogonal holographic transformation $(=_2) (Q_1^{-1})^{\otimes 2} = (=_2)$, and
$Q_1(a, b)^{\tt T}=\sqrt{a^2+b^2}(1, 0)^{\tt T}.$ Thus, 
$$\holant{=_2}{(a, b), \mathcal F}\equiv_T \holant{=_2}{\Delta_0, Q_1\mathcal F}.$$
The base case is proved.

We assume our claim is true for $n=2k-1$. Now, we consider $n=2k+1\geqslant 3$. 
If there is a pair of indices $\{i, j\}$ such that $\partial_{(ij)}f\not\equiv 0$, then we can realize a signature of arity $2k-1$ from $f$. 
By induction hypothesis, we have 
$$\Holant(\Delta_0, Q\mathcal{F})\leqslant_T\Holant(\partial_{ij}f, \mathcal F)\leqslant_T \Holant(\mathcal F).$$ 
Otherwise, $\partial_{ij}f\equiv 0$ for all pairs of indices $\{i, j\}$. 
Thus, we also have
$\widehat{\partial}_{ij}\widehat{f}\equiv 0$ for all $\{i, j\}$.
Then, by Lemma \ref{zero}, we have $\widehat{f}=a(1, 0)^{\otimes n}+\bar{a}(0, 1)^{\otimes n}$ for some $a\neq 0$.
Let $\widehat{Q_2}=\left[
\begin{smallmatrix}
\sqrt[n]{\bar a} & 0\\
0 & \sqrt[n]{ a}
\end{smallmatrix}\right]$.
$Q_2$ is orthogonal (up to a scalar) by  Lemma~\ref{q-parity}.
We have $(\widehat{Q_2})^{\otimes n}\widehat f= |a|^2 \big( (1, 0)^{\otimes n}+(0, 1)^{\otimes n}\big)$.
Thus, 
a holographic transformation by $\widehat{Q_2}$ and $Z^{-1}$ yields 
$$\Holant({\mathcal F}) = \holant{=_2}{ f, {\mathcal F}}\equiv_T\holant{\neq_2}{\widehat f, \widehat{\mathcal F}}\equiv_T \holant{\neq_2}{=_{2k+1}, \widehat{Q_2}\widehat{\mathcal F}}.$$
Thus, $\holant{\neq_2}{=_{2k+1}, \widehat{Q}\widehat{\mathcal F}}\leqslant_T\Holant(\mathcal{F})$ where $k\geqslant 1$.
%
\end{proof}
 
 We will prove \#P-hardness results for $\Holant(\neq_2 \mid =_k, \widehat{\mathcal{F}})$ and  $\Holant(\Delta_0, \mathcal{F})$ in the following subsections.

\subsection{The hardness of Holant$(\neq \mid =_k, \widehat{\mathcal{F}})$}\label{sec4.1}

Recall that $\mathcal{EQ}_k$ denotes the set of equalities of arity $nk$ for all $n\geqslant 1$, i.e., $\mathcal{EQ}_k=\{=_k, =_{2k}, \ldots, =_{nk}, \ldots\}$. The  problem $\CSP_k(\neq_2, \mathcal{F})$ is defined as
$\Holant(\mathcal{EQ}_k \mid  \neq_2, {\mathcal{F}})$.
This is equivalent to $\Holant(\mathcal{EQ}_k, \neq_2 \mid \mathcal{EQ}_k, \neq_2, {\mathcal{F}})$.
First, we have the following reduction.

\begin{lemma}\label{all-k-equ}
For any $k\geqslant 3$, $\CSP_k(\neq_2, {\mathcal{F}})\leqslant_T \Holant(\neq_2 \mid =_k, {\mathcal{F}})$.
\end{lemma}
\begin{proof}
Connecting one input each  of $k$ copies of $(\neq_2)$ on the LHS with the $k$
variables of $(=_k)$ on the RHS produces one $(=_k)$ on the LHS. Once we have
$(=_{(n-1)k})$ on the LHS, we take one copy of $(=_{(n-1)k})$  and two copies of $(=_k)$
on the LHS, and one copy of $(=_k)$ on the RHS with variables $x_1, \ldots, x_k$. Connect 
$x_1$ and $x_2$ to one variable each of the two copies of $(=_k)$
on the LHS, and $x_3, \ldots, x_k$ all to variables of $(=_{(n-1)k})$.
This produces  one copy of $(=_{nk})$ on the LHS.
\end{proof}

Then, we give a dichotomy of  $\CSP_k(\neq_2, {\mathcal{F}})$ for any set  $\mathcal{F}$ of complex-valued signatures.
Let 
$\rho$ be a $4k$-th primitive root of unity,
 $T_k=\left[\begin{smallmatrix}1 & 0\\ 0 & \rho\end{smallmatrix}\right]$,
and $\mathscr{A}_k^d=\{f~|~T^d_kf\in\mathscr{A}\}$,  for $d\in[k]$.

\begin{theorem}\label{cspk-dic-thr}
$\operatorname{\#CSP}_k(\mathcal{F}, \neq_2)$ is $\#\operatorname{P}$-hard except for the following cases
\begin{itemize}
\item $\mathcal{F}\subseteq\mathscr{P}$;
\item there exists $k\in[d]$ such that $\mathcal{F}\subseteq\mathscr{A}_k^d$,
\end{itemize}
which can be computed in polynomial time.
\end{theorem}
\begin{proof}
We will give a proof in Section \ref{sec-cspd}. This result should be of independent interest.
\end{proof}

\begin{remark}
Let $\mathcal{F}$ be a set of real-valued signatures. It is easy to see that when $\mathcal{F}$ does not satisfy the conditions in Theorem~\ref{main-thr}, $\widehat{\mathcal{F}}$ does not satisfy the tractable conditions of Theorem \ref{cspk-dic-thr}, therefore $\CSP_k(\neq_2, \widehat{\mathcal{F}})$ is \#P-hard. 
Furthermore,  we consider $\widehat{Q}\widehat{\mathcal{F}}$ for any orthogonal matrix $Q\in \mathbb{R}^{2\times 2}$.
Let $\widehat{f}\in \widehat{\mathcal{F}}$ be a signature of arity $n$.
By  definition, we have $$\widehat{Q}\widehat{f}=[Z^{-1}Q(Z^{-1})^{\tt T}]^{\otimes n}[(Z^{-1})^{\otimes n}f]=(Z^{-1})^{\otimes n}[Q(Z^{-1})^{\tt T}Z^{-1}]^{\otimes n} f.$$
Note that $2(Z^{-1})^{\tt T}Z^{-1}=\left[\begin{smallmatrix}
1 & 0\\
0 & -1\\
\end{smallmatrix}\right]$. Thus,  $2Q(Z^{-1})^{\tt T}Z^{-1}$ is also a real orthogonal matrix, denoted by $Q_{-1}$. Therefore, up to a scalar $\widehat{Q}\widehat{f}=(Z^{-1})^{\otimes n}(Q_{-1}^{\otimes n}f)=\widehat{Q_{-1}f}$. 
Thus, $\widehat{Q}\widehat{\mathcal{F}}=\widehat{Q_{-1}\mathcal{F}}$.
By Lemma \ref{hard-sign},  $Q_{-1}\mathcal{F}$ does not satisfy conditions (\ref{main-thr}) when $\mathcal{F}$ does not satisfy conditions (\ref{main-thr}), hence
$\CSP_k(\neq_2, \widehat{Q_{-1}\mathcal{F}})$ is \#P-hard.
Therefore $\CSP_k(\neq_2, \widehat{Q}\widehat{\mathcal{F}})$ is   \#P-hard.
\end{remark}

Combine the above result and Lemma \ref{all-k-equ}, we have the following result.
\begin{corollary}\label{cspk-dic}
Let $\mathcal{F}$ be a set of real-valued signatures.
For any $k\geqslant 3$ and any orthogonal matrix $Q$, 
$\Holant(\neq_2 \mid =_k, \widehat{Q}{\widehat{\mathcal{F}}})$ is \#P-hard when $\mathcal{F}$ does not satisfy the conditions in Theorem~\ref{main-thr}.
\end{corollary}

\subsection{The hardness of Holant($\Delta_0, \mathcal F$): the base case}
In the following two subsections, 
we will show that $\Holant(\Delta_0, \mathcal F)$ is {\rm \#}P-hard if $\mathcal F$ does not satisfy conditions (\ref{main-thr}). Since  
$\mathcal F\not\subseteq\mathscr{T}$, there is a signature $f \in \mathcal{F}$ of arity $n\geqslant 3$ that is not in $\mathscr{T}$. 
Recall that $f \notin \mathscr{T}$ means that the state of $|f\rangle$ has multipartite entanglement.
We will prove our claim by induction on the arity $n$.

In this subsection, we focus on the base case that $f$ has arity $3$. In other words, there is a multipartite entangled $3$-qubit.
We first give two conditions that $\Delta_1 = (0, 1)$ can be easily realized from a signature of arity not necessarily $3$ by pinning (Lemma \ref{01-by-pin}) or interpolation (Lemma \ref{01-by-inter}). 
Then, we have $\Holant^c( \mathcal F)\leqslant_T\Holant(\Delta_0, \mathcal F)$,  and we are done by the hardness of $\Holant^c( \mathcal F)$.
For cases in which $\Delta_1$ cannot be realized from a ternary signature $f$, we show that $f$ has special forms. Then, we can use this $f$ to realize $=_3$ or $=_4$, and reduce the problem from $\CSP(Q\mathcal F)$, $\CSP_2(Q\mathcal F)$ or $\CSP_3(\neq_2, \widehat{Q}\widehat{\mathcal F})$ for some orthogonal matrix $Q$. Based on the known hardness results, we will finish the proof for the base case that $f$ has arity $3$ (Lemma \ref{01-base-3}). 

\begin{lemma}\label{01-by-pin}
Let $f\in \mathcal{F}$ be a nonzero signature and $f^{\vec{0}}=0$. Then $\Holant^c(\mathcal F)\leqslant_T\Holant(\Delta_0, \mathcal F)$. 
\end{lemma}
\begin{proof}
We prove this by induction on the arity $n$ of $f$.

If $n=1$,  we have $f=(0, \lambda)$ for  some $\lambda\neq 0$ since $f\not\equiv 0$. Clearly, $\Delta_1$ is realizable from $f$.

Assuming our claim is true when $n=k$, 
we consider the case that $n=k+1$. 
For all indices $i\in [n]$, consider signatures $f_i^0$ realized from $f$ by pinning variable $x_i$ to $0$. We know $f_i^0$ is signature of arity $k$ and $f_i^0(\vec{0}_{k})=f(\vec{0}_{k+1})=0$.
\begin{itemize}
    \item 
If there is an index $i$ such that $f_i^0\not\equiv 0$, then by induction hypothesis,
we have 
$\Holant^c(\mathcal F)\leqslant_T\Holant(\Delta_0, f_i^0, \mathcal F)\leqslant_T\Holant(\Delta_0, \mathcal F).$
\item Otherwise, $f_i^0\equiv 0$   for all indices $i$. Then, by Lemma \ref{all-zero}, we have $f=\lambda(0, 1)^{\otimes n}$ for some $\lambda\neq 0$ since $f\not\equiv 0$. Thus, $\Delta_1$ is realizable from $f$ by Lemma \ref{lin-wang}.
\end{itemize}
Thus, we have 
$\Holant^c(\mathcal F)\leqslant_T\Holant(\Delta_0, \mathcal F).$
\end{proof}

Now for all indices $i$, we consider signatures $\frak m_{i}f$ realized from $f$ by mating. We give a condition by which $\Delta_1$ can be realized from $\frak m_{i}f$ by interpolation. 

\begin{lemma}\label{01-by-inter}
Let $f\in \mathcal{F}$ be a nonzero signature of arity $n\geqslant 2$. If there is an index $i$ such that $M(\frak m_{i}f)$ (as a 2-by-2 matrix) is not the identity matrix up to a scalar ($M(\frak m_{i}f)\neq \mu_i I_2$), then 
\begin{itemize}
    \item there is an unary signature $a(x_i)$ on variable $x_i$ such $a(x_i)\mid f$, or
\item $\Holant^c(\mathcal F)\leqslant_T\Holant(\Delta_0, \mathcal F)$. 
\end{itemize}
\end{lemma}
\begin{proof}
We denote $$M(\frak m_{i}f)=\left[\begin{matrix}
|{\bf f}_i^0|^2 &  \langle {\bf f}_i^0, {\bf f}_i^1 \rangle\\
\langle {\bf f}_i^0, {\bf f}_i^1 \rangle & |{\bf f}_i^1|^2
\end{matrix}\right] ~~~~\text{ by }~~~~
\left[\begin{matrix}
a & b \\
b & c\\
\end{matrix}\right].$$
Since $f$ is real, $M(\frak m_{i}f)$ is real symmetric, and thus diagonalizable with real eigenvalues.
We first consider the case that
 $M(\frak m_{i}f)$ is degenerate.
Then, we have 
$|\langle {\bf f}_i^0, {\bf f}_i^1 \rangle|^2=|{\bf f}_i^0|^2|{\bf f}_i^1|^2$, so  ${\bf f}_i^0$ and ${\bf f}_i^1$ are linearly dependent by Cauchy-Schwarz.
Since $f\not\equiv 0$, either  ${\bf f}_i^0$ and ${\bf f}_i^1$  is nonzero. 
Assume  ${\bf f}_i^0$ is nonzero  (the other case is similar). Then, we have   ${\bf f}_i^1=c \cdot {\bf f}_i^0$ for some constant $c$. 
It follows that   $f=a(x_i)\otimes  {\bf f}_i^0$, for a unary signature
$a(x_i)=(1, c)$. 

Now we assume $M(\frak m_{i}f)$ has  rank 2,  then we have $a, c>0$.
We consider the value of $b$. 
\begin{itemize}
    \item
If $b=0$, then $M(\frak m_{i}f)=\left[\begin{smallmatrix}
a & 0 \\
0 & c\\
\end{smallmatrix}\right]$. Clearly, $a\neq c$ since $M(\frak m_{i}f)$ is not $I_2$ up to a scalar.  Given $a\neq c$ and $\frac a c >0$, we have $|\frac a c|\neq  1$. By Lemma \ref{2by2-interpolation}, we can realize $(0, 0, 0, 1)=(0, 1)^{\otimes 2}$ 
from $\frak m_{i}f$ by interpolation. 
Then, by Lemma \ref{lin-wang}, we can realize $\Delta_1=(0, 1)$ by factorization.

\item Otherwise, $b\neq 0$. Clearly, we know $(1, 0)^{\tt T}$ is not an eigenvector of $M(\frak m_{i}f)$. Suppose $M(\frak m_{i}f)=P^{-1}\left[\begin{smallmatrix}
\lambda_1 & 0\\
0 & \lambda_2\\
\end{smallmatrix}\right]P$, where $\lambda_1$ and $\lambda_2$ are two real  eigenvalues of $M(\frak m_{i}f)$.
Since $M(\frak m_{i}f)$ has  rank 2 and $M(\frak m_{i}f)$ is not $I_2$ up to a scalar, we have $\lambda_1\lambda_2\neq 0$ and $\lambda_1\neq\lambda_2$.
Also, 
by the trace formula, $\lambda_1+\lambda_2=a+c> 0$. Thus $\frac{\lambda_1}{\lambda_2}\neq -1$. Then we have $|\frac{\lambda_1}{\lambda_2}|\neq 1.$
By Lemma \ref{unary-inter}, we can realize $\Delta_1=(0, 1)$ by interpolation.
\end{itemize}
Thus, we have $\Holant^c(\mathcal F)\leqslant_T\Holant(\Delta_0, \mathcal F)$.
\end{proof}

 Moreover, if for all indices $i$, $M(\frak m_{i}f)=\mu_i I_2$,
 then we show all $\mu_i$ have the same value.

\begin{corollary}\label{first-ortho}
Let $f\in \mathcal{F}$ be an irreducible signature of arity $n\geqslant 2$. Then we have 
\begin{itemize}
    \item $\Holant^c(\mathcal F)\leqslant_T\Holant(\Delta_0, \mathcal F)$, or
    \item there exists some $\mu\neq 0$ such that for all indices $i$, $M(\frak{m}_if)= \mu I_2$ i.e., $\frak{m}_if=\mu \cdot (=_2)$.
    \end{itemize}
\end{corollary}
\begin{proof}
In the above proof, 
if $f$ is irreducible, either  $\Holant^c(\mathcal F)\leqslant_T\Holant(\Delta_0, \mathcal F)$, or
  $M(\frak m_{i}f)=\mu_i\left[\begin{smallmatrix}
1 & 0\\
0 & 1\\
\end{smallmatrix}\right]$ for every index $i$, where $\mu_i=|{\bf f}_i^0|^2=|{\bf f}_i^1|^2 \not = 0.$
    If we  further connect the two dangling variables $x_i$ of $\frak m_{i}f$, which totally connects the corresponding pairs of variables in  two copies of $f$, we get a value $2\mu_i$. 
This value  does not depend on the particular index $i$.
Thus, all $\mu_{i}$ have the same value for $i\in [n]$.
    We denote this value by $\mu$.
    \end{proof}
    
    We introduce the following key property, which we call \emph{first order orthogonality}.
    \begin{definition}[First order orthogonality]\label{def-first-order-orth}
       Let $f$ be a complex-valued signature of arity $n \geqslant 2$, we say it satisfies  first order orthogonality if there exists some $\mu\neq 0$ such that for all indices $i\in [n]$, the entries of $f$ satisfy the following equations
\begin{equation}\label{ee1}
    |{\bf f}_{i}^{0}|^2=|{\bf f}_{i}^{1}|^2=\mu, \text{ and } \langle{\bf f}_{i}^{0}, {\bf f}_{i}^{1}\rangle =0.
\end{equation}
    \end{definition}
    \begin{remark}
    When $f$ is a real-valued signature, it satisfies  first order orthogonality is precisely that there is $\mu \not =0$ such that for all indices $i$, $M(\frak{m}_if)= \mu I_2$. 
Now, we consider $\widehat{f}$. Since $\widehat{(=_2)}=Z^{-1}(=_2)=(\neq_2)$,
    we know  that for every index $i$, $\widehat{\frak{m}}_i\widehat{f}=\widehat{\frak{m}_if}= \widehat{\mu  (=_2)}=\mu(\neq_2)$.
    Thus, 
$$M(\widehat{\frak{m}}_i\widehat{f})
=\left[\begin{matrix}
\langle \widehat{{\bf f}}_i^0, \widehat{{\bf f}}_i^1 \rangle & |\widehat{{\bf f}}_i^0|^2\\
|\widehat{{\bf f}}_i^1|^2 & \langle \widehat{{\bf f}}_i^1, \widehat{{\bf f}}_i^0 \rangle
\end{matrix}\right]=
\mu \left[\begin{matrix}
0 & 1 \\
 1 & 0
\end{matrix}\right].
$$
This implies that $\widehat{f}$ also satisfies the first order orthogonality. The reversal is also true. Therefore, $f$ satisfies the first order orthogonality is equivalent to $\widehat{f}$ satisfying it. 
Later, based on  first order orthogonality of $\widehat{f}$, we will carve out  second order orthogonality.
For a real valued binary signature,  it satisfies first order orthogonality if and only if its (2-by-2) signature matrix is an orthogonal matrix. Thus, a binary signature with {\sc ars} satisfies first order orthogonality if and only if it has parity.

Although the first order orthogonality is well-defined for any complex valued signature, when $f$ is not real-valued ($\widehat{f}$ does not satisfy {\sc ars}), we can not say anything about $\frak{m}_if$ and $\widehat{\frak{m}}_i\widehat{f}$. The properties of $\frak{m}_if$ and $\widehat{\frak{m}}_i\widehat{f}$ crucially depend on $f$ being real ($\widehat{f}$ satisfying {\sc ars}).
\end{remark}
    
  First order orthogonality implies some non-trivial properties.
 Consider the vector ${\bf f}_{i}^{0}$. We can pick a second variable $x_j$ and separate ${\bf f}_{i}^{0}$ into two vectors ${\bf f}_{ij}^{00}$ and ${\bf f}_{ij}^{01}$ according to $x_j=0$ or $1$. 
   Then $$|{\bf f}_{i}^{0}|^2=|{\bf f}_{ij}^{00}|^2+|{\bf f}_{ij}^{01}|^2=\mu.$$
    Similarly, we have $$|{\bf f}_{j}^{1}|^2=|{\bf f}_{ij}^{01}|^2+|{\bf f}_{ij}^{11}|^2=\mu.$$
    Comparing the above two equations, we have 
    \begin{equation}\label{ee2}
        |{\bf f}_{ij}^{00}|^2=|{\bf f}_{ij}^{11}|^2.
        \end{equation}
        This is ture for all pairs of indices $\{i, j\}$.
    Similarly, by considering  $$|{\bf f}_{j}^{0}|^2=|{\bf f}_{ij}^{00}|^2+|{\bf f}_{ij}^{10}|^2=\mu,$$
    we have 
    \begin{equation}\label{ee3}
        |{\bf f}_{ij}^{01}|^2=|{\bf f}_{ij}^{10}|^2,
          \end{equation}
          for all pairs $\{i, j\}$.
    
    Now, we are ready to finish the proof of the base case $n=3$ using the above equations. 
    A binary signature satisfies the first order orthogonality iff  its signature matrix is  orthogonal,  and in  this case, we say the binary signature is orthogonal.
    \begin{lemma}[Base case $n=3$]\label{01-base-3}
Let $f\in \mathcal{F}$ be a signature of arity $3$ and $f\notin \mathscr{T}$. Then $\Holant(\Delta_0, \mathcal F)$ is {\rm \#}P-hard unless $\mathcal F$ satisfies conditions (\ref{main-thr}).
\end{lemma}
\begin{proof}
Since $f$ is a ternary signature and  $f\notin \mathscr{T}$, we know $f$ is irreducible. 
If $f^{000}=0$ or $f$ does not satisfy the first order orthogonality,
then by Lemma \ref{01-by-pin} or Corollary \ref{first-ortho}, we have $\Holant^c(\mathcal F)\leqslant_T\Holant(\Delta_0, \mathcal F)$.
By Theorem \ref{hard-result}, $\Holant^c(\mathcal F)$ is \#P-hard when $\mathcal F$ does not satisfy  conditions (\ref{main-thr}), and hence $\Holant(\Delta_0, \mathcal F)$ is {\rm \#}P-hard. 
Therefore, we may assume $f^{000}=1$ by normalization and $f$ satisfies  first order orthogonality and thus equations (\ref{ee1}), (\ref{ee2}) and (\ref{ee3}).
    
    We consider binary signatures $f_{1}^{0}$, $f_{2}^{0}$ and $f_{3}^{0}$ realized by pinning. 
    If there is an index $i$ such that the binary signature $f_i^0$ is irreducible and not orthogonal, then by Corollary \ref{first-ortho} we are done.
Otherwise, $f_{1}^{0}$, $f_{2}^{0}$ and $f_{3}^{0}$ are all either reducible or orthogonal. Let $N$ be the number of orthogonal signatures among $f_{1}^{0}$, $f_{2}^{0}$ and $f_{3}^{0}$. According to 
$N = 0, 1, 2$ or $3$, there are four cases. 
 
 \begin{itemize}
\item $N=0$.
Then $f_{1}^{0}$, $f_{2}^{0}$ and $f_{3}^{0}$ are all reducible. So, $f_{1}^{0}$ is of the form $(1, a ,b, ab)$, and so are $f_{2}^{0}$ and $f_{3}^{0}$. Thus $f$ has the following matrix $$M_{1, 23}(f)=\left[\begin{matrix}
1 & a & b & ab\\
c & ac & bc & d\\
\end{matrix}\right].$$
By the equation $|{\bf f}_{12}^{01}|^2=|{\bf f}_{12}^{10}|^2$ from (\ref{ee3}), we have $$b^2+a^2b^2=c^2+a^2c^2.$$
Then, $(1+a^2)(b^2-c^2)=0$. Being real, we have $1+a^2> 0$, and thus $b^2=c^2$. Similarly by symmetry, we have $a^2=b^2=c^2$. 
By the equation $|{\bf f}_{12}^{00}|^2=|{\bf f}_{12}^{11}|^2$ from (\ref{ee2}),  we have $$1+a^2=b^2c^2+d^2.$$ 
Then, $d^2=1+a^2-a^4$.
By the equation $\langle{\bf f}_{1}^{0}, {\bf f}_{1}^{1}\rangle =0$ from  (\ref{ee1}), we have $$c+a^2c+b^2c+abd=0.$$ 
Then, $c(1+2a^2)=-abd$. Taking squares of both sides,  we have $$a^2(1+4a^2+4a^4)=a^4d^2.$$ Plug in $d^2=1+a^2-a^4$, and we have $$a^2(1+4a^2+4a^4-a^2-a^4+a^6)=a^2(1+a^2)^3=0.$$
Since $1+a^2>0$, we have $a^2=0$, and hence $b^2=c^2=0$ and $d^2=1$.
\begin{itemize}

\item If $d=1$, then $f$ has the signature matrix 
$\left[\begin{smallmatrix}
1 & 0 & 0 & 0\\
0 & 0 & 0 & 1\\
\end{smallmatrix}\right]$, which is $(=_3)$. Then, by Lemma \ref{all-equ}, 
we can realize all equality signatures $(=_k)$.
Thus, we have
$$\CSP(\mathcal{F})\leqslant_T\Holant(=_3, \mathcal F)\leqslant_T\Holant(\Delta_0, \mathcal{F}).$$ By Theorem \ref{hard-result}, we know $\CSP(\mathcal{F})$ is \#P-hard when $\mathcal{F}$ does not satisfy  conditions (\ref{main-thr}), and hence $\Holant(\Delta_0, \mathcal{F})$ is \#P-hard.

\item Otherwise, $d=-1$. We perform a holographic transformation by the orthogonal matrix $Q_1=\left[\begin{smallmatrix}
1 & 0\\
0 & -1\\
\end{smallmatrix}
\right]$. Note that $$(=_2)(Q_1^{-1})^{\otimes 2}=(=_2) ~~~~\text{ and }~~~~ Q_1^{\otimes 3}f=(=_3).$$ Thus, the holographic transformation by $Q_1$ yields 
$$\holant{=_2}{f, \mathcal F}\equiv_T \holant{=_2}{=_3, Q_1\mathcal F}.$$
Again by Lemma \ref{all-equ},
we have  $\CSP(Q_1\mathcal{F})\leqslant_T\Holant(\Delta_0, \mathcal{F}).$
By Theorem \ref{hard-result}, we know that $\CSP(Q_1\mathcal{F})$ is \#P-hard when $\mathcal{F}$ does not satisfy  conditions (\ref{main-thr}), and hence $\Holant(\Delta_0, \mathcal{F})$ is \#P-hard.
\end{itemize}
\item $N=1$.
Without loss of generality, we may assume $f_1^0$ is orthogonal and  $f_2^0$ and $f_3^0$ are reducible. Then $f_1^0$ has the form $(1, a, \epsilon a, -\epsilon)$, $f_2^0$ has the form $(1, a, b, ab)$  and $f_3^0$ has the form $(1, \epsilon a, b, \epsilon ab)$, for some $\epsilon =\pm1$. Therefore,  for some value $x$, $f$ has the signature matrix,
$$M(f)=\left[\begin{matrix}
1 & a & \epsilon a & -\epsilon\\
b & ab & \epsilon ab & x\\
\end{matrix}\right].$$

By the equation $|{\bf f}_{12}^{01}|^2=|{\bf f}_{12}^{10}|^2$ from (\ref{ee3}), we have $$(\epsilon a)^2+(-\epsilon)^2 = b^2 + (ab)^2.$$
Thus $(1+a^2) (1 - b^2) =0$.
So $b^2=1$.
By the equation $|{\bf f}_{12}^{00}|^2=|{\bf f}_{12}^{11}|^2$ from (\ref{ee2}), we have
$$1+a^2=(\epsilon ab)^2+x^2=a^2+x^2.$$
Then, $x^2=1$.
By the equation $\langle{\bf f}_{1}^{0}, {\bf f}_{1}^{1}\rangle =0$  from (\ref{ee1}), we have 
\begin{equation}\label{simple-eqn-p14}
b+a^2b+\epsilon^2a^2b-\epsilon x=0.
\end{equation}
Then, $\epsilon x=b(1+2a^2)$. Taking squares of both sides,  we have $1=(1+2a^2)^2$, which implies that $a=0$. So by (\ref{simple-eqn-p14}), we have $b-\epsilon x=0$, and thus $x=\frac{b}{\epsilon}=\epsilon b$.
It follows that $M(f)=\left[\begin{smallmatrix}
1 & 0 & 0 & -\epsilon\\
b & 0 & 0 & \epsilon b\\
\end{smallmatrix}\right]$, with $b^2=\epsilon^2 =1$.

Mating variable $x_1$ of one copy of $f$ with variable $x_1$ of another copy of $f$ (with $x_2$ and $x_3$ as dangling variables), we get a 4-ary signature $\frak{m}_{23}f$ with the signature matrix
$$M(\frak{m}_{23}f)=M_{x_2x_3, x_1}(f)M_{x_1, x_2x_3}(f)=\left[\begin{matrix}
1 & b\\
0 & 0\\
0 & 0\\
-\epsilon & \epsilon b\\
\end{matrix}\right]\left[\begin{matrix}
1 & 0 & 0 & -\epsilon\\
b & 0 & 0 & \epsilon b\\
\end{matrix}\right]=
\left[\begin{matrix}
2 & 0 & 0 & 0\\
0 & 0 & 0 & 0\\
0 & 0 & 0 & 0\\
0 & 0 & 0 & 2\\
\end{matrix}\right]
=2M(=_4).$$
Therefore, we can realize $(=_4)$, and then by Lemma \ref{all-even-equ} we can realize all equality signatures $(=_{2k})$ of even arity.
Thus, $$\CSP_2(\mathcal{F})\leqslant_T\Holant(=_4, \mathcal F)\leqslant_T\Holant(\Delta_0, \mathcal{F}).$$ 
By Theorem \ref{hard-result}, we know $\CSP_2(\mathcal{F})$ is \#P-hard when $\mathcal{F}$ does not satisfy  conditions \ref{main-thr}, and hence $\Holant(\Delta_0, \mathcal{F})$ is \#P-hard.

\item $N=2$.
Without loss of generality, we may assume $f_2^0$ and $f_3^0$ are orthogonal, and $f_1^0$ is reducible. Then, $f_2^0$ has the form $(1, \epsilon_1 a, a, -\epsilon_1)$ where $\epsilon_1=\pm 1$, $f_3^0$ has the form $(1, \epsilon_2 a, a, -\epsilon_2)$ where $\epsilon_2=\pm1$, and $f_1^0$ has the form $(1, \epsilon_1a, \epsilon_2a, \epsilon_1\epsilon_2a^2)$. Then for some $x$, $f$ has the form
$$M(f)=\left[\begin{matrix}
1 & \epsilon_1a & \epsilon_2 a & \epsilon_1\epsilon_2a^2\\
a & -\epsilon_1 & -\epsilon_2 & x\\
\end{matrix}\right].$$

By the equation $|{\bf f}_{12}^{01}|^2=|{\bf f}_{12}^{10}|^2$, we have $$(\epsilon_2a)^2+(\epsilon_1\epsilon_2a^2)^2
=a^2+(-\epsilon_1)^2.$$
So we get $a^4=1$.
Since $a$ is real, we have $a = \pm 1$.
By the equation $\langle{\bf f}_{1}^{0}, {\bf f}_{1}^{1}\rangle =0$, we have 
$$a-\epsilon_1^2a-\epsilon_2^2a+\epsilon_1\epsilon_2x=-a+\epsilon_1\epsilon_2x=0.$$
Then, $x=\frac{a}{\epsilon_1\epsilon_2}=\epsilon_1\epsilon_2a$.
By mating we get  $\frak{m}_{23}f$, and we have
$$M(\frak{m}_{23}f)=
\left[\begin{matrix}
1 & a\\
\epsilon_1a & -\epsilon_1\\
\epsilon_2a & -\epsilon_2\\
\epsilon_1\epsilon_2 & \epsilon_1 \epsilon_2 a\\
\end{matrix}\right]\left[\begin{matrix}
1 & \epsilon_1a & \epsilon_2 a & \epsilon_1\epsilon_2\\
a & -\epsilon_1 & -\epsilon_2 & \epsilon_1\epsilon_2a\\
\end{matrix}\right]=
2\left[\begin{matrix}
1 & 0 & 0 & \epsilon_1\epsilon_2\\
0 & 1 & \epsilon_1\epsilon_2 & 0\\
0 & \epsilon_1\epsilon_2 & 1 & 0\\
\epsilon_1\epsilon_2 & 0 & 0 & 1\\
\end{matrix}\right].$$
\begin{itemize}
    \item 
If $\epsilon_1\epsilon_2=1$, then $\frak{m}_{23}f$ is 2 times the {\sc Is-Even} signature,
which takes value 1 on  all inputs of even weight, and 0 otherwise.
Note that, for $H =  \frac{1}{\sqrt{2}} \left[\begin{smallmatrix}
1 & 1\\
-1 &  1\\
\end{smallmatrix}\right]$  we have
$$(=_2)(H^{-1})^{\otimes 2}=(=_2) ~~~~\text{ and }~~~~ H^{\otimes 4} (\frak{m}_{23}f) =4(=_4).$$
Thus, a holographic transformation by $H$ yields 
$$\holant{=_2}{\frak{m}_{23}f, \mathcal F}\equiv_T \holant{=_2}{=_4, H\mathcal F}.$$
By Lemma \ref{all-even-equ}, we have 
$$\CSP_2(H\mathcal{F})\leqslant_T\holant{=_2}{=_4, H\mathcal F}\leqslant_T\Holant(\Delta_0, \mathcal{F}).$$
By Theorem \ref{hard-result}, $\CSP_2(H\mathcal{F})$ is \#P-hard when $\mathcal{F}$ does not satisfy conditions (\ref{main-thr}), and hence $\Holant(\Delta_0, \mathcal{F})$ is \#P-hard.

\item Otherwise, $\epsilon_1\epsilon_2=-1$. Then $g(y_1, y_2, y_3, y_4) = \frak{m}_{23}f$ can be normalized as $\left[\begin{smallmatrix}
1 & 0 & 0 & -1\\
0 & 1 & -1 & 0\\
0 & -1 & 1 & 0\\
-1 & 0 & 0 & 1\\
\end{smallmatrix}\right]$, where the row index is $y_1y_2$ and column index is $y_3y_4 \in \{0, 1\}^2$, both listed lexicographically.
After a permutation of variables, we have $M_{y_1y_3, y_2y_4}(\frak{m}_{23}f)=\left[\begin{smallmatrix}
1 & 0 & 0 & 1\\
0 & -1 & -1 & 0\\
0 & -1 & -1 & 0\\
1 & 0 & 0 & 1\\
\end{smallmatrix}\right].$ Connecting variables $y_2, y_4$ of a copy of $\frak{m}_{23}f$ with variables $y_1, y_3$ of another copy of $\frak{m}_{23}f$
respectively, we get a signature with the signature matrix 
$$M_{y_1y_3, y_2y_4}(\frak{m}_{23}f)M_{y_1y_3, y_2y_4}(\frak{m}_{23}f)=
\left[\begin{smallmatrix}
1 & 0 & 0 & 1\\
0 & -1 & -1 & 0\\
0 & -1 & -1 & 0\\
1 & 0 & 0 & 1\\
\end{smallmatrix}\right]
\left[\begin{smallmatrix}
1 & 0 & 0 & 1\\
0 & -1 & -1 & 0\\
0 & -1 & -1 & 0\\
1 & 0 & 0 & 1\\
\end{smallmatrix}\right]
=
2\left[\begin{smallmatrix}
1 & 0 & 0 & 1\\
0 & 1 & 1 & 0\\
0 & 1 & 1 & 0\\
1 & 0 & 0 & 1\\
\end{smallmatrix}\right].
$$
Now perform a holographic transformation by $H$, and we get $(=_4)$,
which 
implies that $\Holant(\Delta_0, \mathcal{F})$ is \#P-hard when $\mathcal{F}$ does not satisfy conditions (\ref{main-thr}).
\end{itemize}
\item $N=3$.
Then for some values $a$,  $x$ and  $\epsilon_1$,  $\epsilon_2 =\pm 1$, the  signature $f$ has the signature matrix 
$$M(f)=\left[\begin{matrix}
1 & \epsilon_1a & \epsilon_2 a & -\epsilon_1\epsilon_2\\
a & -\epsilon_1 & -\epsilon_2 & x\\
\end{matrix}\right].$$
By the equation $\langle{\bf f}_{1}^{0}, {\bf f}_{1}^{1}\rangle =0$, we have $$a-\epsilon_1^2a-\epsilon_2^2a-\epsilon_1\epsilon_2x=-a-\epsilon_1\epsilon_2x=0.$$
Hence, $x
=-\epsilon_1\epsilon_2a$.
A holographic transformation by $Z^{-1}$ yields 
$$\holant{=_2}{f, \Delta_0, \mathcal F}\equiv_T \holant{\neq_2}{\widehat{f}, \widehat{\Delta_0}, \widehat{\mathcal F}}.$$
Note that
$\widehat{\Delta_0}=Z^{-1}(1, 0)^{\tt T}=(1, 1)^{\tt T}$,
and a simple calculation shows
$$M(\widehat{f})=Z^{-1}M(f)((Z^{-1})^{\tt T})^{\otimes 2}=
\left[\begin{smallmatrix}
1-a\ii & 0\\
0 & 1+a\ii\\
\end{smallmatrix}\right]
\left[\begin{smallmatrix}
(1+\epsilon_1)(1+\epsilon_2) & ~(1-\epsilon_1)(1+\epsilon_2) & ~(1+\epsilon_1)(1-\epsilon_2) & ~(1-\epsilon_1)(1-\epsilon_2)\\
(1-\epsilon_1)(1-\epsilon_2) & ~(1+\epsilon_1)(1-\epsilon_2) & ~(1-\epsilon_1)(1+\epsilon_2) & ~(1+\epsilon_1)(1+\epsilon_2)\\
\end{smallmatrix}\right].$$
\begin{itemize}
    \item 
If $\epsilon_1=\epsilon_2=1$, then up to a constant,  $M(\widehat{f})=\left[\begin{smallmatrix}
1-a\ii & 0 & 0 & 0\\
0 & 0 & 0 & 1+a\ii\\
\end{smallmatrix}\right]$. 
Let $\widehat{Q_2}=\left[\begin{smallmatrix}
\sqrt[3]{1+a\ii} & 0\\
0 & \sqrt[3]{1-a\ii}\\
\end{smallmatrix}\right]$. Then $(\widehat{Q_2})^{\otimes 3}\widehat{f}$ has the signature matrix $(1+a^2)\left[\begin{smallmatrix}
1 & 0 & 0 & 0\\
0 & 0 & 0 & 1\\
\end{smallmatrix}\right]$.
Thus, a holographic transformation by $\widehat{Q_2}$ 
yields 
$$\holant{\neq_2}{\widehat{f}, \widehat{\Delta_0}, \widehat{\mathcal F}}\equiv_T \holant{\neq_2}{=_3, \widehat{Q_2}\widehat{\Delta_0},  \widehat{Q_2}\widehat{\mathcal F}}.$$
Thus, we have
$$\holant{\neq_2}{=_3, \widehat{Q_2}\widehat{\Delta_0},  \widehat{Q_2}\widehat{\mathcal F}} \leqslant_T \Holant(\Delta_0,  \mathcal F).$$
By Corollary \ref{cspk-dic}, we know $\holant{\neq_2}{=_3, \widehat{Q_2}\widehat{\Delta_0},  \widehat{Q_2}}$ is \#P-hard when $\mathcal F$ does not satisfy  conditions (\ref{main-thr}), and hence $\Holant(\Delta_0,  \mathcal F)$ is \#P-hard.

\item
Otherwise, $M(\widehat{f})$ has the signature matrix 
$\left[\begin{smallmatrix}
0 & 0 & 0 & 1-a\ii\\
1+a\ii & 0 & 0 & 0\\
\end{smallmatrix}\right]$ up to a permutation of variables. 
Connecting $\widehat{f}$ with $\widehat{\Delta_0}=(1, 1)$ using $\neq_2$, we get a binary signature $\widehat g$ with matrix 
$$M(\widehat g)=[1, 1]
\left[\begin{matrix}
0 & 1 \\
1 & 0\\
\end{matrix}\right]
\left[\begin{matrix}
0 & 0 & 0 & 1-a\ii\\
1+a\ii & 0 & 0 & 0\\
\end{matrix}\right]=(1+a\ii, 0, 0, 1-a\ii).$$

Connecting one variable of $\widehat{f}$ with one variable of $\widehat{g}$ using $\neq_2$, we get a signature $\widehat h$ with the signature matrix 
$$M(\widehat h)=\left[\begin{matrix}
1+a\ii & 0 \\
0 & 1-a\ii\\
\end{matrix}\right]
\left[\begin{matrix}
0 & 1 \\
1 & 0\\
\end{matrix}\right]
\left[\begin{matrix}
0 & 0 & 0 & 1-a\ii\\
1+a\ii & 0 & 0 & 0\\
\end{matrix}\right]=
\left[\begin{matrix}
(1+a\ii)^2 & 0 & 0 & 0\\
0 & 0 & 0 & (1-a\ii)^2\\
\end{matrix}\right].$$

Then, a holographic transformation by $\widehat{Q_3}=\left[\begin{smallmatrix}
\sqrt[3]{(1-a\ii)^2} & 0\\
0 & \sqrt[3]{(1+a\ii)^2}\\
\end{smallmatrix}\right]$ yields 
$$\holant{\neq_2}{\widehat{h}, \widehat{\Delta_0},  \widehat{\mathcal F}}\equiv_T \holant{\neq_2}{=_3, \widehat{Q_3}\widehat{\Delta_0},  \widehat{Q_3}\widehat{\mathcal F}}.$$
Then similarly by Corollary \ref{cspk-dic}, we have $\Holant(\Delta_0,  \mathcal F)$ is \#P-hard.
\end{itemize}
\end{itemize}
Thus, $\Holant(\Delta_0, \mathcal F)$ is {\rm \#}P-hard unless $\mathcal F$ satisfies conditions (\ref{main-thr}).
\end{proof}

\subsection{The hardness of Holant($\Delta_0, \mathcal F$): the induction step}
Now, we consider the inductive step.
The general strategy is that we start with a signature $f\in \mathcal{F}$ of arity $n\geqslant 4$ that is not in $\mathscr{T}$, and realize a signature $g$ of arity $n-1$ or $n-2$ by pinning or merging that is also not in $\mathscr{T}$. 
By a sequence of reductions (that is constant in length independent of the problem instance size), we can realize a signature $h$ of arity $3$ that is not in $\mathscr{T}$. Then we are done. In other words, we want to show that the multipartite entanglement is preserved under projections onto $|0\rangle$ or forming self-loops by $|\phi^{+}\rangle$.

For all indices $i$ and all pairs of indices $\{j, k\}$, consider $f_i^0$ and $\partial_{jk}f$.  If there exist $i$ or $\{j, k\}$ such that $f_i^0$ or $\partial_{jk}f \notin \mathscr{T}$, 
then we can realize
$g=f_i^0$ or $\partial_{jk}f$ which has arity $n-1$ or $n-2$, and we are done.
Otherwise, $f_i^0$ and $\partial_{jk}f \in \mathscr T$ for all $i$ and all $\{j, k\}$. 
 We denote this property 
 by $f\in \int\mathscr T$. 
 Under the assumption that $f\in \int\mathscr T$, 
 our goal is to show that we can realize $\Delta_1$ and hence we are done by the hardness of $\Holant^c(\mathcal{F})$,
 or there is an unary signature $a(x_u)$ or binary signature $b(x_v, x_w)$ such that $a(x_u) \mid f$
or $b(x_v, x_w) \mid f.$
Then, we have $f=a(x_u)\otimes g$ or $f=b(x_v, x_w) \otimes g$ for some $g$ of arity $n-1$ or $n-2$. By the definition of $\mathscr T$, we know $g \notin \mathscr T$ since $f \notin \mathscr T$. By Lemma \ref{lin-wang}, we can realize $g$ by factorization, and we are done.
When $n\geqslant 5$, the above induction proof can be achieved by the interplay of the unique factorization, and the commutivity of $f_i^0$ (pinning) and $\partial_{jk}f$ (merging) operations 
on disjoint indices (Lemmas \ref{unary-factor} and \ref{binary-factor}).
For $n=4$, the proof requires more work  (Lemma \ref{unary-n=4}); we need to combine the induction proof and  first order orthogonality to handle it.   

We denote the property that $f_i^0\in \mathscr T_1$ 
for all $i$ 
by $\int_1 \mathscr{T}_1$. 
We carry out our induction proof by the following lemmas.


\begin{lemma}\label{all-div}
Let $f$ be a signature of arity $n \geqslant 3$. If there exists a nonzero signature $g$, the scope of which is a subset of the scope of $f$, 
such that $g \mid f_i^0$ for all indices $i$ disjoint with the scope of $g$ and furthermore, $g \mid \partial_{jk}f$ for some pair of indices $\{j, k\}$ disjoint with the scope of $g$, then $g \mid f$.  (Note that if $\partial_{jk}f \equiv 0$ then $g \mid \partial_{jk}f$ is satisfied.)
\end{lemma}
\begin{proof}
We may assume $f$ is nonzero, for otherwise the conclusion trivially holds.
We now prove this for a unary signature $g=(a, b)$. We assume $g$ is on the variable $x_u$.
Consider the signature $f'= bf_u^0-af_u^1$.
Clearly, $x_j$ and $x_k$ are in the scope of $f'$. Thus, $f'$ has arity at least $2$. 
For every $i\neq u$, we have $f_i^0=(a, b)\otimes h$ for some $h$. Then, $(f_i^0)_u^0=a\cdot h$, $(f_i^0)_u^1=b\cdot h$, and hence
$$(f')_i^0= (bf_u^0-af_u^1)_i^0=bf_{ui}^{00}-af_{ui}^{10}=b(f_{i}^{0})_u^0-a(f_{i}^{0})_u^1=ba\cdot h-ab\cdot h \equiv 0.$$
Moreover, there are indices $j, k\neq u$ such that $g(x_u) \mid \partial_{jk}f$. Then, $\partial_{jk}f=(a, b)\otimes h'$, for some $h'$.
Then, we have $(\partial_{jk}f)_u^0=a\cdot h'$, $(\partial_{jk}f)_u^1= b \cdot h'$, and hence 
$$\partial_{jk}(f')=\partial_{jk}(bf_u^0-af_u^1)=b(\partial_{jk}f)_u^0-a(\partial_{jk}f)_u^1=ba\cdot h'-ab\cdot h'\equiv 0.$$
By Lemma \ref{all-zero}, we have $f'\equiv 0$. Thus, we have $f_u^0 : f_u^1 = a : b$, and hence $g(x_u)\mid f$.

For a signature $g$ of arity $\geqslant n-2$, the proof is essentially the same, which we omit here.
\end{proof}

\begin{lemma}\label{unary-factor}
Let $f$ be a signature of arity $n\geqslant 5$, $f\notin \mathscr T$, $f \in \int \mathscr T$ and $f \in \int_1 \mathscr T_1$. 
Then there is a unary signature $a(x_u)$ such that $a(x_u)\mid f$, or $\Delta_1$ is realizable from $f$.

\end{lemma}
\begin{proof}
As $f \notin \mathscr T$, $f$ is nonzero. 
We may further assume $f_i^0\not\equiv 0$ for all indices $i$. Otherwise, we have $f^{\vec{0}}=0$. Then, by Lemma \ref{01-by-pin}, we can realize $\Delta_1$. 
By the same reason, we may further assume $f_{ij}^{00}\not\equiv 0$ for all pairs of indices $\{i, j\}$.


For some arbitrary index $r$, we consider $f_r^0$. Since $f_r^0\in \mathscr{T}_1$, there exists some unary signature $a(x_u)$ such that $a(x_u)\mid f_r^0$. We show $a(x_u)\mid f$.
 Consider $f_i^0$ for all indices $i\neq u, r$.
 Since $f_i^0\in \mathscr{T}_1$, 
 there is a unary signature $a'(x_u)$ such that $a'(x_u)\mid f_i^0$, and hence we have $a'(x_u)\mid (f_i^0)_r^0.$
On the other hand, 
since $a(x_u)\mid f_r^0$, we also have $a(x_u)\mid (f_r^0)_i^0$. 
Note that the pinning operations on different variables commute. Thus, we have $(f_r^0)_i^0=(f_i^0)_r^0$, and we know it is a nonzero signature. Then, by UPF (Lemma \ref{unique}), we have $a(x_u)\sim a'(x_u)$. Thus, $a(x_u)\mid f_i^0$ for all indices $i\neq u$.

Then, we show  $a(x_u)\mid \partial_{jk}f$ for some arbitrary pair of indices $j, k\neq u$.
If $\partial_{jk}f\equiv 0$, then we have $a(x_u)\mid \partial_{jk}f$ and hence $a(x_u)\mid f$ by Lemma \ref{all-div}.
Next, we assume $ \partial_{jk}f\not\equiv 0$. Similarly, if for some index $i\neq j, k$, we have $(\partial_{jk}f)_i^0\equiv 0$, then we have $\partial_{jk}f(\vec{0})=0$ and hence by Lemma \ref{01-by-inter},
we can realize $\Delta_1$.
Otherwise, $(\partial_{jk}f)_i^0\not\equiv 0$ for all $i\notin \{j, k\}$. Recall that $\partial_{jk}f \in \mathscr{T}$. We show the variable $x_u$ must appear in a unary signature in the UPF of $\partial_{jk}f \in \mathscr{T}$.
\begin{itemize}
\item For a contradiction, suppose there is an irreducible binary signature $b(x_u, x_v)$ such that $b(x_u, x_v)\mid \partial_{jk}f$. Since $f$ has arity $n\geqslant 5$, we can pick some index $i \notin \{u, v, j, k\}$ such that $b(x_u, x_v)\mid (\partial_{jk}f)_i^0$. 
Note that $(\partial_{jk}f)_i^0=\partial_{jk}(f_i^0)\not\equiv 0$ by the commutativity of pinning and merging.
Thus, $\partial_{jk}(f_i^0)$ has an irreducible binary tensor divisor $b(x_u, x_v)$.
However, $f_i^0\in \mathscr{T}_1$ and so is $\partial_{jk}(f_i^0)$. By UPF, we get a contradiction.
    \item 
Thus, there is a unary signature $a''(x_u)$ such that $a''(x_u)\mid \partial_{jk}f$. Pick some index $i\notin \{u, j, k\}$, and we have $a''(x_u)\mid (\partial_{jk}f)_i^0$. We also have $a(x_u)\mid \partial_{jk}(f_i^0)$ since $a(x_u)\mid f_i^0$.
Again, $(\partial_{jk}f)_i^0=\partial_{jk}(f_i^0)\not\equiv 0$ by commutativity. 
Then by UPF, we have $a''(x_u) \sim a(x_u)$. Thus, $a(x_u)\mid f$ by Lemma \ref{all-div} and we are done.
\end{itemize}
\end{proof}

\begin{lemma}\label{binary-factor}
Let $f$ be a signature of arity $n\geqslant 5$, $f\notin \mathscr T$, $f \in \int \mathscr T$ and $f \notin \int_1 \mathscr T_1$. Then, there is an irreducible binary signature $b(x_v, x_w)$ such that 
$b(x_v, x_w) \mid f$, or $\Delta_1$ is realizable from $f$.
\end{lemma}
\begin{proof}
Since $f \notin \int_1 \mathscr T_1$, but  $f \in \int \mathscr T$,  there is some index $r$ such that $f^0_r$ is  nonzero and has an irreducible binary signature factor $b(x_v, x_w)$.
We will show this $b(x_v, x_w)$ divides $f$.
Again, we may assume $f_{i}^{0}\not\equiv 0$ and $f_{ij}^{00}\not\equiv 0$ for all $i$  and all  $\{i, j\}$. Otherwise, we can realize $\Delta_1$ by Lemma \ref{01-by-pin}.

Consider $f_i^0$ for all indices $i\notin \{v, w, r\}$. Since $f_i^0\in \mathscr{T}$ and $f_i^0 \not\equiv 0$, there is either a unary signature $a(x_v)$ or an irreducible binary signature $b'(x_v, x_{w'})$ such that $a(x_v)\mid f_i^0$ or $b'(x_v, x_{w'})\mid f_i^0$. 
We also have $b(x_v, x_w)\mid (f_r^0)_i^0$ since $b(x_v, x_w)\mid f_r^0$. Again, we have $(f_r^0)_i^0=(f_i^0)_r^0\not\equiv 0$.
Then by UPF, we know that the  unary signature $a(x_v)$ does not exist, and it must be $b'(x_v, x_{w'})\mid f_i^0$ and $b(x_v, x_w)=b'(x_v, x_{w'})$. 
Thus, we have $b(x_v, x_w)\mid f_i^0$ for all $i\notin \{v, w\}$.

Then, for an arbitrary pair of indices $\{j, k \}$ disjoint with $\{v, w\}$, we show $b(x_v, x_w) \mid \partial_{jk}f$.
Again, we may assume $\partial_{jk}f\not\equiv 0$ 
(for otherwise $b(x_v, x_w) \mid \partial_{jk}f$ is proved) and furthermore $(\partial_{jk}f)_i^0\not\equiv 0$ for all $i$ disjoint with $\{j, k\}$, for otherwise, we can realize $\Delta_1$.
 Since $f$ has arity $n\geqslant 5$, we can pick some index $i \notin \{u, v, j, k\}$ such that $b(x_v, x_w)\mid \partial_{jk}(f_i^0)$ due to $b(x_v, x_w)\mid f_i^0$. 
 Recall that $\partial_{jk}f\in \mathscr{T}$, we consider the UPF of $\partial_{jk}f$. Using a similar argument as in   the previous paragraph, we have $b(x_v, x_w)\mid \partial_{jk}f$ by UPF.
\end{proof}

Combining the above two lemmas,  we have the following result.
\begin{lemma}[Inductive step for $n\geqslant 5$]\label{01-induction}
If $f\in \mathcal{F}$ is a signature of arity $n\geqslant 5$ and $f\notin \mathscr T$, then
\begin{itemize}
    \item 
$\Holant^c(\mathcal{F})\leqslant_T \Holant(\Delta_0, \mathcal F)$ or
\item there is a signature $g\notin \mathscr T$ of arity $n-1$ or $n-2$ such that $\Holant(\Delta_0, g, \mathcal{F})\leqslant_T \Holant(\Delta_0, \mathcal{F})$.
\end{itemize}
\end{lemma}

Now, the only case left for the induction proof
is when $f$ is a signature of arity $4$.
We deal with it by using 
the first order orthogonality condition.
    \begin{lemma}[Inductive step $n=4$]\label{unary-n=4}
Let $f\in \mathcal{F}$ be a signature of arity $4$ and $f\notin \mathscr{T}$. Then 
\begin{itemize}
    \item 
$\Holant^c(\mathcal{F})\leqslant_T \holant{=_2}{\Delta_0, \mathcal F}$, or
\item $\CSP_2(\mathcal{F})\leqslant_T 
\holant{=_2}{=_4, \mathcal F}\leqslant_T 
\holant{=_2}{\Delta_0, \mathcal F}$, or
\item there is a signature $g\notin \mathscr T$ of arity $3$ such that $\Holant(\Delta_0, g, \mathcal{F})\leqslant_T \Holant(\Delta_0, \mathcal{F})$.
\end{itemize}
\end{lemma}

\begin{proof}
First, we may assume $f$ is irreducible. Otherwise, we consider its irreducible factors. 
Since $f\notin \mathscr{T}$, it has an irreducible factor $g$ of arity $3$ such that $g\notin \mathscr{T}$. By Lemma \ref{lin-wang}, $g$ is realizable from $f$ by factorization, and the lemma is proved.
Also we may assume $f^{0000}=1$ after normalization and $f$ satisfies  first order orthogonality; otherwise, by Lemma \ref{01-by-pin} and Corollary \ref{first-ortho}, we are done. 
We consider signatures $f_i^0$ realized by pinning $x_i$ to 0 in $f$,
for all $i$. If there is $i$ such that the ternary signature $f_i^0 \notin \mathscr{T}$, then  we are done, since $f_i^0$ has arity 3. Also, since $f$ has arity 4,
$\partial_{ij}f$ is a binary signature for any pair of indices $\{i,j\}$.
Hence $\partial_{ij}f \in \mathscr{T}$.
Thus, we may assume $f\in \int\mathscr{T}$. 

\begin{itemize}
    \item If $f \in \int_1\mathscr{T}_1$, then there are three unary signatures such that $f_1^0=a_1(x_2)\otimes a_2(x_3) \otimes a_3(x_4)$. By the same proof in Lemma \ref{unary-factor}, we have $a_2(x_3)\mid f_2^0$ and $a_3(x_4) \mid f_2^0$. Thus, $a_2(x_3)\otimes a_3(x_4)\mid f_2^0$.
    \item Otherwise, there is an index $i$ such that $f_i^0$ has an irreducible binary factor.
    Without loss of generality, we assume that $f_1^0=a_1(x_2) \otimes b_1(x_3, x_4)$ where $b_1(x_3, x_4)$ is irreducible. By the same proof  as in Lemma \ref{binary-factor}, we have $b_1(x_3, x_4)\mid f_2^0$.
\end{itemize}
Therefore, in both cases, there is a binary signature $b(x_3, x_4)$, which may be reducible, i.e., $b(x_3, x_4)= a_2(x_3)\otimes a_3(x_4)$, such that $b(x_3, x_4)\mid f_1^0$ and $b(x_3, x_4)\mid f_2^0$.
Thus, we have $$f_1^0= a_1(x_2) \otimes b(x_3, x_4) ~~~~\text{ and }~~~~ f_2^0= a'_1(x_1) \otimes b(x_3, x_4).$$
By a normalization we may let $b(x_3, x_4)=(1, a, b, c)$, $a_1(x_2)=(1, x)$ and $a'_1(x_1)=(1, y)$. Then, $f$ has the  signature matrix, for some $z, z_1, z_2, z_3$ $$M_{12,34}(f)=\begin{bmatrix}
1 & a & b & c\\
x & ax & bx & cx\\
y & ay & by & cy\\
z & z_1 & z_2 & z_3\\
\end{bmatrix}.$$

Then, we consider the signature $f_3^0$. It has the  signature matrix 
$$M_{12,4}(f_3^0)=\begin{bmatrix}
1 & a \\
x & ax\\
y & ay\\
z & z_1\\
\end{bmatrix}.$$
We have $f_3^0\in \mathscr{T}$, and nonzero. In its unique factorization,
if $x_2$ and $x_4$ belong to an irreducible binary signature, then
$(f_3^0)_1^0$, which has the signature matrix 
$M_{2,4}(f_{13}^{00})= \left[ \begin{smallmatrix}
1 & a \\
x & ax\\
\end{smallmatrix}\right]$,  would have been an irreducible binary signature, a contradiction.
Similarly $x_1$ and $x_4$ do not belong to an irreducible binary signature
in the unique factorization of $f_3^0$. Therefore $x_4$ appears in a
unary signature in the factorization of $f_3^0$.
It follows that
$z_1=az$. Similarly from $f_4^0\in \mathscr{T}$, we can prove
 $z_2=bz$. We also write $z_3$ as $cz+w$. Thus, we have
$$M_{12,34}(f)=(1, x, y, z)^{\tt T}\otimes(1, a, b, c)+w ({(0, 1)^{\tt T}})^{\otimes 2} {\otimes}(0, 1)^{\otimes 2}.$$
We know $w\neq 0$ since $f \notin \mathscr{T}.$
By pinning any 3 of the 4 variables to 0,
we can realize four unary signatures $(1, a), (1, b), (1, x)$ and $(1, y)$. For example, $(1, x)$ can be realized from $f$ by pinning variables $x_1, x_3$ and $x_4$ to $0$.
\begin{itemize}
    \item 
Suppose $a, b, x, y$ are not all zero, say $x \neq 0$.
We connect the unary $(1, x)$ with the variable $x_2$ of $f$, and we get a signature $g$ with the signature matrix 
$$M_{1,34}(g)=\begin{bmatrix}
1+x^2 & a(1+x^2) & b(1+x^2) & c(1+x^2)\\
y+xz & a(y+xz) & b(y+xz) & c(y+xz)+xw\\
\end{bmatrix}.$$
Clearly, $1+x^2\neq 0$. By normalization, we have 
$$M_{1,34}(g)=\begin{bmatrix}
1 & a & b & c\\
x' & ax' & bx' & cx'+w'\\
\end{bmatrix},$$
where $x'=\frac{y+xz}{1+x^2}$ and $w'=\frac{xw}{1+x^2}$,
and $w' \neq 0$ since $xw\neq 0$.
Thus
\begin{equation}\label{rank-1-pertubation-g}
    g = (1, x')_{x_1} \otimes (1,a,b,c)_{x_3,x_4} + w' (0, 1)^{\otimes 3}.
\end{equation}

We claim that  $g \notin \mathscr{T}$.
Otherwise consider the unique factorization of $g$ in $\mathscr{T}$.
By the same proof above for $f_3^0 \in \mathscr{T}$,
we can see that $x_1$ of $g$ cannot appear in an irreducible binary
signature, either with $x_3$ or with $x_4$,
as a tensor factor in the unique prime factorization of $g$.
Hence $x_1$ must appear in a unary signature in this
factorization. This  would imply that $w'=0$, by the form of 
$M_{1,34}(g)$,
a contradiction.

It follows that   $g \notin \mathscr{T}$, 
and we are done.

%
%

\item Otherwise, $a=b=x=y=0$. Then, we know 
$$M_{12,34}(f)=\begin{bmatrix}
1 & 0 & 0 & c\\
0 & 0 & 0 & 0\\
0 & 0 & 0 & 0\\
z & 0 & 0 & z_3\\
\end{bmatrix}.$$
Here, we write $z_3$ as $cz+w$. 
By equation (\ref{ee2}), we have $1+c^2=z^2+z^2_3$ and $1+z^2=c^2+z^2_3$. Thus, we have $c^2=z^2$ and $z_3^2=1$. 
By pinning variables $x_1$ and $x_2$ to 0, we can realize the binary signature $(1, 0, 0, c)$. If it is not reducible or orthogonal, then by Lemma \ref{01-by-inter} we can realize $\Delta_1$. Otherwise, we have $c=0$ or $c=\pm 1$. Similarly, we have $z=0$ or $z=\pm 1$. 
As we already have $c^2=z^2$, we get $c=z=0$ or $c^2=z^2=1$.
We consider the signature $\frak{m}_{34}f$ realized by mating variables $x_3, x_4$ of $f$. We have 
$$M(\frak{m}_{34}f)=M_{12,34}(f)(M^{\tt}_{12,34}(f))^{\tt T}=\begin{bmatrix}
1+c^2 & 0 & 0 & z+cz_3\\
0 & 0 & 0 & 0\\
0 & 0 & 0 & 0\\
z+cz_3 & 0 & 0 & z^2+z^2_3\\
\end{bmatrix}.$$
If $c=z=0$, then we have $M(\frak{m}_{34}f)=\left[\begin{smallmatrix}
1 & 0 & 0 & 0\\
0 & 0 & 0 & 0\\
0 & 0 & 0 & 0\\
0 & 0 & 0 & 1\\
\end{smallmatrix}\right]=M(=_4).$ 
Otherwise, $c^2=z^2=1$. Also, we know $z_3\neq cz$ since 
$f\notin \mathscr{T}$.
Note that $z_3^2 = 1$ and $(cz)^2=1$. This implies that $z_3=-cz$. Then, we have $z+cz_3=z-c^2z=z-z=0$. Thus, we have $M(\frak{m}_{34}f)=\left[\begin{smallmatrix}
2 & 0 & 0 & 0\\
0 & 0 & 0 & 0\\
0 & 0 & 0 & 0\\
0 & 0 & 0 & 2\\
\end{smallmatrix}\right]=2M(=_4).$ 

Therefore, we can realize $(=_4)$ from $f$, and then by Lemma \ref{all-even-equ} we can realize all equality signatures $=_{2k}$ of even arity. Thus, we have $$\CSP_2(\mathcal{F})\leqslant_T\Holant(=_4, \mathcal F)\leqslant_T\Holant(\Delta_0, \mathcal{F}).$$ 
This completes the proof of the lemma.
\end{itemize}
\end{proof}

\begin{theorem}\label{01-dic}
$\Holant(\Delta_0, \mathcal F)$ is {\rm \#}P-hard unless $\mathcal F$ satisfies 
the tractable conditions (\ref{main-thr}).
\end{theorem}
\begin{proof}
Assume $\mathcal F$ does not satisfy conditions (\ref{main-thr}).
Then  $\mathcal F \not \subseteq  \mathscr T$. There is a signature $f \in \mathcal{F}$ of arity $n\geqslant 3$ that is not in $\mathscr{T}$.
If $n=3$, then by Lemma \ref{01-base-3}, we are done.

Suppose our statement is true for $3 \leqslant n\leqslant k$. Consider $n=k+1\geqslant 4$. By Lemmas \ref{01-induction} and  \ref{unary-n=4}, we have 
$\Holant^c(\mathcal{F})$, or $\CSP_2(\mathcal{F})$, or $\holant{=_2}{\Delta_0, g, \mathcal F}\leqslant_T \holant{=_2}{\Delta_0, \mathcal F}$ for some $g\notin \mathscr T$ of arity $k-1$ or $k$ at least $3$.
By Theorem \ref{hard-result} and the induction hypothesis, we know $\Holant^c(\mathcal{F})$, $\CSP_2(\mathcal{F})$ and $\holant{=_2}{\Delta_0, g, \mathcal F}$ are all \#P-hard when $\mathcal F$ does not satisfy  conditions (\ref{main-thr}), and hence $\holant{=_2}{\Delta_0, \mathcal F}$ is \#P-hard.
\end{proof}
\subsection{The proof of Theorem \ref{odd-dic}}
\begin{proof}
The tractability is known by Theorem \ref{main-thr}.

By Lemma \ref{odd-red}, for some orthogonal matrix $Q$, we have $\Holant(=_2\mid\Delta_0, Q\mathcal{F})\leqslant_T\Holant(=_2\mid\mathcal{F})$ or
    $\holant{\neq_2}{=_{2k+1}, \widehat{Q}\widehat{\mathcal F}}\leqslant_T\Holant(=_2\mid \mathcal{F})$.
    Since $\mathcal{F}$ does not satisfy conditions (\ref{odd-red}), we know $Q\mathcal{F}$ also does not satisfy  conditions (\ref{odd-red}). Then, by Theorem \ref{01-dic} and Corollary \ref{cspk-dic}, we have $\Holant(=_2\mid\Delta_0, Q\mathcal{F})$ and  $\holant{\neq_2}{=_{2k+1}, \widehat{Q}\widehat{\mathcal{F}}}$ are both \#P-hard. Hence, $\Holant(=_2\mid \mathcal{F})$ is \#P-hard.
    \end{proof}

\subsection{The proof of Theorem \ref{cspk-dic-thr}}\label{sec-cspd}
In this subsection,  we will prove Theorem~\ref{cspk-dic-thr}.
For $k=1$, Theorem~\ref{cspk-dic-thr} follows Theorem~\ref{csp-dic}.
Note that $(\neq_2)\notin\mathscr{L}$.
Thus for $k=2$, Theorem~\ref{cspk-dic-thr} follows Theorem~\ref{even-csp-dic}.
This implies that Theorem~\ref{cspk-dic-thr} has been proved for $k=1, 2$.

We will use the notations $\frak i^2=-1, \alpha^2=\frak i$ and $\beta^2=\alpha$ in this section.
And we use $[1, 0, \cdots, 0, a]_r$ to denote a general equality signature of arity $r$ in the following.

Note that
\[\#\operatorname{CSP}_k(\neq_2, \mathcal{F})\equiv_T \operatorname{Holant}({\mathcal EQ}_k | \neq_2, \mathcal{F}).\]
Moreover, by the following two gadgets,
we have
$(\neq_2), {\mathcal EQ}_k$ on both sides i.e.,
\begin{equation}\label{two-forms}
\#\operatorname{CSP}_k(\neq_2, \mathcal{F})\equiv \operatorname{Holant}({\mathcal EQ}_k, \neq_2|{\mathcal EQ}_k, \neq_2, \mathcal{F}).
\end{equation}

\setlength{\unitlength}{5mm}
\begin{picture}(24,10)(-3,-2)\label{const}
\put(1,3){\line(2, 1){4}}
\put(1,3){\line(1,0){5}}
\put(1,3){\line(2,-1){4}}
\put(4.5,2.9){\tiny{$\blacksquare$}}
\put(4,1.2){\tiny{$\blacksquare$}}
\put(4,4.5){\tiny{$\blacksquare$}}
\put(0.9, 2.8){\Large{$\bullet$}}
\put(3.3,1.9){\Huge{$\vdots$}}
\qbezier(15.1, 3)(18, 6)(21.2, 3)
\qbezier(15.2, 3)(18, 0)(21.2, 3)
\put(12,3){\line(1,0){12}}
\put(15,2.8){\Large{$\bullet$}}
\put(18,2.8){\tiny{$\blacksquare$}}
\put(18,4.3){\tiny{$\blacksquare$}}
\put(18,1.3){\tiny{$\blacksquare$}}
\put(18,1.8){\Huge{$\vdots$}}
\put(21, 2.8){\Large{$\bullet$}}
\put(-2.5,-0.5){Fig 1. In an instance of $\operatorname{Holant}({\mathcal EQ}_k | \neq_2, \mathcal{F})$, the first gadget realizes
$(=_{nk})$ on the RHS.}
\put(-2.5, -1.5){The squares are labeled by $(\neq_2)$ and the round vertices are labeled by $(=_{nk})$. The second }
\put(-2.5, -2.5){gadget realizes $(\neq_2)$ on the LHS, where  the 
two round vertices are labeled by $(=_{k})$ on the RHS.}
\end{picture}

\vspace{.3in}

The following two lemmas show that in $\#\operatorname{CSP}_k(\neq_2, \mathcal{F})$ the pinning signatures $\Delta_0 = [1, 0]$ and $\Delta_1 = [0, 1]$ are freely available.
The first lemma is from \cite{huang-lu-real}.
It shows that we have $[1, 0]^{\otimes k}, [0, 1]^{\otimes k}$ freely in $\#\operatorname{CSP}_k(\mathcal{F})$.
\begin{lemma}\label{pinning}
\[\#\operatorname{CSP}_k(\mathcal{F}, [1, 0]^{\otimes k}, [0, 1]^{\otimes k})\leq_T\#\operatorname{CSP}_k(\mathcal{F}). \]
\end{lemma}
The second lemma is from \cite{wang-lin}.
It shows that we can remove the tensor powers of $[1, 0]^{\otimes k}, [0, 1]^{\otimes k}$.
\begin{lemma}\label{remove-duplicate-pinning}
\[\#\operatorname{CSP}_k(\mathcal{F}, [1, 0], [0, 1])\leq_T\#\operatorname{CSP}_k(\mathcal{F}, [1, 0]^{\otimes k}, [0, 1]^{\otimes k}).\]
\end{lemma}

Firstly, we prove Theorem~\ref{cspk-dic-thr} for the case that there exists a general equality of arity less than $k$ in $\mathcal{F}$.

\begin{lemma}\label{general-equality-induction}
Let $f=[1, 0, \cdots, 0, a]_r$ with $r<k$ and $a\neq 0$, and $\mathcal{F}$ be a signature set.
Then $\operatorname{\#CSP}_k(\mathcal{F}, \neq_2, f)$ is $\#\operatorname{P}$-hard except for the following cases
\begin{itemize}
\item $\mathcal{F}\subseteq\mathscr{P}$;
\item there exists $d\in[k]$ such that $\mathcal{F}\cup \{f\}\subseteq\mathscr{A}_k^d$,
\end{itemize}
which can be computed in polynomial time.
\end{lemma}

\begin{proof}
Note that the lemma has been proved for the cases $k=1, 2$ by Theorem~\ref{csp-dic}
and Theorem~\ref{even-csp-dic}. We will prove the lemma by induction on $k$ in the following.
Let $k=tr+r_1$ with $0<r_1\leq r$. Note that $k>r$, so $t\geq 1$.
In $\operatorname{Holant}(\mathcal{EQ}_k, \neq_2|\mathcal{EQ}_k, \neq_2, f, \mathcal{F})$,
connecting $\ell t$ copies of $[1, 0, \cdots, 0, a]_r$ to ($=_{\ell k}$)
we get $[1, 0, \cdots, 0, a^{t\ell}]_{\ell r_1}$ for $\ell=1, 2, \cdots$, i.e.,
\[\operatorname{Holant}(\mathcal{EQ}^a_{r_1}, \neq_2|\mathcal{EQ}_k, \neq_2, f, \mathcal{F})
\leq_T
\operatorname{Holant}(\mathcal{EQ}_k, \neq_2|\mathcal{EQ}_k, \neq_2, f, \mathcal{F}),\]
where $\mathcal{EQ}^a_{r_1}=\{[1, 0, \cdots, 0, a^{t}]_{r_1},
[1, 0, \cdots, 0, a^{2t}]_{2r_1}, \cdots, [1, 0, \cdots, 0, a^{\ell t}]_{\ell r_1}, \cdots\}$. Let
$T=\left[\begin{smallmatrix}1 & 0\\0 & a^{\frac{t}{r_1}}\end{smallmatrix}\right]$, then
$T^{-1}(\mathcal{EQ}_{r_1}^a)=\mathcal{EQ}_{r_1}$. Thus after the holographic transformation using $T$, we have
\[\operatorname{Holant}(\mathcal{EQ}_{r_1}|T\mathcal{EQ}_k, \neq_2, T^{\otimes r}f, T\mathcal{F})
\leq_T
\operatorname{Holant}(\mathcal{EQ}_k, \neq_2|\mathcal{EQ}_k, \neq_2, f, \mathcal{F}),\]
i.e.,
\[\operatorname{\#CSP}_{r_1}(T\mathcal{EQ}_k, \neq_2, T^{\otimes r}f, T\mathcal{F})
\leq_T
\operatorname{\#CSP}_k(\neq_2, f, \mathcal{F}).\]
By induction, if $\{T\mathcal{EQ}_k, T\mathcal{F}, T^{\otimes r}f\}\nsubseteq\mathscr{P}$ and
$\{T\mathcal{EQ}_k, T\mathcal{F}, T^{\otimes r}f\}\nsubseteq\mathscr{A}_{r_1}^{d'}$ for any $d'\in[r_1]$,
then
$\operatorname{\#CSP}_{r_1}(T\mathcal{EQ}_k, \neq_2, T^{\otimes r}f, T\mathcal{F})$
is \#P-hard. Thus $\operatorname{\#CSP}_k(\neq_2, f, \mathcal{F})$ is \#P-hard.

Otherwise, if $\{T\mathcal{EQ}_k, T\mathcal{F}, T^{\otimes r}f\}\subseteq\mathscr{P}$, then $\mathcal{F}\subseteq\mathscr{P}$
since $T$ is a diagonal matrix. Moreover,
if $\{T\mathcal{EQ}_k\bigcup T\mathcal{F}\bigcup T^{\otimes r}f\}\subseteq\mathscr{A}_{r_1}^{d'}$ for some $d'\in[r_1]$,
let $T'=\left[\begin{smallmatrix}1 & 0\\
0  & \gamma^{d'}\end{smallmatrix}\right]$,
where
 $\gamma$ is a 4$r_1$-th primitive root of unary, so, $\gamma^{4r_1}=1$, then
 $T'^{\otimes k}T^{\otimes k}(=_k)\in\mathscr{A}$,
 $T'^{\otimes r}T^{\otimes r}f\in\mathscr{A}$
 and  $T'T\mathcal{F}\subseteq\mathscr{A}$.
 Firstly, by $T'^{\otimes k}T^{\otimes k}(=_k)\in\mathscr{A}$,
 we have $(\gamma^{d'k}a^{\frac{kt}{r_1}})^4=1$.
 This implies that $\gamma^{d'}a^{\frac{t}{r_1}}$ is a $4k$-th root of unity,
 i.e., there exists $d\in[k]$ such that $\gamma^{d'}a^{\frac{t}{r_1}}=\rho^{d}$, then
  $T'T=
 \left[\begin{smallmatrix}
 1 & 0\\
 0 & \rho^d
 \end{smallmatrix}\right]$.
 Thus $\mathcal{F}\subseteq \mathscr{A}_k^d$ and $f\in\mathscr{A}_k^d$.
 This finishes the proof.
\end{proof}



\begin{definition}
Let $f=(f_{i_1i_2\cdots i_n}), g=(g_{i_1i_2\cdots i_n})$ be two $n$-ary signatures,
then $fg$ is an $n$-ary signature and $(fg)_{i_1i_2\cdots i_n}=f_{i_1i_2\cdots i_n}g_{i_1i_2\cdots i_n}$ for any $i_1i_2\cdots i_n\in\{0, 1\}^n$.
In particular,  $f^k=(f^k_{i_1i_2\cdots i_n})$, and for a signature set $\mathcal{F}$,
$\mathcal{F}^k=\{f^k|f\in\mathcal{F}\}$.
\end{definition}

\begin{definition}
If $f$ has affine support of rank $r$, and $X=\{x_{j_1}, x_{j_2}, \cdots, x_{j_r}\}$ is a set of free variables, then $\underline{f_{X}}$ is the compressed signature of $f$ for $X$ such that $\underline{f_X}(x_{j_1}, x_{j_2}, \cdots, x_{j_r})=f(x_1, x_2, \cdots, x_n),$ where $(x_1x_2\cdots x_n)$ is in the support of $f$. When it is clear from the context, we omit $X$ and use $\underline{f}$ to denote $\underline{f_X}$.
\end{definition}

Note that if $f$ has affine support, then $f\in\mathscr{A}$ iff $\underline{f}\in\mathscr{A}$.

 Theorem~\ref{cspk-dic-thr} has been proved for $k=1, 2$.
We will prove it by induction on $k$.
The following lemma is the main tool to do the induction.
More precisely, the following lemma shows that if $f=f_1f_2\cdots f_{k'}$ for some $k'|k$, then we can simulate
$\#\operatorname{CSP}_{\frac{k}{k'}}(\neq_2, f)$ by $\#\operatorname{CSP}_k(\neq_2, f_1, f_2, \cdots, f_{k'})$.
If $f\notin\mathscr{P}$ and $f\notin\mathscr{A}_{\frac{k}{k'}}^{d'}$ for any $d'\in[\frac{k}{k'}]$, then by induction,
$\#\operatorname{CSP}_{\frac{k}{k'}}(\neq_2, f)$ is \#P-hard. Thus $\#\operatorname{CSP}_k(\neq_2, f_1, f_2, \cdots, f_{k'})$ is \#P-hard.

\begin{lemma}\label{induction}
Let $\mathcal{F}$ be a signature set. Signatures $f_1, f_2, \cdots, f_{k'}$ has the same arity and $f=f_1f_2\cdots f_{k'}$, where $k'|k$,
 then
\[\#\operatorname{CSP}_{\frac{k}{k'}}(f, \mathcal{F}^{k'})\leq_T\#\operatorname{CSP}_k(f_1, f_2, \cdots, f_{k'}, \mathcal{F}).\]
\end{lemma}
\begin{proof}
In an instance of $\#\operatorname{CSP}_{\frac{k}{k'}}(f, \mathcal{F}^{k'})$,
by expanding each variable $x$ to $k'$ copies of $x$, replacing each $h^{k'}\in\mathcal{F}^{k'}$ by $k'$
copies of $h$, and replacing $f$  by $f_1, f_2, \cdots, f_{k'}$, we get an instance of $\#\operatorname{CSP}_k(f_1, f_2, \cdots, f_{k'}, \mathcal{F})$
and its value is same as the value of the instance of $\#\operatorname{CSP}_{\frac{d}{d'}}(f, \mathcal{F}^{k'})$.
This finishes the proof.
\end{proof}

If  the support of $f$ is not affine, then the support of $f^k$ is not affine.
By Lemma~\ref{induction},
\[\#\operatorname{CSP}(f^k)\leq_T\#\operatorname{CSP}_k(f).\]
Moreover, by Theorem~\ref{csp-dic}, $\#\operatorname{CSP}(f^k)$ is \#P-hard. Thus
$\#\operatorname{CSP}_k(f)$ is \#P-hard
and we finish the proof of Theorem~\ref{cspk-dic-thr}.
So in the following, we assume all the signatures have affine supports.

By the definition of $\mathscr{P}$, every signature $f\in\mathscr{P}$ has a decomposition as a product of signatures
over disjoint of variables, where each factor has support contained in a pair of antipodal points: There exists a partition
$X=\{x_1, x_2, \cdots, x_n\}=\bigcup_{j=1}^{\ell}X_j$, and a signature $f_j$ on $X_j$ such that
$f(X)=\prod_{j=1}^{\ell}f_j(X_j)$, and for all $1\leq j\leq \ell,$ the support of $f_j$
is contained in $\{\alpha_j, \bar{\alpha_j}\}$ for some $\alpha_j\in\{0, 1\}^{|X_j|}$.

The proof of the following lemma totally follows the the proof of Lemma~4.8 and Lemma~4.9 of \cite{CLX-HOLANTC}.
We just generalize it from $k=2$ to general $k$.
\begin{lemma}\label{non-product}
If $f\nsubseteq\mathscr{P}$ but $f^k\in\mathscr{P}$, then for any $d$, in $\#\operatorname{CSP}_{k}(\mathcal{F})$, we
can construct $h_1, h_2, \cdots, h_d$ such that $h=h_1h_2\cdots h_d\notin\mathscr{P}$, i.e.,
\[\#\operatorname{CSP}_{k}(h, \mathcal{F})\leq_T\#\operatorname{CSP}_{k}(\mathcal{F}).\]
\end{lemma}
\begin{proof}
Just replacing $2$ by $k$ in the proof of Lemma~4.8 and Lemma~4.9 of \cite{CLX-HOLANTC},
we can construct a rank-2 signature $g$ from $f$ in $\#\operatorname{CSP}_{k}(\mathcal{F})$
such that $g$ has the support
$(x_1)_{k_1}(x_2)_{k_2}$ and $\underline{g}=(1, a, b, -ab)$ up to a nonzero scalar, where $ab\neq 0$.
By pinning all the $x_2=0$, we get a signature $u_1$ has the form $(x_1)_{k_1}$
and $\underline{u_1}=(1, a)$,  and by pinning all the $x_1=0$, we get a signature $u_2$ has the form $(x_2)_{k_2}$
and $\underline{u_2}=(1, b)$.
Let $u=u_1\otimes u_2$.
For any $d$, let $h_1=g, h_2=\cdots=h_k=u$, then
$h=h_1h_2\cdots h_d$ has the compressed signature $(1, a^k, b^k, -a^kb^k)$ that is not in $\mathscr{P}$.
\end{proof}

Since $\mathcal{F}\nsubseteq\mathscr{P}$, by Lemma~\ref{non-product} and Lemma~\ref{induction}, we have
\[\#\operatorname{CSP}(h, \neq_2, \mathcal{F}^k)\leq \#\operatorname{CSP}_k(\neq_2, \mathcal{F}),\]
where $h\notin\mathscr{P}$.
If $\mathcal{F}^k\nsubseteq\mathscr{A}$, then $\#\operatorname{CSP}(h, \neq_2, \mathcal{F}^k)$ is \#P-hard by Theorem~\ref{csp-dic}.
This implies that $\#\operatorname{CSP}_k(\neq_2, \mathcal{F})$ is \#P-hard.
So in the following, we assume that $\mathcal{F}^k\subseteq\mathscr{A}$.

\begin{lemma}\label{rho-power}\cite{Cai-Fu-csp}
If a signature $f$ has affine support and $f(x_1, x_2, \cdots, x_n)$ is a power of $\rho$ for any $(x_1x_2\cdots x_n)$ in the support of $f$, then
 there exists multilinear polynomial $Q(x_{i_1}, x_{i_2}, \cdots, x_{i_r})$
such that $f(x_1, x_2, \cdots, x_n)=\rho^{Q(x_{i_1}, x_{i_2}, \cdots, x_{i_r})}$, where
$\{x_{i_1}, x_{i_2}, \cdots, x_{i_r}\}$ is a set of free variables.
\end{lemma}

Since $\mathcal{F}^k\subseteq\mathscr{A}$, for any signature $f\in\mathcal{F}$ of rank $r$, without loss of generality,
assume that $\{x_1, x_2,\cdots, x_r\}$ is a set of free variables of $f$,
by Lemma~\ref{rho-power},
there exists a multilinear polynomial $Q(x_1, x_2, \cdots, x_r)\in\mathbb{Z}[X]$
such that $\underline{f}(x_1, x_2, \cdots, x_r)=\rho^{Q(x_1, x_2, \cdots, x_r)}$ up to a scalar.


\begin{definition}
Suppose $f$ has affine support of rank $r$ with $\{x_1, x_2, \cdots, x_r\}$ as a  set of free variables. We use all non-empty combinations $\sum_{j=1}^ra_jx_j (a_j\in\mathbb{Z}_2,$ not all zero$)$ of $x_1, x_2, \cdots, x_r$ as the names of bundles of $f$. The type of each bundle is a possibly empty multiset of $``+"$ and $``-"$, and is defined as follows: For every input variable $x_k (1\leq k\leq n)$ of $f$ there is a unique bundle named $\sum_{j=1}^ra_jx_j $ such that on the support of $f$, $x_k$ is either always equal to $\sum_{j=1}^ra_jx_j \pmod 2$ or always equal to $\sum_{j=1}^ra_jx_j+1 \pmod 2$. In the former case we add a $``+"$, and in the latter case we add a $``-"$ to the bundle type for the bundle named $\sum_{j=1}^ra_jx_j $, and we say the variable $x_k$  belong to this bundle.

All input variables are partitioned into bundles.
If the number of variables in each bundle is multiple $\ell$ for some integer $\ell$, then we say $f$ has $\ell$-type support.
\end{definition}
We can list a signature's input variables, by listing all its non-empty bundles followed by the bundle type. For example, $f(x_1(++), x_2(++), (x_1+x_2)(--))$ has rank 2, arity 6 and its support is 2-type.

Connecting one variable $x_i$ of a signature to $(=_k)$ using $(\neq_2)$,
is equivalent to replace $x_i$ by $(k-1)$ copies of $\bar{x_i}$.
We call this operation is $(k-1)$-multiple.
If the variables in the same bundle is greater than $k$, by connecting $d$ copies of them to $=_k$, we make these $k$ variables disappear. We call this operation is collation.

\begin{definition}
Let $S$ be a subset of $[k]$. If $f\in\mathscr{A}_k^d$ iff $k\in S$,
then we call $f$ is an $S$-affine signature.
In particular, if $S$ is exactly all the even (odd) numbers of $[k]$, then we call $f$ is an even-affine (odd-affine) signature.
\end{definition}

\begin{lemma}\label{connect-equality-d}
Let $f$ be a signature of affine support.
By doing $(k-1)$-multiple and collation to $f$ we get a new signature $f'$, then $f\in\mathscr{P}$ iff $g\in\mathscr{P}$; $f\in\mathscr{A}_k^d$ iff $g\in\mathscr{A}_k^d$ for any $d\in[k]$.
\end{lemma}
\setlength{\unitlength}{5mm}
\begin{picture}(24,8)(-18,-2)\label{const}
\put(1,3){\line(2, 1){4}}
\put(1,3){\line(1,0){5}}
\put(-4,3){\line(1,0){5}}
\put(1,3){\line(2,-1){4}}
\put(-1.5,2.9){\tiny{$\blacksquare$}}
\put(-2.8,3.2){$x$}
\put(-0.3,3.2){$\bar{x}$}
\put(5,3.2){$\bar{x}$}
\put(3.8,4.8){$\bar{x}$}
\put(4.5,1.5){$\bar{x}$}
\put(-4.2,2.9){\small{$\blacktriangle$}}
\put(-4,3){\line(-1, 0){5}}
\put(-4,3){\line(-2, 1){4}}
\put(-4,3){\line(-2, -1){4}}
\put(0.9, 2.8){\Large{$\bullet$}}
\put(3.3,1.9){\Huge{$\vdots$}}
\put(-7,1.9){\Huge{$\vdots$}}
\put(-17.5,-0.5){Fig. 2 Transforming the variable $x$ to $k-1$ copies of $\bar{x}$ by connecting $(=_k)$ using $(\neq_2)$.
}
\put(-17.5,-1.5){The triangle, square and bullet are labeled by $f$, $(\neq_2)$ and $(=_k)$ respectively.}
\end{picture}

\begin{proof}
We prove the lemma for $(k-1)$-multiple operation for the case $f\in\mathscr{A}_k^d$.
Other cases are similar and we omit them here.
Note that $(\neq_2)\in\mathscr{A}_k^d$ and $(=_k)\in\mathscr{A}_k^d$ for each $d\in[k]$.
Thus if $f\in\mathscr{A}_k^d$ for some $d$, then $g\in\mathscr{A}_k^d$ by the closure of $\mathscr{A}$.
Conversely, note that the variable $x$ of $f$ which is connected to $(=_k)$ using $(\neq_2)$ is flipped into $(k-1)$ copies of $\bar{x}$ in $g$.
By connecting these $(k-1)$ copies of $\bar{x}$ to $(=_k)$ using $(k-1)$ copies of $(\neq_2)$, we get $f'$ that is same as $f$.
Thus if $g\in\mathscr{A}_k^d$, then $f'\in\mathscr{A}_k^d$ and so does $f$.
\end{proof}

By Lemma~\ref{connect-equality-d}, we can flip the variable $\bar{x}$  to $x$
in the support of $f$ at the same time. Thus we can assume that there is no $\bar{x}$.
\begin{definition}
Let $f$ be an $n$-ary signature which has support of rank $r$.
If there exists a set of free variables, without loss of generality, assume that it is  $\{x_{1}, x_{2}, \cdots, x_{r}\}$,
such that all the bundles have type ``+", i.e., \[f(x_1(\underbrace{++\cdots+}_{n_1})x_2(\underbrace{++\cdots+}_{n_2})\cdots (x_1+x_2+\cdots+x_r)(\underbrace{++\cdots+}_{n_{12\cdots r}})),\] where $n_1+n_2+\cdots+n_{12\cdots r}=n$,
then we call $f$ is monotone. And we denote its support
by
  \[(x_1)_{n_1}(x_2)_{n_2}\cdots (x_r)_{n_r}(x_1+x_2)_{n_{12}}\cdots(x_1+x_2+\cdots +x_r)_{n_{12\cdots r}}.\]
\end{definition}


%

\begin{lemma}\label{monotone}
If $f\notin\mathscr{A}_k^d$ for some $d\in[k]$, then 
we have
\[\#\operatorname{CSP}_k(\neq_2, g, \mathcal{F})\leq_T\#\operatorname{CSP}_k(\neq_2, f, \mathcal{F}), \]
where $g$
is monotone, $g\nsubseteq\mathscr{A}_k^d$ and has the same rank as $f$.
\end{lemma}
\begin{proof}
Without loss of generality, assume that $\{x_1, x_2, \cdots, x_r\}$ is a set of free variables of $f$.
If there exists variable $x_i$ such that $x_i=a_1x_1+a_2x_2+\cdots +a_rx_r+1$ in the support of $f$,
then a ($k-1$)-multiple to $x_i$, we get a new signature $g$ which is not in
  $\mathscr{A}_k^d$ by Lemma~\ref{connect-equality-d}.
  Moreover, $g$ has affine support and $\{x_1, x_2, \cdots, x_r\}$ is a set of free variables and
  the variable $x_i$ is flipped into $k-1$ copies of  $\bar{x}_i=a_1x_1+a_2x_2+\cdots +a_rx_r$.
\end{proof}

By Lemma~\ref{monotone}, in the following we can assume that all the signatures are monotone.
The following lemma will simplify the signatures further by reducing ranks.

\begin{lemma}\label{rank-reducing}
Let $f$ be a signature and $\mathcal{F}$ be a signature set.
If $f\notin\mathscr{A}_k^d$ for some $d\in[k]$, then we have
\[\#\operatorname{CSP}_k(g, \mathcal{F}, [1, 0], [0, 1])\leq_T\#\operatorname{CSP}_k(f, \mathcal{F}, [1, 0], [0, 1]), \]
where $g\notin\mathscr{A}_k^d$ and $g$ has rank at most 3.
\end{lemma}
\begin{proof}
Without loss of generality, assume $\{x_1, x_2, \cdots, x_r\}$ is a set of free variables of $f$,
and $\underline{f}(x_1, x_2, \cdots, x_r)=\rho^{Q(x_1, x_2, \cdots, x_r)}$
 by Lemma~\ref{rho-power} up to a scalar, where $Q(x_{1}, x_{2}, \cdots, x_{r})$ is a multilinear polynomial.
By the holographic transformation using $T_k^d=\left[\begin{smallmatrix}
1 & 0\\
0 & \rho^k
\end{smallmatrix}\right]$, we have
\[
\operatorname{Holant}(E^d_k|\widehat{f}, \widehat{\mathcal{F}}, [1, 0], [0, 1])\equiv_T\#\operatorname{CSP}_k(f, \mathcal{F}, [1, 0], [0, 1]),
\]
where $\widehat{f}=(T_k^d)^{\otimes arity(f)}f$, $\widehat{\mathcal{F}}=T_k^d\mathcal{F}$, and $E^d_k={\mathcal EQ}_k(T_k^d)^{-1}$.
Note that $\widehat{f}$ has the same support as $f$ and the ratio of $f(x_1, x_2, \cdots, x_n)$ and $\widehat{f}(x_1, x_2, \cdots, x_n)$
is a power of $\rho$ for any $(x_1x_2\cdots x_n)$ in the support of $f$.
Thus
 there exists a multilinear polynomial $\widehat{Q}(x_1, x_2, \cdots, x_r)$
such that $\underline{\widehat{f}}(x_1, x_2, \cdots, x_r)=\rho^{\widehat{Q}(x_1, x_2, \cdots, x_r)}$.
Assume that
\begin{equation}\label{multylinear}
Q(x_1, x_2, \cdots, x_r)=\sum_{1\leq i\leq r}a_ix_i+\sum_{1\leq j, \ell\leq r}a_{j\ell}x_jx_{\ell}+P(x_1, x_2, \cdots, x_r),
\end{equation}
where $P(x_1, x_2, \cdots, x_r)$ is a polynomial and all the terms have power at least 3.
In (\ref{multylinear}),
\begin{itemize}
\item If there exists $a_i\not\equiv 0\pmod k$, we pin all free variables to 0 by $[1, 0]$ except $x_i$, then we get a rank-1 signature
 that is not in $\mathscr{A}$.

\item If there exists $a_{j\ell}\not\equiv 0\pmod {2k}$, we pin all free variables to 0 by $[1, 0]$ except $x_j, x_{\ell}$, then we get a rank-2 signature
that is not in $\mathscr{A}$.

\item Finally, if $P(x_1, x_2, \cdots, x_r)\not\equiv 0 \pmod{4k}$, suppose the monomial $M$ has the minimum degree, among all monomials in $P$
whose coefficient that is nonzero
modulo $4k$. We pin all free variables which are not in $M$ to 0 by $[1, 0]$ and pin the variables in $M$ to 1
by $[0, 1]$
except 3 of them, then we get a rank-3 function
that is not in $\mathscr{A}$.
\end{itemize}
If $a_i\equiv 0\pmod k$, $a_{j\ell}\equiv 0\pmod {2k}$ for all $a_i$ and $a_{j\ell}$, and $P(x_1, x_2, \cdots, x_r)\equiv 0\pmod{4k}$,
then $\widehat{f}\in\mathscr{A}$. This is a contradiction.
In total, we always can get a non-affine signature of degree at most 3 in
$\operatorname{Holant}(E^d_k|\widehat{f}, \widehat{\mathcal{F}}, [1, 0], [0, 1]).$
This means we can get a signature $g\notin \mathscr{A}_k^d$  of degree at most 3 in
$\#\operatorname{CSP}_k(f, \mathcal{F}, [1, 0], [0, 1]).$
This finishes the proof.
\end{proof}

\begin{remark}
Let $f$ be a rank-3 signature with
\[\underline{f}(x_1, x_2, x_3)=\frak i^{a_1x_1+a_2x_2+a_3x_3+a_{12}x_1x_2+a_{13}x_1x_3+a_{23}x_2x_3+a_{123}x_1x_2x_3}\]
and the support $(x_1)_{k_1}(x_2)_{k_2}(x_3)_{k_3}(x_1+x_2)_{k_{12}}(x_1+x_3)_{k_{13}}(x_2+x_3)_{k_{23}}(x_1+x_2+x_3)_{k_{123}}$,
and $g$ be a signature with
\[\underline{g}(y_1, y_2)=\frak i^{b_1y_1+b_2y_2+a_3x_3+b_{12}y_1y_2}\]
and the support $(y_1)_{k'_1}(y_2)_{k'_2}(y_1+y_2)_{k'_{12}}$.
Assume that $k_{12}+k'_1<\ell k, k_{13}+k'_2<\ell k, k_{23}+k'_{12}<\ell k$ for some integer $\ell$.
By connecting the variable bundles $(y_1), (y_2), (y_1+y_2)$ of $g$ to
the variables bundles $(x_1+x_2), (x_1+x_3), (x_2+x_3)$ of $f$ by $(=_{\ell k})$
to get a new signature $f'$.
Note that \[\underline{f'}(x_1, x_2, x_3)=\underline{f}(x_1, x_2, x_3)\underline{g}(x_1+x_2, x_1+x_3, x_2+x_3),\] i.e.,
\[\underline{f'}(x_1, x_2, x_3)=\frak i^{a_1x_1+a_2x_2+a_3x_3+a_{12}x_1x_2+a_{13}x_1x_3+a_{23}x_2x_3+a_{123}x_1x_2x_3
+b_1(x_1+x_2)+b_2(x_1+x_3)+b_{12}(x_1+x_2)(x_1+x_3)}.\]
Let \[\underline{f'}(x_1, x_2, x_3)=\frak i^{c_1x_1+c_2x_2+c_3x_3+c_{12}x_1x_2+c_{13}x_1x_3+c_{23}x_2x_3+c_{123}x_1x_2x_3}.\]
We emphasis that \[a_{123}=c_{123}.\]
In the following, we will connect the variables of $g$ to $f$ using $(=_{\ell k})$ in different ways.
We emphasis that the coefficient of $x_1x_2x_3$ will not be changed.
\end{remark}

The proof of the following lemma follows the proof of Lemma~4.17 of \cite{CLX-HOLANTC}.
We just generalize it from \#CSP$_2$ to \#CSP$_k$.
\begin{lemma}\label{h-hhh}
Let $f$ be a rank-3 signature which has the support
\[(x_1)_{\frac{k}{2}}(x_2)_{\frac{k}{2}}(x_3)_{\frac{k}{2}}(x_1+x_2)_{\frac{k}{2}}(x_1+x_3)_{\frac{k}{2}}(x_2+x_3)_{\frac{k}{2}}(x_1+x_2+x_3)_{\frac{k}{2}},\]
and the compressed signature
\[\underline{f}(x_1, x_2, x_3)=\frak i^{a_1x_1+a_2x_2+a_3x_3+2a_{12}x_1x_2+2a_{13}x_1x_3+2a_{23}x_2x_3+2a_{123}},\] where $a_{123}$ is odd,
and $g$ be a rank-2 signature that has the support
\[(y_1)_{\frac{k}{2}}(y_2)_{\frac{k}{2}}(y_1+y_2)_{\frac{k}{2}},\]
and the compressed signature
\[\underline{g}(y_1, y_2)=\frak i^{b_1y_1+b_2y_2+2b_{12}y_1y_2}.\]
Then \[\#\operatorname{CSP}_k(f')\leq\#\operatorname{CSP}_k(f, g),\]
where $f'\notin\mathscr{A}$ is a rank-3 signature and has $k$-type support.
\end{lemma}
\begin{proof}
We take one copy of $f((x_1)_{\frac{k}{2}}(x_2)_{\frac{k}{2}}(x_3)_{\frac{k}{2}}(x_1+x_2)_{\frac{k}{2}}(x_1+x_3)_{\frac{k}{2}}(x_2+x_3)_{\frac{k}{2}}(x_1+x_2+x_3)_{\frac{k}{2}})$
and three copies of $g$: $g_1((u_1)_{\frac{k}{2}}(u_2)_{\frac{k}{2}}(u_1+u_2)_{\frac{k}{2}})$,
$g_2((v_1)_{\frac{k}{2}}(v_2)_{\frac{k}{2}}(v_1+v_2)_{\frac{k}{2}})$,
$g_3((w_1)_{\frac{k}{2}}(w_2)_{\frac{k}{2}}(w_1+w_2)_{\frac{k}{2}})$.
We connect the variables $\underbrace{x_1x_1\cdots x_1}_{\frac{k}{2}}$ of $f$ and $\underbrace{u_1u_1\cdots u_1}_{\frac{k}{2}}$ of $g_1$ by $(=_{2k})$
as the following figures shows.

\setlength{\unitlength}{5mm}
\begin{picture}(24,10)(-18,-2)\label{combining}
\put(-3,5.9){\tiny{$\blacksquare$}}
\put(-2.8, 6){\line(-1, 1){1.5}}
\put(-2.8, 6){\line(-2, 1){1.5}}
\put(-2.8, 6){\line(-1, -1){1.5}}
\put(-4,5.35){\Huge{$\vdots$}}
\put(-2,6.35){$x_1$}
\put(-2,4.9){$x_1$}
\put(-3,1.9){\small{$\blacktriangle$}}
\put(-2.8, 2){\line(-1, 1){1.5}}
\put(-2.8, 2){\line(-2, 1){1.5}}
\put(-2.8, 2){\line(-1, -1){1.5}}
\put(-4,1.35){\Huge{$\vdots$}}
\put(-2,2.8){$u_1$}
\put(-2,1.5){$u_1$}
\put(0.9, 3.8){\Large{$\bullet$}}
\put(1,4){\line(1, 1){1.5}}
\put(1,4){\line(2, 1){1.5}}
\put(1.1,4.1){\line(1, -1){1.5}}
\put(2,3.5){\Huge{$\vdots$}}
\qbezier(-2.9, 6)(0, 6.5)(1, 4.1)
\qbezier(-2.8, 6)(-2, 4)(1, 4.1)
\qbezier(-2.8, 1.9)(0, 2)(1, 4.1)
\qbezier(-2.8, 1.9)(-2, 4.2)(1, 4.1)
\put(-1,2.8){\Huge{$\vdots$}}
\put(-1,4.7){\Huge{$\vdots$}}
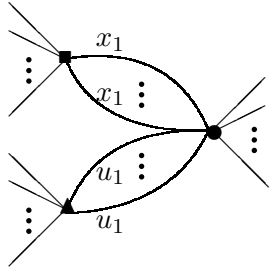
\put(-15.5,-1){Fig. 3 The square, triangle and bullet is labeled by $f$, $g$ and $(=_{2k})$ respectively.}
\end{picture}
\newline
Similarly, we connect the variables
\begin{center}
$\underbrace{x_2x_2\cdots x_2}_{\frac{k}{2}}$ and $\underbrace{v_1v_1\cdots v_1}_{\frac{k}{2}}$,

$\underbrace{x_3x_3\cdots x_3}_{\frac{k}{2}}$ and $\underbrace{w_1w_1\cdots w_1}_{\frac{k}{2}}$,

$\underbrace{(x_2+x_3)(x_2+x_3)\cdots (x_2+x_3)}_{\frac{k}{2}}$ and $\underbrace{u_2u_2\cdots u_2}_{\frac{k}{2}}$,

$\underbrace{(x_1+x_3)(x_1+x_3)\cdots (x_1+x_3)}_{\frac{k}{2}}$ and $\underbrace{v_2v_2\cdots v_2}_{\frac{k}{2}}$,

$\underbrace{(x_1+x_2)(x_1+x_2)\cdots (x_1+x_2)}_{\frac{k}{2}}$ and $\underbrace{w_2w_2\cdots w_2}_{\frac{k}{2}}$
\end{center}
by $(=_{2k})$ respectively.
Then we get a signature $f'$ which has the support \[(x_1)_k(x_2)_k(x_3)_k(x_1+x_2)_k(x_1+x_3)_k(x_2+x_3)_k(x_1+x_2+x_3)_{2k}\]
and
\[\underline{f'}(x_1, x_2, x_3)=\frak i^{c_1x_1+c_2x_2+c_3x_3+2c_{12}x_1x_2+2c_{13}x_1x_3+2c_{23}x_2x_3+2c_{123}x_1x_2x_3}\]
for some $c_1, c_2, c_3, c_{12}, c_{13}, c_{23}$.
But  $c_{123}=a_{123}$ by the remark.
Since $a_{123}$ is odd, $f'$ is not affine and we finish the proof.
\end{proof}

In the following proof, we often use the following obvious fact: If $f$ has $k$-type support and $f\notin\mathscr{A}_k^{d_0}$ for some $d_0\in[k]$,
then $f\notin\mathscr{A}_k^{d}$ for any $d\in[k]$.

\begin{lemma}\label{non-affine-transformable}
If $\mathcal{F}\nsubseteq\mathscr{A}_k^d$ for any $d\in[k]$,
 then in $\operatorname{\#CSP}_k(\neq_2, \mathcal{F})$ we can construct $[1, 0, \cdots, 0, a]_r$,  where $a\neq 0$ and $0< r< k$,
 or for some $1\leq k'<k$ and $k'|k$, we can construct
 $h$ such that $h\notin\mathscr{A}_{\frac{k}{k'}}^{d'}$ for any $d'\in[\frac{k}{k'}]$, and
 $h=f_1f_2\cdots f_{k'}$, where $f_1, f_2, \cdots, f_{k'}$
 can be constructed in $\operatorname{\#CSP}_k(\neq_2, \mathcal{F})$.
\end{lemma}
\begin{proof}
By Lemma~\ref{monotone}, we can assume all the signatures in $\mathcal{F}$ are monotone.
Moreover,  we can assume that all the signature in $\mathcal{F}$ has rank less than 4 by Lemma~\ref{rank-reducing}.
Then by Lemma~\ref{pinning} and Lemma~\ref{remove-duplicate-pinning}, the pinning signatures $[1, 0], [0, 1]$
are available freely in $\#\operatorname{CSP}_k(\neq_2, \mathcal{F})$.

Assume that $f\in\mathcal{F}$ and $f\notin\mathscr{A}_k^d$.

 If $f$ has rank-1,  we can assume that
 $f=[x, 0, \cdots, 0, y]$ since $f$ is monotone.
 If $xy=0$, then $f\in\mathscr{A}_k^d$ for any $d\in[k]$, this is a contradiction.
 Thus up to a scalar, we can assume that $f=[1, 0, \cdots, 0, a]$ with $a\neq 0$.
 Let $r$ be the arity of $f$.
 If $r\not\equiv 0\pmod k$, then we can assume that $r=r_1+\ell_1k$, where  $0< r_1< k$.
 After collations we get $[1, 0, \cdots, 0, a]$ of arity $r_1$
 and finish the proof.
 If $r\equiv 0\pmod k$, we can assume that
  $f=[1, 0, \cdots, 0, a]_{\ell k}$. Since $f$ has $k$-type support and $f\notin\mathscr{A}^d_k$ for some $d$,
  we have $f\notin\mathscr{A}$. Then we are done by letting $h=f$.


  If $f$ has rank 2, we can assume that $f$ has the support $(x_1)_{k_1}(x_2)_{k_2}(x_1+x_2)_{k_{12}}$
 since $f$ is monotone.
By pinning $x_1=0$, we have a rank-1 signature $f_1$ whose support has the form $(x_2)_{k_2+k_{12}}$, i.e.,
$f_1=[f_1(0, 0, \cdots, 0), 0, \cdots, 0, f_1(1, 1, \cdots, 1)]$, where $f_1(0, 0, \cdots, 0)=\underline{f}(0, 0)$
and $f_1(1, 1, \cdots, 1)=\underline{f}(0, 1)$
are both nonzero. Thus up to a scalar, we can assume that $f_1=[1, 0, \cdots, 0, a]$ with $a\neq 0$.
Moreover,
 by the proof for the rank-1 case, if $k_2+k_{12}\not\equiv 0\pmod k$, then we are done.
 Thus
 \begin{equation}\label{rank-2-equations-1}
 k_2+k_{12}\equiv 0\pmod k.
 \end{equation}
Then by pinning $x_2=0$ and $x_1+x_2=0$, we have
\begin{equation}\label{rank-2-equations}
\begin{split}
&k_1+k_{12}\equiv 0 \pmod k,\\
&k_1+k_{2}\equiv 0 \pmod k.
\end{split}
\end{equation}
By (\ref{rank-2-equations-1}) and the first equation of (\ref{rank-2-equations}), we have $k_1\equiv k_{2}\pmod k$.
Then by the second equation, we have $2k_1\equiv 0\pmod k.$
This implies that $k_1\equiv 0\pmod k$ or $k_1\equiv \frac{k}{2}\pmod k$.
So
we have \[k_1\equiv k_2\equiv k_{12}\equiv 0\pmod k, ~~
\operatorname{or}~~k_1\equiv k_2\equiv k_{12}\equiv \frac{k}{2}\pmod k.\]
If $k_1\equiv k_2\equiv k_{12}\equiv 0\pmod k$, then we are done since the support of $f$ has $k$-type
and $f\notin\mathscr{A}$.
Now we can assume that the support of $f$ has the form $(x_1)_{\frac{k}{2}}(x_2)_{\frac{k}{2}}(x_1+x_2)_\frac{k}{2}$
after collations.
If $f^2\notin \mathscr{A}$, then we are done by letting $h=f^2$.
Otherwise, we can assume that
 \[\underline{f}(x_1, x_2)=\alpha^{b_1x_1+b_2x_2+2b_{12}x_1x_2}\]
  up to a scalar.
  By pinning $x_2=0$ to $f$, we get $[1, 0, \cdots, 0, \alpha^{b_1}]_k$.
If $b_1\not\equiv 0\pmod 2$, then $[1, 0, \cdots, 0, \alpha^{b_1}]\notin\mathscr{A}$ and has $k$-type.
Thus we finish the proof.
So we have $b_1\equiv 0\pmod 2$. Then by pinning $x_2=0, x_1+x_2=0$, we have
$b_2\equiv 0\pmod 2$ and $2b_{12}\equiv 0\pmod 2$ respectively.
Now we can assume that \[\underline{f}(x_1, x_2)=\frak i^{b_1x_1+b_2x_2+b_{12}x_1x_2}.\]
Note that if $b_{12}\equiv 0\pmod 2$ then $f$ is even-type affine, and $b_{12}\equiv 1\pmod 2$, then $f$ is odd-type affine.

$\mathbf{Claim:}$
There is no rank-1 signature in $\mathcal{F}$ which is not in $\mathscr{A}_k^d$ for some $d\in[k]$.
If there exists a rank-2 signature $f\in\mathcal{F}$ such that $f\notin\mathscr{A}_k^d$ for some $d\in[k]$,
then $f$ has the support $(x_1)_{k_1}(x_2)_{k_2}(x_1+x_2)_{k_{12}}$
with $k_1\equiv k_2\equiv k_{12}\equiv \frac{k}{2}\pmod k$ and \[\underline{f}(x_1, x_2)=\alpha^{b_1x_1+b_2x_2+2b_{12}x_1x_2},\]
where $b_1\equiv b_2\equiv 0\pmod 2.$

 If $f$ has rank-3, then $f$ has the support
  \[(x_1)_{k_1}(x_2)_{k_2}(x_3)_{k_3}(x_1+x_2)_{k_{12}}(x_1+x_3)_{k_{13}}(x_2+x_3)_{k_{23}}(x_1+x_2+x_3)_{k_{123}}\]
    since $f$ is monotone.
By pinning $x_1=0, x_2=0, x_3=0$ to $f$, we get three signatures which have the support
\begin{equation*}
\begin{split}
&(x_2)_{k_2+k_{12}}(x_3)_{k_3+k_{13}}(x_2+x_3)_{k_{23}+k_{123}},\\
&(x_1)_{k_1+k_{12}}(x_3)_{k_3+k_{23}}(x_1+x_3)_{k_{13}+k_{123}},\\
&(x_1)_{k_1+k_{13}}(x_2)_{k_2+k_{23}}(x_1+x_2)_{k_{12}+k_{123}}
\end{split}
\end{equation*}
respectively.
By the $\mathbf{Claim}$, we have
\begin{equation}\label{rank-3-equations-1}
\begin{split}
&k_2+k_{12}\equiv k_3+k_{13}\equiv k_{23}+k_{123}\equiv 0\pmod{\frac{k}{2}},\\
&k_1+k_{12}\equiv k_3+k_{23}\equiv k_{13}+k_{123}\equiv 0\pmod {\frac{k}{2}},\\
&k_1+k_{13}\equiv k_2+k_{23}\equiv k_{12}+k_{123}\equiv 0\pmod {\frac{k}{2}}.
\end{split}
\end{equation}
and
\begin{equation}\label{rank-3-equations-1-modd}
\begin{split}
&k_2+k_{12}\equiv k_3+k_{13}\equiv k_{23}+k_{123}\pmod{k},\\
&k_1+k_{12}\equiv k_3+k_{23}\equiv k_{13}+k_{123}\pmod {k},\\
&k_1+k_{13}\equiv k_2+k_{23}\equiv k_{12}+k_{123}\pmod {k}.
\end{split}
\end{equation}

By (\ref{rank-3-equations-1})
 we have \[k_1\equiv k_2\equiv k_3\equiv k_{123}\equiv -k_{12}\equiv -k_{13}\equiv -k_{23}\pmod{\frac{k}{2}}.\]
Moreover,  by pinning $x_1+x_2=0$, we get the signature which has the support
 \[(x_1)_{k_1+k_2}(x_3)_{k_3+k_{123}}(x_1+x_3)_{k_{13}+k_{23}}.\]
By the $\mathbf{Claim}$,
\[k_1+k_2\equiv k_3+k_{123}\equiv k_{13}+k_{23}\equiv 0\pmod {\frac{k}{2}}.\]
This implies that \[2k_1\equiv  2k_{12}\equiv 0\pmod{\frac{k}{2}}.\]

So after collations, by (\ref{rank-3-equations-1-modd}),  the support of $f$ has to be one of the following forms:
\begin{equation}\label{ddd}
(x_1)_{\epsilon_1k}(x_2)_{\epsilon_2k}(x_3)_{\epsilon_3k}
(x_1+x_2)_{\epsilon_{12}k}(x_1+x_3)_{\epsilon_{13}k}(x_2+x_3)_{\epsilon_{23}k}(x_1+x_2+x_3)_{{\epsilon_{123}k}},\\
\end{equation}
\begin{equation}\label{hhh}
(x_1)_{\frac{k}{2}}(x_2)_{\frac{k}{2}}(x_3)_{\frac{k}{2}}(x_1+x_2)_{\frac{k}{2}}
(x_1+x_3)_{\frac{k}{2}}(x_2+x_3)_{\frac{k}{2}}(x_1+x_2+x_3)_{\frac{k}{2}},
\end{equation}
\begin{equation}\label{dhd}
(x_1)_{\epsilon_1k}(x_2)_{\epsilon_2k}(x_3)_{\epsilon_3k}(x_1+x_2)_{\frac{k}{2}}(x_1+x_3)_{\frac{k}{2}}(x_2+x_3)_{\frac{k}{2}}
(x_1+x_2+x_3)_{\epsilon_{123}k},\\
\end{equation}
\begin{equation}\label{hdh}
(x_1)_{\frac{k}{2}}(x_2)_{\frac{k}{2}}(x_3)_{\frac{k}{2}}
(x_1+x_2)_{\epsilon_{12}k}(x_1+x_3)_{\epsilon_{13}k}(x_2+x_3)_{\epsilon_{23}k}(x_1+x_2+x_3)_{\frac{k}{2}},\\
\end{equation}
\begin{equation}\label{qqq}
(x_1)_{\frac{k}{4}}(x_2)_{\frac{k}{4}}(x_3)_{\frac{k}{4}}(x_1+x_2)_{\frac{k}{4}}(x_1+x_3)_{\frac{k}{4}}(x_2+x_3)_{\frac{k}{4}}(x_1+x_2+x_3)_{\frac{k}{4}},
\end{equation}
\begin{equation}\label{q3q}
(x_1)_{\frac{k}{4}}(x_2)_{\frac{k}{4}}(x_3)_{\frac{k}{4}}(x_1+x_2)_{\frac{3k}{4}}(x_1+x_3)_{\frac{3k}{4}}(x_2+x_3)_{\frac{3k}{4}}(x_1+x_2+x_3)_{\frac{k}{4}},
\end{equation}
\begin{equation}\label{3q3}
(x_1)_{\frac{3k}{4}}(x_2)_{\frac{3k}{4}}(x_3)_{\frac{3k}{4}}(x_1+x_2)_{\frac{k}{4}}(x_1+x_3)_{\frac{k}{4}}(x_2+x_3)_{\frac{k}{4}}(x_1+x_2+x_3)_{\frac{3k}{4}},
\end{equation}
\begin{equation}\label{333}
(x_1)_{\frac{3k}{4}}(x_2)_{\frac{3k}{4}}(x_3)_{\frac{3k}{4}}(x_1+x_2)_{\frac{3k}{4}}(x_1+x_3)_{\frac{3k}{4}}(x_2+x_3)_{\frac{3k}{4}}(x_1+x_2+x_3)_{\frac{3k}{4}},
\end{equation}
where $\epsilon_i=0$ or 1.

Note that except for the cases (\ref{ddd}) and (\ref{hhh}),
we always have a rank-2 signature $g$ by pinning $x_1=0$ to $f$ which has support
$(y_1)_{\frac{k}{2}}(y_2)_{\frac{k}{2}}(y_1+y_2)_{\frac{k}{2}}$ after collations.
If $f$ has the support as (\ref{ddd}), then we are done by letting $h=f$ since it is $k$-type.

In the following, we firstly assume all the signatures in $\mathcal{F}$ have the form (\ref{hhh}).
Otherwise, by the $\mathbf{Claim}$, we can
assume that there is a rank-2 signature $g\in\mathcal{F}$ which has the support $(y_1)_{\frac{k}{2}}(y_2)_{\frac{k}{2}}(y_1+y_2)_{\frac{k}{2}}$.

 Firstly, assume that
  all the signatures in $\mathcal{F}$ have
   the support (\ref{hhh}).
 If $f$ has the support (\ref{hhh}) and $f^2\notin\mathscr{A}$, then we are done by letting $h=f^2$.
 Otherwise, we can assume that
 \[\underline{f}(x_1, x_2, x_3)=\alpha^{a_1x_1+a_2x_2+a_3x_3+2a_{12}x_1x_2+2a_{13}x_1x_3+2a_{23}x_2x_3+4a_{123}x_1x_2x_3}.\]
 By pinning $x_1=0$, we have a rank-2 signature $g_1$ with the support $(x_2)_k(x_3)_k(x_2+x_3)_k$ and
 \[\underline{g_1}(x_2, x_3)=\alpha^{a_2x_2+a_2x_2+a_3x_3+2a_{23}x_2x_3}.\]
 If $g_1\notin\mathscr{A}$, then we are done by letting $h=g_1$.
 Otherwise, $a_2\equiv a_3\equiv a_{23}\equiv 0\pmod 2$.
 Moreover, by pinning $x_2=0, x_3=0$, we have
 $a_1\equiv a_2\equiv  a_3\equiv a_{12}\equiv a_{13}\equiv a_{23}\equiv 0\pmod 2$.
 This implies that
 \[\underline{f}(x_1, x_2, x_3)=\frak i^{a_1x_1+a_2x_2+a_3x_3+2a_{12}x_1x_2+2a_{13}x_1x_3+2a_{23}x_2x_3+2a_{123}x_1x_2x_3}.\]
 Let $\widehat{f}=\left[\begin{smallmatrix}
 1 & 0\\
 0 & \rho^d
 \end{smallmatrix}\right]^{\otimes \frac{7}{2}k}f$, then
 \[\underline{\widehat{f}}(x_1, x_2, x_3)=\rho^{\frac{dk}{2}(4x_1+4x_2+4x_3+4x_1x_2+4x_1x_3+4x_2x_3+4x_1x_2x_3)}f(x_1, x_2, x_3),\]
 i.e.,
  \[\underline{\widehat{f}}(x_1, x_2, x_3)=\frak i^{(2d+a_1)x_1+(2d+a_2)x_2+(2d+a_3)x_3+
  2(d+a_{12})x_1x_2+2(d+a_{13})x_1x_3+2(d+a_{23})x_2x_3+2(d+a_{123})x_1x_2x_3}.\]
  So $f$ is even-type if $a_{123}$ is even and $f$ is odd type if $a_{123}$ is odd.
  Since $\mathcal{F}\nsubseteq\mathscr{A}_k^d$ for any $k\in[k]$, there at least two signatures $f_1, f_2\in\mathcal{F}$ and
   \[\underline{f}_i(x_1, x_2, x_3)=\frak i^{a^{(i)}_1x_1
   +a^{(i)}_2x_2+a^{(i)}_3x_3+2a^{(i)}_{12}x_1x_2+2a^{(i)}_{13}x_1x_3+2a^{(i)}_{23}x_2x_3+2a^{(i)}_{123}x_1x_2x_3}\]
   for $i=1, 2$, where $a^{(1)}_{123}$ is odd and $a_{123}^{(2)}$ is even.
   Let $h=f_1f_2$, then $h$ is $\frac{k}{2}$-type and
   \begin{equation*}
   \begin{split}
   &\underline{h}(x_1, x_2, x_3)=\\
   &\frak i^{(a^{(1)}_1+a^{(2)}_1)x_1+(a^{(1)}_2+a^{(2)}_2)x_2+(a^{(1)}_3+a^{(2)}_3)x_3
   +2(a^{(1)}_{12}+a^{(2)}_{12})x_1x_2+2(a^{(1)}_{13}+a^{(2)}_{13})x_1x_3
   +2(a^{(1)}_{23}+a^{(2)}_{23})x_2x_3+2(a^{(1)}_{123}+a^{(2)}_{123})x_1x_2x_3}.
   \end{split}
   \end{equation*}
   Since $a^{(1)}_{123}+a^{(2)}_{123}$ is odd, $h\notin\mathscr{A}$.
   Moreover, $h$ has $\frac{k}{2}$-type and we are done.

Now we assume that there exists a rank-2 signature $g\in\mathcal{F}$.
By the $\mathbf{Claim}$, $g$ is $\frac{k}{2}$-type and
\[\underline{g}(y_1, y_2)=\frak i^{b_1x_1+b_2x_2+b_{12}x_1x_2},\]
where $b_1\equiv b_2\equiv 0\pmod 2.$
Depending on whether $b_{12}$ is even or odd, $g$ is even-affine or odd-affine.
We assume that $b_{12}$ is even.
 After the holographic transformation using
$\left[\begin{smallmatrix}
1 & 0\\
0 & \rho
\end{smallmatrix}\right]$, the following proof can work for the case that $b_{12}$ is odd and we omit it here.
Since $g$ is even-affine, for each even $d\in[k]$, there exists $f\in\mathcal{F}$ such that $f\notin\mathscr{A}_k^d$
since $\mathcal{F}\nsubseteq\mathscr{A}_k^d$ for any $d\in[k]$.
Let
\begin{center}
$\overline{\mathcal{F}}^{\operatorname{even}}=\{f\in\mathcal{F}|f\notin\mathscr{A}_k^d$ for some even $d\in[k]\}$.
\end{center}
 If there exists another rank-2 signature $g'\in\overline{\mathcal{F}}^{\operatorname{even}}$, by the $\mathbf{Claim}$,
 $g'$ is $\frac{k}{2}$-type and
 \[\underline{g'}=\frak i^{b'_1x_1+b'_2x_2+b'_{12}x_1x_2},\]
 where $b'_{12}$ is odd. Then we are done by letting $h=gg'$.
 \begin{itemize}
 \item If there exists a rank-3 signature $f\in\overline{\mathcal{F}}^{\operatorname{even}}$ of the type (\ref{ddd}),
 then we are done by letting $h=f$.
 If there exists a rank-3 signature $f\in\overline{\mathcal{F}}^{\operatorname{even}}$ of the type (\ref{hhh}),
 then we are done by Lemma~\ref{h-hhh} since we have $f$ and $g$ in hand.
\item
 Assume that there is a signature $f\in\overline{\mathcal{F}}^{\operatorname{even}}$ that has the support (\ref{dhd}).
 If $f^2\notin\mathscr{A}$, then we are done by letting $h=f^2$. Otherwise, we can assume that
\[\underline{f}(x_1, x_2, x_3)=\alpha^{a_1x_1+a_2x_2+a_3x_3+2a_{12}x_1x_2+2a_{13}x_1x_3+2a_{23}x_2x_3+4a_{123}x_1x_2x_3}.\]
By pinning $x_1=0$, we get a rank-2 signature $g''$. After collations, $g''$
has the support \[(x_2)_{\frac{k}{2}}(x_3)_{\frac{k}{2}}(x_2+x_3)_{\frac{k}{2}}\] and
\[\underline{g''}(x_2, x_3)=\alpha^{a_2x_2+a_3x_3+2a_{23}x_2x_3}.\]
If $g''\notin\mathscr{A}$, then we are done by letting $h=gg'$.
Otherwise, we have $a_{2}\equiv a_3\equiv a_{23}\equiv 0\pmod 2$.
Similarly, by pinning $x_2=0, x_3=0$, we have $a_1\equiv a_{2}\equiv a_3\equiv a_{12}\equiv a_{13}\equiv a_{23}\equiv 0\pmod 2$.
This implies that
\[\underline{f}(x_1, x_2, x_3)=\frak i^{a_1x_1+a_2x_2+a_3x_3+2a_{12}x_1x_2+2a_{13}x_1x_3+2a_{23}x_2x_3+2a_{123}x_1x_2x_3}.\]
Let $\widehat{f}=\left[\begin{smallmatrix}
1 & 0\\
0 & \rho^d
\end{smallmatrix}\right]^{\otimes \frac{11}{2}k}f$, then
\[\underline{\widehat{f}}=\rho^{\frac{dk}{2}(6x_1+6x_2+6x_3+6x_1x_2+6x_1x_3+6x_2x_3+8x_1x_2x_3)}f(x_1, x_2, x_3),\]
i.e.,
\[\underline{\widehat{f}}=\frak i^{(3d+a_1)x_1+(3d+a_2)x_2+(3d+a_3)x_3+
  (3d+2a_{12})x_1x_2+(3d+2a_{13})x_1x_3+(3d+2a_{23})x_2x_3+(4d+2a_{123})x_1x_2x_3}.\]
 Since $f\notin\mathscr{A}_k^d$ for some even $d$, $a_{123}$ is odd.
 By connecting the $\frac{k}{2}$ copies of variables $y_1, y_2, y_1+y_2$ of $g$
 to the $\frac{k}{2}$ copies of variables of $x_1+x_2, x_1+x_3, x_2+x_3$ of $f$ using $(=_{2k})$,
 we get a $k$-type signature $f'$ and
 \[\underline{f'}=\frak i^{a'_1x_1+a'_2x_2+a'_3x_3+2a'_{12}x_1x_2+2a'_{13}x_1x_3+2a'_{23}x_2x_3+2a_{123}x_1x_2x_3}.\]
 Since $a_{123}$ is odd, we have $f'\notin\mathscr{A}$. Then we are done  by letting $h=f'$.
\item
 Assume that there is a signature $f\in\overline{\mathcal{F}}^{\operatorname{even}}$ which has the support (\ref{hdh}).
By the same statement as the case (\ref{dhd}), we can assume that
\[\underline{f}(x_1, x_2, x_3)=\frak i^{a_1x_1+a_2x_2+a_3x_3+2a_{12}x_1x_2+2a_{13}x_1x_3+2a_{23}x_2x_3+2a_{123}x_1x_2x_3}.\]
Let $\widehat{f}=\left[\begin{smallmatrix}
1 & 0\\
0 & \rho^d
\end{smallmatrix}\right]^{\otimes \frac{11}{2}k}f$, then
\[\underline{\widehat{f}}=\rho^{\frac{dk}{2}(6x_1+6x_2+6x_3+6x_1x_2+6x_1x_3+6x_2x_3+4x_1x_2x_3)}f(x_1, x_2, x_3),\]
i.e.,
\[\underline{\widehat{f}}=\frak i^{(3d+a_1)x_1+(3d+a_2)x_2+(3d+a_3)x_3+
  (3d+2a_{12})x_1x_2+(3d+2a_{13})x_1x_3+(3d+2a_{23})x_2x_3+(2d+2a_{123})x_1x_2x_3}.\]
 Since $f\notin\mathscr{A}_k^d$ for some even $k$, $a_{123}$ is odd.
 By connecting the $\frac{k}{2}$ copies of variables $y_1+, y_2, y_1+y_2$ of $g$
 to the $\frac{k}{2}$ copies of variables of $x_1+x_2, x_1+x_3, x_2+x_3$ of $f$ using $(=_{2k})$,
 we get a $\frac{k}{2}$-type signature $f''$ which has the support
 \[(x_1)_{\frac{k}{2}}(x_2)_{\frac{k}{2}}(x_3)_{\frac{k}{2}}(x_1+x_2)_{\frac{k}{2}}(x_1+x_3)_{\frac{k}{2}}(x_2+x_3)_{\frac{k}{2}}(x_1+x_2+x_3)_{\frac{k}{2}}\] after collations
 and
 \[\underline{f''}=\frak i^{a'_1x_1+a'_2x_2+a'_3x_3+2a'_{12}x_1x_2+2a'_{13}x_1x_3+2a'_{23}x_2x_3+2a_{123}x_1x_2x_3}.\]
 Note that $a_{123}$ is odd. Then we are done  by Lemma~\ref{h-hhh} since we have $g$ and $f''$ in hand..
\end{itemize}

 Now we can assume that there are no signatures that have the support (\ref{ddd}), (\ref{hhh}), (\ref{hdh}) or (\ref{dhd}) in
$\overline{\mathcal{F}}^{\operatorname{even}}$,
i.e., all the signatures in $\overline{\mathcal{F}}^{\operatorname{even}}$ have the supports
(\ref{qqq}), (\ref{3q3}), (\ref{q3q}) or (\ref{333}).
For $f\in\overline{\mathcal{F}}^{\operatorname{even}}$,
if $f^4\notin\mathscr{A}$, then we are done by letting $h=f^4.$
Otherwise, we have
\[\underline{f}(x_1, x_2, x_3)=\beta^{a_1x_1+a_2x_2+a_3x_3+2a_{12}x_1x_2+2a_{13}x_1x_3+2a_{23}x_1x_2x_3+4a_{123}x_1x_2x_3}.\]
Moreover, by pinning $x_1=0$, we get a rank-2 signature $u$ that has the support
\[(x_2)_{\frac{k}{2}}(x_3)_{\frac{k}{2}}(x_2+x_3)_{\frac{k}{2}}\]
and
\[\underline{u}=\beta^{a_2x_2+a_3x_3+2a_{23}x_2x_3}.\]
If one of $a_2, a_3, a_{23}$ is nonzero modulo 4, then we are done by letting $h=gg'$.
Otherwise, we have $a_2\equiv a_3\equiv a_{23}\equiv 0\pmod 4$.
Moreover, by pinning $x_2=0, x_3=0$, we have
$a_2\equiv a_2\equiv a_3\equiv a_{12}\equiv a_{13}\equiv a_{23}\equiv 0\pmod 4$, i.e.,
\[\underline{f}(x_1, x_2, x_3)=\frak i^{a_1x_1+a_2x_2+a_3x_3+2a_{12}x_1x_2+2a_{13}x_1x_3+2a_{23}x_1x_2x_3+a_{123}x_1x_2x_3}.\]
Note that
\begin{itemize}
\item if $a_{123}$ is odd, then $f\notin\mathscr{A}_k^d$ for any even $d\in[k]$;
\item if $a_{123}\equiv 0\pmod 4$, then $f\notin\mathscr{A}_k^d$ for $d\equiv 2\pmod 4$ and
$f\in\mathscr{A}_k^d$ for $d\equiv 0\pmod 4$;
\item if $a_{123}\equiv 2\pmod 4$, then $f\notin\mathscr{A}_k^d$ for $d\equiv 0\pmod 4$ and
$f\in\mathscr{A}_k^d$ for $d\equiv 2\pmod 4$.
\end{itemize}

Firstly, assume that $a_{123}$ is odd.
Take two copies of $f$ with free variable set $\{x_1, x_2, x_3\}$ and $\{y_1, y_2, y_3\}$.
By connecting the bundles $(x_1)$ and $(y_1)$, $(x_2)$ and $(y_2)$,
$(x_3)$ and $(y_3)$, $(x_1+x_2)$ and $(y_1+y_2)$, $(x_1+x_3)$ and $(y_1+y_3)$, $(x_2+x_3)$ and $(y_2+y_3)$, $(x_1+x_2+x_3)$ and $(y_1+y_2+y_3)$
by $(=_{2k})$ respectively, we get a $\frac{k}{2}$-type signature $f'''$ and
\[\underline{f'''}(x_1, x_2, x_3)=\frak i^{2a_1x_1+2a_2x_2+2a_3x_3+4a_{12}x_1x_2+4a_{13}x_1x_3+4a_{23}x_1x_2x_3+2a_{123}x_1x_2x_3}.\]
Note that $a_{123}$ is odd. Then by Lemma~\ref{h-hhh}, we are done since we have $g$ and $f'''$ in hand.

If $a_{123}\equiv 0\pmod 4$, since $\mathcal{F}\nsubseteq\mathscr{A}_k^d$ for any $k\in[k]$,
there exists another signature $f^*\in\overline{\mathcal{F}}^{\operatorname{even}}$
such that $f^*$ is $\frac{k}{4}$-type and
\[\underline{f^*}(x_1, x_2, x_3)=\frak i^{a^*_1x_1+a^*_2x_2+a^*_3x_3+2a^*_{12}x_1x_2+2a^*_{13}x_1x_3+2a^*_{23}x_1x_2x_3+a^*_{123}x_1x_2x_3},\]
where $a^*\equiv 2\pmod 4$. Take one copy of $f$ with free variables $\{x_1, x_2, x_3\}$ and
one copy of $f^*$ with free variables $\{y_1, y_2, y_3\}$.
By connecting the bundles $(x_1)$ and $(y_1)$, $(x_2)$ and $(y_2)$,
$(x_3)$ and $(y_3)$, $(x_1+x_2)$ and $(y_1+y_2)$, $(x_1+x_3)$ and $(y_1+y_3)$, $(x_2+x_3)$ and $(y_2+y_3)$, $(x_1+x_2+x_3)$ and $(y_1+y_2+y_3)$
by $(=_{2k})$ respectively, we get a $\frac{k}{2}$-type signature $f''''$ and
\[\underline{f''''}(x_1, x_2, x_3)=\frak i^{c_1x_1+c_2x_2+c_3x_3+2c_{12}x_1x_2+2c_{13}x_1x_3+2c_{23}x_1x_2x_3+c_{123}x_1x_2x_3},\]
where $c_i=a_i+a^*_i, c_{jk}=a_{jk}+a^*_{jk}$ for $1\leq i\leq 3, 1\leq j<k\leq 3$ and $c_{123}=a_{123}+a^*_{123}$.
Note that $c_{123}\equiv 2\pmod 4$. Then by Lemma~\ref{h-hhh}, we are done since we have $g$ and $f''''$ in hand.
\end{proof}

Now we can prove Theorem~\ref{cspk-dic-thr}.
\begin{proof}
We will prove the theorem by induction on $k$.
Note that the theorem has been proved for the cases $k=1, 2$ by Theorem~\ref{csp-dic}
and Theorem~\ref{even-csp-dic}.
In the following we assume that $k\geq 3.$

If $\mathcal{F}\subseteq\mathscr{P}$ or $\mathcal{F}\subseteq\mathscr{A}_k^d$ for some $k\in[k]$, the tractability is obvious.
Then we assume that $\mathcal{F}\nsubseteq\mathscr{P}$ and $\mathcal{F}\nsubseteq\mathscr{A}_k^d$ for any $d\in[k]$.

For any $d'\in[k]$, by Lemma~\ref{non-product} and $\mathcal{F}\nsubseteq{\mathscr{P}}$,  we can construct
$g\notin\mathscr{P}$ and there exist $g_1, g_2, \cdots, g_{k'}$ such that $g=g_1g_2\cdots g_{k'}$, where
the signatures $g_1, g_2, \cdots, g_{k'}$ can be constructed in \#CSP$^k(\neq_2, \mathcal{F})$, i.e.,
\begin{equation}\label{nproduct}
\#\operatorname{CSP}_{k}(\neq_2, g)\leq_T\#\operatorname{CSP}_{k}(\neq_2, \mathcal{F}).
\end{equation}

By $\mathcal{F}\nsubseteq\mathscr{A}_k^d$ and Lemma~\ref{non-affine-transformable}, we can construct general equality
$[1, 0, \cdots, 0, a]$ of arity $r<k$, or we can construct $h$ such that for some $k'|k$, $k'>1$
$h\notin\mathscr{A}_{\frac{k}{k'}}^{d'}$ for any $d'\in[\frac{k}{k'}]$, and
there exist $f_1, f_2, \cdots, f_{k'}$
such that $h=f_1, f_2, \cdots, f_{k'}$.

If we have $[1, 0, \cdots, 0, a]$ of arity $r<k$, then we are done by Lemma~\ref{general-equality-induction}.
Otherwise, we can construct $f_1, f_2, \cdots, f_{\frac{k}{k'}}$ in \#CSP$^k(\neq_2, \mathcal{F})$
such that $h=f_1, f_2, \cdots, f_{k'}$ , i.e.,
\begin{equation}\label{naffine}
\#\operatorname{CSP}_{k}(\neq_2, h, \mathcal{F})\leq_T\#\operatorname{CSP}_{k}(\neq_2, \mathcal{F}),
\end{equation}
where $f\notin\mathscr{A}_{\frac{k}{k'}}^{d'}$ for any $d'\in[\frac{k}{k'}]$.

By  (\ref{nproduct}), (\ref{naffine}) and Lemma~\ref{induction},
we have
\[\#\operatorname{CSP}_{\frac{k}{k'}}(\neq_2, g, h)\leq_T\#\operatorname{CSP}_{k}(\neq_2, \mathcal{F}),\]
where $g\notin\mathscr{P}$ and $h\notin\mathscr{A}_{\frac{k}{k'}}^{d'}$ for any $d'\in[\frac{k}{k'}]$.
Then by induction, $\#\operatorname{CSP}_{\frac{k}{k'}}(\neq_2, f, g)$ is \#P-hard. Thus $\#\operatorname{CSP}_{k}(\neq_2, \mathcal{F})$ is \#P-hard.
\end{proof}

\bibliography{main}{}

\begin{thebibliography}{10}

\bibitem{Backens-class}
Miriam Backens.
\newblock The inductive entanglement classification yields ten rather than
  eight classes of four-qubit entangled states.
\newblock {\em Physical Review A}, 97:022329, 2017.

\bibitem{Backens-holant-plus}
Miriam Backens.
\newblock A new holant dichotomy inspired by quantum computation.
\newblock In {\em Proceedings of the 44th International Colloquium on Automata,
  Languages, and Programming}, pages 16:1--16:14, 2017.

\bibitem{Backens-Holant-c}
Miriam Backens.
\newblock A complete dichotomy for complex-valued holant$^c$.
\newblock In {\em Proceedings of the 45th International Colloquium on Automata,
  Languages, and Programming}, pages 12:1--12:14, 2018.

\bibitem{bennett-tele}
Charles~H Bennett, Gilles Brassard, Claude Cr{\'e}peau, Richard Jozsa, Asher
  Peres, and William~K Wootters.
\newblock Teleporting an unknown quantum state via dual classical and
  einstein-podolsky-rosen channels.
\newblock {\em Physical review letters}, 70(13):1895, 1993.

\bibitem{bulatov2012csp}
Andrei Bulatov, Martin Dyer, Leslie~Ann Goldberg, Markus Jalsenius, Mark
  Jerrum, and David Richerby.
\newblock The complexity of weighted and unweighted \#csp.
\newblock {\em Journal of Computer and System Sciences}, 78(2):681--688, 2012.

\bibitem{bulatov2005gh}
Andrei Bulatov and Martin Grohe.
\newblock The complexity of partition functions.
\newblock {\em Theoretical Computer Science}, 348(2-3):148--186, 2005.

\bibitem{bulatov-ccsp}
Andrei~A Bulatov.
\newblock The complexity of the counting constraint satisfaction problem.
\newblock {\em Journal of the ACM (JACM)}, 60(5):1--41, 2013.

\bibitem{jcbook}
Jin-Yi Cai and Xi~Chen.
\newblock {\em Complexity Dichotomies for Counting Problems: Volume 1, Boolean
  Domain}.
\newblock Cambridge University Press, 2017.

\bibitem{cai-chen-csp}
Jin-Yi Cai and Xi~Chen.
\newblock Complexity of counting csp with complex weights.
\newblock {\em Journal of the ACM (JACM)}, 64(3):1--39, 2017.

\bibitem{cai-chen-lu-gh}
Jin-Yi Cai, Xi~Chen, and Pinyan Lu.
\newblock Graph homomorphisms with complex values: A dichotomy theorem.
\newblock {\em SIAM Journal on Computing}, 42(3):924--1029, 2013.

\bibitem{cai-chen-lu-nonnegative-csp}
Jin-Yi Cai, Xi~Chen, and Pinyan Lu.
\newblock Nonnegative weighted\# csp: An effective complexity dichotomy.
\newblock {\em SIAM Journal on Computing}, 45(6):2177--2198, 2016.

\bibitem{cai-fu-eight}
Jin-Yi Cai and Zhiguo Fu.
\newblock Complexity classification of the eight-vertex model.
\newblock {\em arXiv preprint arXiv:1702.07938}, 2017.

\bibitem{Cai-Fu-csp}
Jin-Yi Cai and Zhiguo Fu.
\newblock Holographic algorithm with matchgates is universal for planar\# csp
  over boolean domain.
\newblock {\em SIAM Journal on Computing}, (0):STOC17--50, 2019.

\bibitem{cai-fu-shao-eo}
Jin-Yi Cai, Zhiguo Fu, and Shuai Shao.
\newblock Complexity of counting weighted eulerian orientations with ars.
\newblock {\em arXiv preprint arXiv:1904.02362}, 2019.

\bibitem{cai-fu-xia-six}
Jin-Yi Cai, Zhiguo Fu, and Mingji Xia.
\newblock Complexity classification of the six-vertex model.
\newblock {\em Information and Computation}, 259:130--141, 2018.

\bibitem{cai-artem}
Jin-Yi Cai and Artem Govorov.
\newblock Perfect matchings, rank of connection tensors and graph
  homomorphisms.
\newblock In {\em Proceedings of the Thirtieth Annual ACM-SIAM Symposium on
  Discrete Algorithms}, pages 476--495. SIAM, 2019.

\bibitem{cai-guo-tyson-vanishing}
Jin-Yi Cai, Heng Guo, and Tyson Williams.
\newblock A complete dichotomy rises from the capture of vanishing signatures.
\newblock {\em SIAM Journal on Computing}, 45(5):1671--1728, 2016.

\bibitem{clfford-gate-1}
Jin-Yi Cai, Heng Guo, and Tyson Williams.
\newblock Clifford gates in the holant framework.
\newblock {\em Theoretical Computer Science}, 745:163--171, 2018.

\bibitem{cai-lu-xia-holant*}
Jin-Yi Cai, Pinyan Lu, and Mingji Xia.
\newblock Dichotomy for holant* problems of boolean domain.
\newblock In {\em Proceedings of the twenty-second annual ACM-SIAM symposium on
  Discrete Algorithms}, pages 1714--1728. SIAM, 2011.

\bibitem{clx-csp}
Jin-Yi Cai, Pinyan Lu, and Mingji Xia.
\newblock The complexity of complex weighted boolean \#csp.
\newblock {\em Journal of Computer and System Sciences}, 80(1):217--236, 2014.

\bibitem{CLX-HOLANTC}
Jin-Yi Cai, Pinyan Lu, and Mingji Xia.
\newblock Dichotomy for real holant$^c$ problems.
\newblock In {\em Proceedings of the Twenty-Ninth Annual ACM-SIAM Symposium on
  Discrete Algorithms}, pages 1802--1821. SIAM, 2018.

\bibitem{clfford-gate-2}
Jeroen Dehaene and Bart De~Moor.
\newblock Clifford group, stabilizer states, and linear and quadratic
  operations over gf (2).
\newblock {\em Physical Review A}, 68(4):042318, 2003.

\bibitem{slocc2}
Wolfgang D{\"u}r, Guifre Vidal, and J~Ignacio Cirac.
\newblock Three qubits can be entangled in two inequivalent ways.
\newblock {\em Physical Review A}, 62(6):062314, 2000.

\bibitem{dyer-green-gh}
Martin Dyer and Catherine Greenhill.
\newblock The complexity of counting graph homomorphisms.
\newblock {\em Random Structures \& Algorithms}, 17(3-4):260--289, 2000.

\bibitem{dyer-richerby}
Martin Dyer and David Richerby.
\newblock An effective dichotomy for the counting constraint satisfaction
  problem.
\newblock {\em SIAM Journal on Computing}, 42(3):1245--1274, 2013.

\bibitem{ekert}
Artur~K Ekert.
\newblock Quantum cryptography based on bell’s theorem.
\newblock {\em Physical review letters}, 67(6):661, 1991.

\bibitem{Freedman-et-al}
Michael Freedman, L\'aszl\'o Lov\'asz, and Alexander Schrijver.
\newblock Reflection positivity, rank connectivity, and homomorphism of graphs.
\newblock {\em Journal of the American Mathematical Society}, 20(1):37--51,
  2007.

\bibitem{2-system-2}
Mariami Gachechiladze and Otfried G{\"u}hne.
\newblock Completing the proof of “generic quantum nonlocality”.
\newblock {\em Physics Letters A}, 381(15):1281--1285, 2017.

\bibitem{4-qubit-g}
Masoud~Gharahi Ghahi and Stefano Mancini.
\newblock Comment on “inductive entanglement classification of four qubits
  under stochastic local operations and classical communication”.
\newblock {\em Physical Review A}, 98(6):066301, 2018.

\bibitem{goldberg2010gh}
Leslie~Ann Goldberg, Martin Grohe, Mark Jerrum, and Marc Thurley.
\newblock A complexity dichotomy for partition functions with mixed signs.
\newblock {\em SIAM Journal on Computing}, 39(7):3336--3402, 2010.

\bibitem{Daniel-Gottesman}
Daniel Gottesman.
\newblock The heisenberg representation of quantum computers, talk at.
\newblock In {\em International Conference on Group Theoretic Methods in
  Physics}. Citeseer, 1998.

\bibitem{surver-entangle}
Ryszard Horodecki, Pawe{\l} Horodecki, Micha{\l} Horodecki, and Karol
  Horodecki.
\newblock Quantum entanglement.
\newblock {\em Reviews of modern physics}, 81(2):865, 2009.

\bibitem{huang-lu-real}
Sangxia Huang and Pinyan Lu.
\newblock A dichotomy for real weighted holant problems.
\newblock In {\em 2012 IEEE 27th Conference on Computational Complexity}, pages
  96--106. IEEE, 2012.

\bibitem{jl}
Richard Jozsa and Noah Linden.
\newblock On the role of entanglement in quantum-computational speed-up.
\newblock {\em Proceedings of the Royal Society of London. Series A:
  Mathematical, Physical and Engineering Sciences}, 459(2036):2011--2032, 2003.

\bibitem{4-qubit-l}
Lucas Lamata, Juan Le{\'o}n, D~Salgado, and E~Solano.
\newblock Inductive entanglement classification of four qubits under stochastic
  local operations and classical communication.
\newblock {\em Physical Review A}, 75(2):022318, 2007.

\bibitem{3-qubit-l}
Lucas Lamata, Juan Le{\'o}n, D~Salgado, and Enrique Solano.
\newblock Inductive classification of multipartite entanglement under
  stochastic local operations and classical communication.
\newblock {\em Physical Review A}, 74(5):052336, 2006.

\bibitem{wang-lin}
Jiabao Lin and Hanpin Wang.
\newblock The complexity of boolean holant problems with nonnegative weights.
\newblock {\em SIAM Journal on Computing}, 47(3):798--828, 2018.

\bibitem{2-2-n}
Akimasa Miyake and Frank Verstraete.
\newblock Multipartite entanglement in 2$\times$ 2$\times$ n quantum systems.
\newblock {\em Physical Review A}, 69(1):012101, 2004.

\bibitem{michael-book}
Michael~A Nielsen and Isaac Chuang.
\newblock Quantum computation and quantum information, 2002.

\bibitem{2-system-1}
Sandu Popescu and Daniel Rohrlich.
\newblock Generic quantum nonlocality.
\newblock {\em Physics Letters A}, 166(5-6):293--297, 1992.

\bibitem{valiant-holo-02}
Leslie~G Valiant.
\newblock Quantum circuits that can be simulated classically in polynomial
  time.
\newblock {\em SIAM Journal on Computing}, 31(4):1229--1254, 2002.

\bibitem{valiant-holo-08}
Leslie~G Valiant.
\newblock Holographic algorithms.
\newblock {\em SIAM Journal on Computing}, 37(5):1565--1594, 2008.

\bibitem{4-qubit-v}
Frank Verstraete, Jeroen Dehaene, Bart De~Moor, and Henri Verschelde.
\newblock Four qubits can be entangled in nine different ways.
\newblock {\em Physical Review A}, 65(5):052112, 2002.

\end{thebibliography}


\end{document}